\documentclass[12pt,reqno]{gtpart}
\usepackage{graphicx,bbm,comment}
\usepackage{amsmath, amssymb, amsfonts, stmaryrd}
\usepackage[utf8]{inputenc}
\usepackage{epstopdf}
\usepackage{tikz-cd} 
\usepackage{setspace}
\usepackage{xcolor}
\usepackage{tcolorbox}
\usepackage[titletoc,title]{appendix}
\usepackage{enumitem}
\usepackage{ulem}

\usepackage{hyperref}
\definecolor{darkred}{rgb}{0.5,0.15,0.15}
\hypersetup{colorlinks=true,urlcolor=darkred,linkcolor=darkred,citecolor=darkred}

\allowdisplaybreaks

\renewcommand{\Z}{\mathbb{Z}}

\newcommand{\BR}{{\mathbb R}}
\newcommand{\ii}{\mathrm{i}}
\newcommand{\ra}{\to}
\newcommand{\rf}[1]{(\ref{#1})}

\newcommand{\be}{\beta}
\newcommand{\ga}{\gamma}
\newcommand{\Ga}{\Gamma}
\newcommand{\de}{\delta}

\newcommand\numberthis{\addtocounter{equation}{1}\tag{\theequation}}

\newcommand{\ep}{\epsilon}
\newcommand{\la}{\lambda}

\newcommand{\si}{\sigma}

\newcommand{\vf}{\varphi}

\newcommand{\pa}{\partial}

\newcommand{\nco}{\newcommand}

\nco{\on}{\operatorname}

\newcommand{\CC}{{\mathcal C}}

\newcommand{\CF}{{\mathcal F}}

\newcommand{\CM}{{\mathcal M}}

\newcommand{\CW}{{\mathcal W}}
\newcommand{\CX}{{\mathcal X}}
\newcommand{\CZ}{{\mathcal Z}}

\renewcommand{\sf}{{\mathsf f}}

\newcommand{\BC}{{\mathbb C}}

\newcommand{\BZ}{{\mathbb Z}}

\DeclareMathOperator*{\Res}{Res}

\renewcommand{\ell}{X} 

\newcommand{\e}{{\mathrm e}}

\newcommand{\I}{{\mathrm i}}

\newcommand{\ee}{\end{eqnarray}}
\newcommand{\bea}{\begin{eqnarray}}
\newcommand{\eea}{\end{eqnarray}}

\newcommand{\ben}{\begin{eqnarray}}
\newcommand{\een}{\end{eqnarray}}

\renewcommand{\hat}{\widehat}

\theoremstyle{plain}
\newtheorem{thm}{Theorem}[section]
\newtheorem{prop}[thm]{Proposition}
\newtheorem{lem}[thm]{Lemma}
\newtheorem{cor}[thm]{Corollary}

\theoremstyle{definition}
\newtheorem{dfn}[thm]{Definition}

\newtheorem*{question*}{Question}

\theoremstyle{remark}
\newtheorem{rem}[thm]{Remark}

\numberwithin{equation}{section}

\title{Mathematical structures of non-perturbative topological string theory:\\ \\
from GW to DT invariants}

\author{Murad Alim}
\author{Arpan Saha}
\author{J\"org Teschner}
\author{Iv\'an Tulli}
\address{M.A., A.S., I.T.: Fachbereich Mathematik, Universit\"at Hamburg, Bundesstr. 55, 20146, Hamburg; J.T.: Fachbereich Mathematik, Universit\"at Hamburg, Bundesstr. 55, 20146 and DESY Theory, Notkestr. 85, 22607 Hamburg}

\begin{document}

\begin{abstract}
We study the Borel summation of the Gromov--Witten potential for the resolved conifold.
The Stokes phenomena associated to this Borel summation
are shown to encode the Donaldson--Thomas invariants of the resolved conifold, having a direct relation  to the
Riemann--Hilbert problem formulated by T. {Bridgeland}.
There {exist} distinguished integration contours for which the 
Borel summation reproduces previous proposals for the {non-perturbative} topological
string partition functions of the resolved conifold. 
These partition functions are shown to have another asymptotic expansion 
at strong {topological string} coupling. We demonstrate that the Stokes phenomena of the 
strong-coupling expansion encode the DT invariants of the resolved
conifold in a second way. Mathematically, one finds a relation 
to Riemann--Hilbert problems associated to DT invariants which is different
from the one found at weak coupling. The Stokes phenomena of 
the strong-coupling expansion turn out to be closely related to 
the wall-crossing phenomena 
in the spectrum of BPS states on the resolved conifold
studied in the context of supergravity by D. Jafferis and G. Moore.

\end{abstract}
\maketitle
\setcounter{tocdepth}{2}
\tableofcontents

\section{Introduction}

The study of the geometric structures associated to quantum field and string theories has been extremely fruitful in revealing connections between different areas of mathematics as well as in putting forward organizing principles and relations for mathematical structures and invariants.\\

The focus of this work is on the connection of two types of invariants associated to families of Calabi--Yau (CY) threefolds. On the one hand, the Gromov--Witten invariants are characteristics of the enumerative geometry of maps into the CY.
Their generating function is closely related to the partition function of topological string theory. The latter is a formal power series which is asymptotic in the topological string coupling constant. On the other hand, the Donaldson--Thomas or BPS invariants associated to the same geometry can be defined using the enumeration of coherent sheaves supported on holomorphic submanifolds on the same CY subject to a stability condition. Physically, the latter correspond to BPS states, which are realized by D-branes supported on subspaces of the CY geometry. The generating functions of BPS invariants are expected to correspond to physical partition functions of black holes. In physical terms, the 
topological string theory is obtained from a perturbative formulation of the underlying string theory, while the BPS invariants represent data representing non-perturbative effects in string theory. Relations between the two very different types of data and mathematical invariants have long been expected both from the points of view of physics \cite{INOV,DVV,Ooguri:2004zv} as well as mathematics \cite{MNOP1,MNOP2}.\\

The link between GW and DT invariants is thus expected to be intimately related to the non-perturbative structure of topological string theory. Since the latter is defined by an asymptotic series in the topological string coupling, the most canonical path to its non-perturbative structure is to consider the theory of resurgence and Borel resummation; see \cite{Marinolecture} and references therein for an overview. This has indeed been applied to topological string theory in connection with Chern--Simons theory and matrix models in \cite{Pasquetti:2009jg} as well as for the resolved conifold in \cite{HO}. In particular, \cite{HO} used a generalization of the Borel resummation and produced via Borel resummation a partition function which matched the expectations of a proposal for the non-perturbative structure of topological string theory on non-compact CY manifolds put forward earlier in \cite{HMMO,GHM}. A non-perturbative definition of the topological string free
energy for general toric CY has been proposed in \cite{GHM} in terms of the spectral determinants of the 
finite difference operators obtained by quantising the mirror curves. 
In \cite{Couso-Santamaria:2014iia}, techniques of resurgence and transseries were applied to the study of topological string theory via the holomorphic anomaly equations of BCOV \cite{Bershadsky:1993cx}; see also \cite{CousoSantamaria:2014xml} and references therein. 
These techniques have been applied to the study of the proposal of \cite{GHM} in \cite{CMS}.
The link to BPS structures started to emerge more clearly recently \cite{Grassi:2019coc,Gu:2021ize} where connections between Stokes phenomena and BPS invariants have been investigated. See also \cite{kontsevich2020analyticity,Garoufalidis:2020xec} for works in related directions.\\

On the side of DT invariants and BPS structures, exciting insights are coming from the study of wall-crossing phenomena. The wall-crossing formulas of Kontsevich and Soibelman \cite{KS} as well as Joyce and Song \cite{JS} have led to a lot of progress on wall-crossing phenomena of BPS states. In \cite{GMN1,GMN2,GMN3}, Gaiotto, Moore and Neitzke (GMN) provided a physical interpretation of these developments as well as new geometric constructions of hyperk\"ahler manifolds having metrics determined by the BPS spectra; see e.~g.~\cite{Neitzkenotes}. More recent developments are concerned with the analytic and integrable structures behind wall-crossing phenomena. The emerging links indicate new connections between DT invariants and GW invariants, going substantially beyond the scope of the MNOP relation \cite{MNOP1,MNOP2}. Bridgeland \cite{BridgelandDT} formulated a Riemann--Hilbert associated to the Donaldson--Thomas invariants of a given derived category
and defined an associated potential called 
Tau-function in \cite{BridgelandDT}. In simple examples including the resolved conifold \cite{BridgelandCon}, it was shown that an asymptotic expansion of the Tau-function reproduces the full Gromov--Witten potential. In \cite{CLT20} it was proposed that the topological string partition functions for a certain class of local CY represent local sections of certain canonical 
holomorphic line bundles 
defined by the relevant solutions to the Riemann-Hilbert problems from \cite{BridgelandDT}. \\

In this paper, we will revisit the Borel summation of the 
resolved
conifold partition function from a new perspective. We will
show, on the one hand, that the Stokes jumps of the Borel 
summation of the expansion in powers of the topological 
string coupling have a close relation to the jumps defining 
the Riemann--Hilbert problem defined by Bridgeland using 
DT invariants as input data in \cite{BridgelandCon}. The Stokes jumps serve as 
certain types of potentials for the jumps of the Darboux
coordinates defining the Riemann--Hilbert problem in \cite{BridgelandCon}.\\

The Borel summations along different rays $\rho$
are found to have the following structure
\begin{equation}\label{Frho-decomp}
F_{\rho}(\la,t)=F_{\mathrm{GV}}(\la,t)+F_{\mathrm D}(\la,t;\rho),
\end{equation}
where $\la$ is the topological string coupling, and 
$t$ the complexified K\"ahler parameter. The contribution denoted 
$F_{\mathrm{GV}}(\la,t)$ is the canonical re-organisation known 
from the work of Gopakumar and Vafa of the formal series in 
powers of $\la$ as a  series in powers of $Q=e^{2\pi\ii t}$ which is 
convergent for $\mathrm{Im}(t)>0$. $F_{\mathrm{GV}}(\la,t)$ does not depend on the
ray $\rho$.  The second part, $F_{\mathrm D}(\la,t;\rho)$ strongly depends
on the choice of a ray $\rho$.  $F_{\mathrm D}(\la,t;\rho)$ can be 
represented as functions of the variables $Q'=e^{4\pi^2\ii t/\la}$
and $q'=e^{4\pi^2\ii /\la}$, suggesting an interpretation in terms 
of non-perturbative effects  associated to D-branes in type II string 
theory. It is known that there exist non-perturbative
effects in string theory represented by disk amplitudes associated to  
stable
D-branes. Closely related effects have recently been identified with
non-perturbative corrections to the metric on the 
hypemultiplet moduli space in type II string theory on CY three-folds
\cite{ASS}.
It seems natural to interpret 
the jumps of $F_{\rho}(\la,t)$ across Stokes rays
as the consequences of changes of the set of stable
objects contributing to the non-perturbative effects 
in the partition functions.
The explicit results for the jumps take a  
particularly simple form, having a direct
relation to the Riemann-Hilbert problems associated to DT-invariants
in \cite{BridgelandCon} further discussed below.\\

The results associated to different rays $\rho$
interpolate between two special functions which had previously 
been proposed as candidates for non-perturbative definitions of the topological string 
partition functions: Integration along the imaginary axis yields
the Gopakumar--Vafa resummation $F_{\mathrm{GV}}(\la,t)$ on the one hand, while choosing $\rho$ to be
the positive real axis, $\rho=\BR_{>0}$, yields
a function  closely related to the triple sine function. 
In the case $\rho=\BR_{>0}$, we find that 
the function $F_{\mathrm D}(\la,t;\rho)$ 
appearing in
equation \rf{Frho-decomp} specialises to 
the previously known function $F_{\mathrm{NS}}^{}(\la,t)$ which 
can be obtained from the refined version of $F_{\mathrm{GV}}^{}$ introduced in 
\cite{IKV} in the limits studied by Nekrasov and Shatashvili \cite{NS}.
The combination 
\begin{equation}
F_{\rm np}^{}(\la,t):=F_{\mathrm{GV}}^{}(\la,t)+F_{\mathrm{NS}}^{}(\la,t),
\end{equation}
appearing on the right side of \rf{Frho-decomp} in the case $\rho=\BR_{>0}$ 
has been studied before as a candidate for a non-perturbative 
completion of the topological string partition function \cite{HMMO}. 
Relations to previous work studying the function $F_{\rm np}^{}(\la,t)$
in connection to topological string theory are further discussed in Section 
\ref{prevres}.\\

It turns out that there is an appealing way to encode the Stokes data geometrically,
in line with the previous suggestions made in \cite{CLT20}. It will be shown that
the Stokes jumps
can be interpreted as transition functions of a certain line bundle 
canonically associated to the solution of the Riemann--Hilbert problem 
considered by Bridgeland. We will show that this line bundle 
is closely related to the hyperholomorphic  line bundles studied in 
relation to hyperk\"ahler geometry in \cite{APP,Neitzke_hyperhol}.
The Borel summations of the topological string partition functions represent
local sections of this line bundle. 
In the previous work \cite{Coman:2018uwk,CLT20}, it had been demonstrated that the Fourier transforms
of the topological string partition functions associated to a certain class of 
local CY are related to the isomonodromic tau-functions which represent 
local sections of this line bundle. Due to the absence of compact four-cycles, 
the tau-functions  simply coincide with the topological string 
partition functions for the case at hand.\\

The Borel summation along the positive real axis 
appears to be distinguished in some ways.
This function also 
has an asymptotic expansion for $\la\ra\infty$,
referred to as the strong-coupling expansion in the following. The Borel summations 
of the strong-coupling expansion along different rays $\rho'$ 
are found to have the following structure:
\begin{equation}
{F}_{\rho'}'(\la,t)=F_{\rm BPS}(\la,t;\rho')+F_{\rm NS}(\la,t).
\end{equation}
The contribution $F_{\rm NS}(\la,t)$ is now independent of $\rho'$, while
$F_{\rm BPS}(\la,t;\rho')$ exhibits jumps when $\rho'$ crosses certain rays 
in the complex plane of the variable $1/\lambda$. We find that 
$Z_{\rm BPS}(\la,t;\rho'):=e^{F_{\rm BPS}(\la,t;\rho')}$ is closely related to  the counting functions for BPS states 
previously studied in the context of supergravity by Jafferis and Moore \cite{JM}. 
The Stokes jumps of $Z_{\mathrm{ BPS}}(\la,t;\rho')$ display a precise correspondence
to the wall-crossing behaviour of the counting functions for BPS states
studied in \cite{JM}. 
In the case of $\rho'=\BR_{>0}$ we recover $F_{\mathrm{GV}}(\la,t)$.\\

Mathematically one may again observe a close relation to a Riemann--Hilbert 
problem associated to DT theory. However, the jumps 
of $Z_{\mathrm{BPS}}(\la,t;\rho')$ now directly coincide with the jumps of a particular
coordinate function in a close relative of Bridgeland's Riemann--Hilbert problem, 
as could have been expected from previous 
computations of  $Z_{\mathrm{BPS}}(\la,t;\rho')$ 
on the basis of wall-crossing formulae \cite[Appendix A]{Banerjee:2019apt}. 
It should be stressed that both the 
location of jumps, and the functional form of the jumps are different for 
weak- and strong-coupling expansions. However, we find that both are determined 
by Riemann--Hilbert type problems associated to DT invariants, 
albeit quite 
remarkably in somewhat different ways. \\

At least in the example studied in this paper,
we have identified two new ways to extract non-perturbative 
information on DT invariants from the GW invariants
defining the topological string partition functions. Our results
suggest that these data are deeply encoded in the analytic 
structures of non-perturbatively defined partition functions. 
The way this happens indicates close connections 
to string-theoretic 
S-duality conjectures, as will be briefly discussed 
in Section \ref{sec:S-dual}. \\

\textbf{Acknowledgments:} we have benefited from discussions with Vicente Cort\'es, Timo Weigand, and Alexander Westphal around common research projects within the Cluster of Excellence ``Quantum Universe". The authors would furthermore like to thank Sergei Alexandrov, Tom Bridgeland, Marcos Mariño, Greg Moore, and Boris Pioline for comments on a preliminary version of this paper. The work of J.T. and I.T. is funded by the Deutsche Forschungsgemeinschaft (DFG, German Research Foundation) under Germany's Excellence Strategy EXC 2121 Quantum Universe 390833306. The work of M.A and A.S. is supported through the DFG Emmy Noether grant AL 1407/2-1.

\section{Borel summations of the resolved conifold partition function}\label{mainresults}

We are going to study the  formal series
 \[
  \widetilde{F}(\lambda,t)= \frac{1}{\lambda^2} \mathrm{Li}_{3}(Q)+\frac{B_2}{2}\mathrm{Li}_1(Q) + \sum_{g=2}^{\infty} \lambda^{2g-2} \frac{(-1)^{g-1}B_{2g}}{2g (2g-2)!} \mathrm{Li}_{3-2g} (Q),
  \]
with $Q=\exp(2\pi \I t)$, and polylogarithms $\mathrm{Li}_s(z)$ and Bernoulli numbers $B_n$  defined by
\begin{equation}\label{polylogbern}
\mathrm{Li}_s(z) = \sum_{n=1}^{\infty} \frac{z^n}{n^s}\, ,\quad s\in \mathbb{C}\,,\qquad
\frac{w}{e^w-1} = \sum_{n=0}^{\infty} B_n \frac{w^n}{n!}\,.
\end{equation}
Borel summation of this formal series will repackage the information 
contained in it in an interesting way, revealing non-obvious  
mathematical structures.    Our goals in this section will be to state the results on the Stokes phenomena of the Borel sums of $\widetilde{F}(\lambda,t)$,
to discuss
some of its implications, and relations to previous results in the literature. 

\subsection{Motivation: Topological string theory on the resolved conifold}

Topological string theory motivates the consideration of the topological string partition functions. One expects to be able to 
associate such partition functions 
to families of Calabi--Yau (CY) threefolds $X=X_t$, with 
$t=(t^1,\dots,t^n)$ being a set of distinguished local coordinates on the  
CY K\"ahler moduli space $\mathcal{M}$ of dimension $n=h^{1,1}(X_t)$. 
The partition function is expected to be defined by 
an asymptotic series in the topological string coupling $\lambda$ of the form
\begin{equation}
Z_{\rm top} (\lambda,t)= \exp \left(\sum_{g=0}^{\infty} \lambda^{2g-2} {F}^{g}(t)\right).
\end{equation}
In order to provide a rigorous mathematical basis for the definition of topological 
string partition functions, one may start by defining
the GW potential of a Calabi--Yau threefold $X$ as the formal power series
\begin{equation}\label{F-formal}
\CF(Q,\lambda) = \sum_{g\ge 0}  \lambda^{2g-2} \CF^g(Q)=\sum_{g\ge 0} 
\sum_{\beta\in H_2(X,\mathbb{Z})} \lambda^{2g-2}
  N_\be^g \,Q^{\beta}\, ,
\end{equation}
where $Q^{\beta}=\prod_{r=1}^nQ_r^{\be_r}$ if $\be=\sum_{r=1}^n\be_r\ga_r$,
with $\{\ga_1,\dots,\ga_n\}$ being an integral  basis for $H_2(X,\BZ)$, 
and $Q_r$ being formal variables for $r=1,\dots,n$.
One may note that the term associated to $\be=0$ is independent of the K\"ahler class
$\be$, motivating the decomposition
\begin{equation}\label{freeenergydecomp}
\CF(Q,\la)=\CF_{0}(\la) + \widetilde{\CF}(Q,\la)\,, 
\end{equation}
where the contribution $\CF_{0}(\la)$ takes the universal form \cite{FP}
\begin{equation}\label{constmapterms}
\CF_{0}(\la)=\sum_{g\ge 0}  \lambda^{2g-2}F_{0}^g,\qquad
F_{0}^g = \frac{\chi(X)(-1)^{g-1}\, B_{2g}\, B_{2g-2}}{4g (2g-2)\, (2g-2)!}\,, \quad g\ge2\,,
\end{equation}
with $\chi(X)$ being the Euler characteristic of $X$.
The formal series $\widetilde{\CF}(Q,\la)$ is defined as
\begin{equation}
\widetilde{\CF}(Q,\la)=\sum_{g\ge 0}  
\sum_{\beta\in \Ga} \lambda^{2g-2} [\mathrm{GW}]_{\be,g}^{} \,Q^{\beta},
\end{equation}
where $\Ga=\{\be\in H_2(X,\BZ);\be\neq 0\}$, 
with $[\mathrm{GW}]_{\be,g}^{}$ being the Gromov--Witten invariants.
In this way, one  arrives at a precise definition of $\CF(Q,\la)$ as a formal series.\\
 
There is a class of CY manifolds where the series  $\CF^g(Q)$ actually 
have finite radii of convergence, allowing us to define 
the functions $F^g(t)=\CF^g(e^{2\pi\I t_1},...,e^{2\pi\I t_n})$,
where $t=(t_1,...,t_n)$.
The resulting power series in $\la$ is not expected to be convergent, in general. 
One may hope, however, that there can exist analytic functions 
having the series $F(\la,t)=\sum_{g\geq 0}\la^{2g-2}F^g(t)$ as asymptotic expansion.\\

We are here considering a particular example of a CY manifold $X$ called the 
resolved conifold. This
CY threefold represents the total space of the rank two bundle over the projective line:
\begin{equation}
X := \mathcal{O}(-1) \oplus \mathcal{O}(-1) \rightarrow \mathbb{P}^1\,,
\end{equation}
and corresponds to the resolution of the conifold singularity.\\ 

The GW potential for this geometry was determined in physics \cite{Gopakumar:1998ii,GV}, and in mathematics \cite{Faber} with the following outcome for the non-constant maps:\footnote{See also \cite{MM} for the determination of $F^g$ from a string theory duality and the explicit appearance of the polylogarithm expressions.}
\begin{equation}\label{resconfree}
\widetilde{F}(\lambda,t)= \sum_{g=0}^{\infty} \lambda^{2g-2} \widetilde{F}^g(t)= \frac{1}{\lambda^2} \mathrm{Li}_{3}(Q)+ \sum_{g=1}^{\infty} \lambda^{2g-2} \frac{(-1)^{g-1}B_{2g}}{2g (2g-2)!} \mathrm{Li}_{3-2g} (Q) \, ,
\end{equation}
using the notation $Q=e^{2\pi\ii t}$. The constant map constribution has the form (\ref{constmapterms}) with $\chi(X)=2$ and $F_0^0=-\zeta(3)$. The value of $F_0^1$ only shifts $\mathcal{F}(Q,\lambda)$ by an overall constant, and its specific value won't be important. Our first goal will be to study the Borel summability of the series \rf{resconfree} and (\ref{constmapterms}).
This was first studied in 
\cite{Pasquetti:2009jg}. The results presented below
complete and clarify previous work on this subject, as will be discussed in more
detail below. 

\subsection{Statement of results for the Borel sum and its Stokes phenomena}\label{resultssect}

Here we state a theorem collecting the results that we wish to prove. The proof of each part will be presented in Section \ref{proofs}.\\

Before stating the theorem, we briefly recall the definition of Borel summation. Given a formal power series $a(\check{\lambda}) \in \check{\lambda}\mathbb{C}[[\check{\lambda}]]$, we consider its Borel transform $\mathcal{B}(a)(\xi)$, where

\begin{equation}
    \mathcal{B}\colon\check{\lambda}\mathbb{C}[[\check{\lambda}]]\to \mathbb{C}[[\xi]], \;\;\;\; \mathcal{B}(\check{\lambda}^{n+1})=\frac{\xi^n}{n!}.
\end{equation}
Let $\check{\lambda}\in \mathbb{C}^{\times}$ and let $\rho$ be a ray from $0$ to $\infty$ in the complex $\xi$-plane. If $\mathcal{B}(a)(\xi)$ defines an analytic function along $\rho$, we define the Borel sum of $a(\check{\lambda})$ at $\check{\lambda}$, along $\rho$ by 
\begin{equation}\label{defborelsum}
    \int_{\rho}\mathrm d\xi \;e^{-\xi/\check{\lambda}}\mathcal{B}(a)(\xi)\,.
\end{equation}
If (\ref{defborelsum}) is finite, we say $a(\check{\lambda})$ is Borel summable at $\check{\lambda}$, along $\rho$.

\begin{thm} \label{theorem1} Consider the formal series
\begin{align*}\label{formal}
\widetilde{F}(\lambda,t)&=  \frac{1}{\lambda^2} \mathrm{Li}_{3}(Q)+\frac{B_2}{2}\mathrm{Li}_1(Q) + \sum_{g=2}^{\infty} \lambda^{2g-2} \frac{(-1)^{g-1}B_{2g}}{2g (2g-2)!}\, \mathrm{Li}_{3-2g} (Q) \,  \\
&=\frac{1}{\lambda^2} \mathrm{Li}_{3}(Q)  +\frac{B_2}{2}\mathrm{Li}_1(Q)+ \Phi(\check{\lambda},t)\,, \quad \check{\lambda}=\frac{\lambda}{2\pi}\,, \quad Q=e^{2\pi \I t}\,. \numberthis
\end{align*} 
Then we have the following:
\begin{itemize}
    \item[(i)] (Borel transform)  For $t\in \mathbb{C}^{\times}$ with $|\mathrm{Re}(t)|<1/2$, let $G(\xi,t):=\mathcal{B}(\Phi(-,t))(\xi)$ denote the Borel transform of $\Phi(\check \lambda,t)$. Then $G(\xi,t)$ converges for $|\xi|<2\pi |t|$.  Furthermore, $G(\xi,t)$ admits a series representation of the form
    \begin{equation}\label{Borel-sum}
        G(\xi,t) = \frac{1}{(2\pi )^2}\sum_{m\in\Z \setminus \{0\}}\frac{1}{m^3}
        \frac{1}{2\xi}\frac{\partial}{\partial \xi}
        \bigg(\frac{\xi^2}{1-e^{-2\pi \I t + \xi/m}}-\frac{\xi^2}{1-e^{-2\pi \I t - \xi/m}}\bigg).
\end{equation}
We can use the above series representation to analytically continue $G(\xi,t)$ in the $\xi$ variable to a meromorphic function with poles at $\xi=2\pi \I(t+k)m$ for $k \in \mathbb{Z}$ and $m \in \mathbb{Z}\setminus\{0\}$.
\item[(ii)] (Borel sum) For $k\in \mathbb{Z}$ let $l_k:=\mathbb{R}_{<0}\cdot 2\pi \I(t+k)$ and $l_{\infty}:=\I\mathbb{R}_{<0}$. Given any ray $\rho$ from $0$ to $\infty$ different from $\{\pm l_{k}\}_{k\in \mathbb{Z}}\cup \{\pm l_{\infty}\}$, and $\lambda$ in the half-plane $\mathbb{H}_{\rho}$ centered at $\rho$, we define the Borel sum of $\widetilde{F}(\lambda,t)$ along $\rho$ as
\begin{equation}
    F_{\rho}(\lambda,t):=\frac{1}{\lambda^2} \mathrm{Li}_{3}(Q)+ \frac{B_2}{2}\mathrm{Li}_{1}(Q)+\int_{\rho}\mathrm d\xi\, e^{-\xi/\check{\lambda}} G(\xi,t)\,.
\end{equation}
Taking $\rho=\mathbb{R}_{>0}$, and assuming that $\mathrm{Im}(t)>0$ and $0<\mathrm{Re}(t)<1$, we have the following identity whenever $\mathrm{Re}(t)<\mathrm{Re}(\check{\lambda}+1)$:

\begin{equation}\label{FR+-def}
    F_{\mathbb{R}_{>0}}(\lambda,t)= - \int_{\mathbb{R}+\I 0^+} \frac{\mathrm du}{8u}\frac{e^{u(t-1/2)}}{\sinh(u/2)(\sinh(\check\lambda u/2))^2}\,.
\end{equation}
\item[(iii)] (Stokes jumps) Let $\rho_k$ be a ray in the sector determined by the Stokes rays $l_{k}$ and $l_{k-1}$. Then if $\mathrm{Im}(t)>0$, on the overlap of their domains of definition in the $\lambda$ variable we have
\begin{equation}
    \phi_{\pm l_k}(\lambda,t):=F_{\pm \rho_{k+1}}(\lambda,t)
    -F_{\pm \rho_k}(\lambda,t) =\frac{1}{2\pi \I}\partial_{\check\lambda}\Big(\check{\lambda}\,\mathrm{Li}_2\big(e^{\pm 2\pi \I(t+k)/\check \lambda}\big)\Big)\;.
\end{equation}
If $\mathrm{Im}(t)<0$, then the previous jumps also hold provided $\rho_{k+1}$ is interchanged with $\rho_{k}$ in the above formula. \\

\item[(iv)] (Limits to $\pm \I\mathbb{R}_{>0}$)  Let $\rho_k$ denote any ray between the rays $l_{k}$ and $l_{k-1}$. Furthermore, assume that $0<\mathrm{Re}(t)<1$, $\mathrm{Im}(t)>0$, $\mathrm{Re}(\lambda)>0$, $\mathrm{Im}(\lambda)<0$, and $\mathrm{Re}\, t < \mathrm{Re} (\check{\lambda}+1)$. Then
\begin{equation}
    \lim_{k\to \infty}F_{\rho_k}(\lambda,t)=\lim_{k\to \infty}F_{-\rho_k}(-\lambda,t)=\sum_{k=1}^\infty\frac{e^{2\pi \I k t}}{k\big(2\sin\big(\frac{\lambda k}{2}\big)\big)^2}\,.
\end{equation}
Furthemore, we can write the sum of the Stokes jumps along $l_k$ for $k\geq 0$ as
\begin{equation}
    \sum_{k=0}^{\infty}\phi_{l_k}(\lambda,t)=\frac{1}{2\pi \I }\partial_{\lambda}\Big(\lambda \sum_{l=1}^{\infty}\frac{w^l}{l^2(1-\widetilde{q}^l)}\Big), \;\;\;\; w:=e^{2\pi \I t/\check\lambda},\;\; \widetilde{q}:=e^{2\pi \I/\check{\lambda}}\;.
\end{equation}
If, on the other hand, we take $0<\mathrm{Re}(t)<1$, $\mathrm{Im}(t)>0$, $\mathrm{Re}(\lambda)>0$, $\mathrm{Im}(\lambda)>0$, $\mathrm{Re}(t)<\mathrm{Re}(\check\lambda +1)$ and furthermore assume that $|e^{2\pi \I t/\check{\lambda}}|<1$, then we also have
\begin{equation}
    \lim_{k\to -\infty}F_{\rho_k}(\lambda,t)=\lim_{k\to -\infty}F_{-\rho_k}(-\lambda,t)=\sum_{k=1}^\infty\frac{e^{2\pi \I k t}}{k\big(2\sin\big(\frac{\lambda k}{2}\big)\big)^2}.
\end{equation}
    
\end{itemize}
\end{thm}
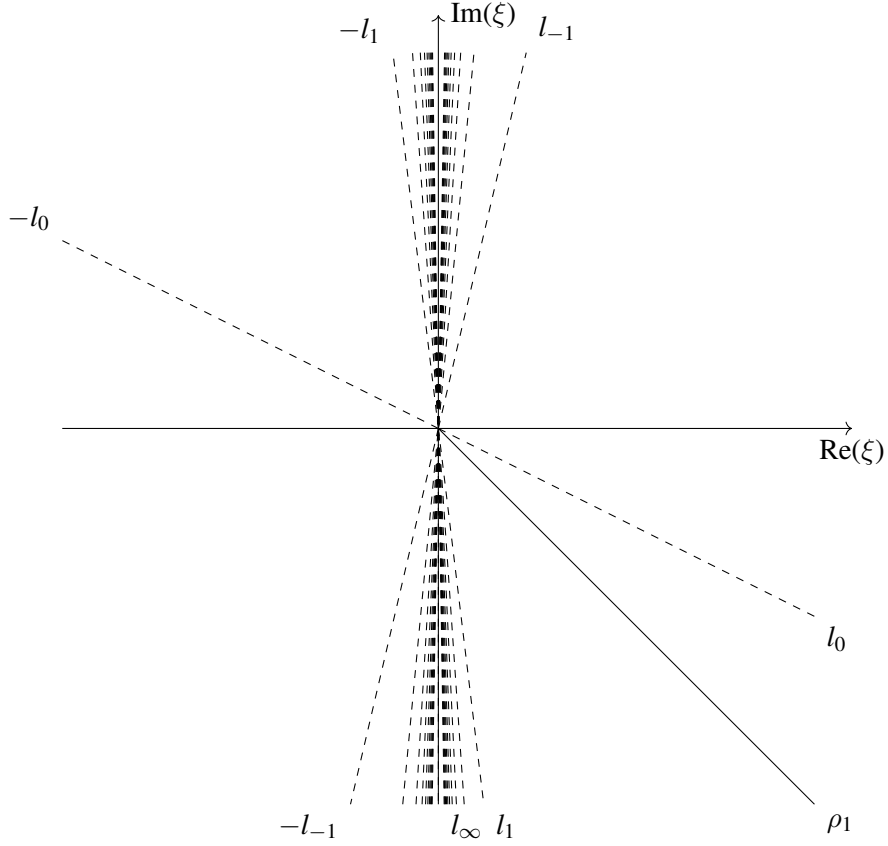
\begin{figure}
\begin{center}
\begin{tikzpicture}
  \draw[->] (-5,0) -- (5.5,0) coordinate[label = {below:$\textrm{Re} (\xi)$}] (xmax);
  \draw[->] (0,-5) -- (0,5.5) coordinate[label = {right:$\textrm{Im} (\xi)$}] (ymax) ;
 \draw[dashed] (0,0)--(5,-2.5) node[anchor=north west] {$l_0$};
\draw[dashed] (0,0)--(-5,2.5) node[anchor=south east] {$-l_0$};
\draw (0,0)--(5,-5) node[anchor=north west] {$\rho_1$};
  \draw[dashed] (0,0)--({5/(-2-6.28)},5) node[anchor=south east] {$-l_1$};
  \draw[dashed] (0,0)--({-5/(-2-6.28)},-5) node[anchor=north west] {$l_1$}; 
 \draw[dashed] (0,0)--({5/(-2-6.28*(-1))},5) node[anchor=south west] {$l_{-1}$};
  \draw[dashed] (0,0)--({-5/(-2-6.28*(-1)},-5) node[anchor=north east] {$-l_{-1}$};
 \draw[dashed] (0,0)--(0,-5) node[anchor=north west] {$l_{\infty}$}; 
 \foreach \k in {-10,...,-2} {
  \draw[dashed] ({-5/(-2-6.28*\k)},-5)--({5/(-2-6.28*\k)},5);
  };
 \foreach \k in {2,...,10} {
  \draw[dashed] ({-5/(-2-6.28*\k)},-5)--({5/(-2-6.28*\k)},5);
  };
\end{tikzpicture}
\end{center}

\caption{Illustration of the Stokes rays $l_k=\mathbb{R}_{<0}\cdot 2\pi \I (t+k)$ in the Borel plane, plotted for $t= \frac{1}{\pi}\left(1+ \frac{\I}{2}\right)$ and $k=-10,\dots,10$, as well as a possible integration ray $\rho_1$.}
\end{figure}
Let us note that under the assumptions of the first part of (iv),  $\lim_{k\to \infty}F_{\rho_k}(\lambda,t)$ differs from $F_{\mathbb{R}_{>0}}(\la,t)$ by the sum over all jumps 
$\phi_{l_k}(\lambda,t)$ for $k\geq 0$, leading to the decomposition
\begin{equation}\label{Fnp-decomp}
F_{\BR_{>0}}(\la,t)=\lim_{k\to \infty}F_{\rho_k}(\lambda,t)-\frac{1}{2\pi\ii}\frac{\pa}{\pa \la}
\Big(\la\sum_{l=1}^\infty \frac{w^l}{l^2(1-\widetilde{q}^{l})}\Big).
\end{equation}
As we will see in Proposition \ref{rblimit1}, this decomposition can be obtained by evaluating the integral 
on the right of \rf{FR+-def} as a sum over residues. Part (iv)
of the theorem above identifies the second term on the right of 
\rf{Fnp-decomp} with the sum over the Stokes jumps in the lower right quadrant
of the Borel plane.\\

It will turn out that the Borel summation $F_{0,\rho}(\la)$
of the formal series \rf{constmapterms} is closely related to the value of the function 
$F_{\rho}(\la,0)$ defined in Theorem \ref{theorem1}. The relation will be found to be of the form 
\begin{equation}\label{Borelconstmap}
    F_{0,\rho}(\la) = -F_{\rho}(\la,0)-\frac{1}{12}\log\check{\la}+C,
\end{equation} 
where $C$ is a constant independent of $\rho$ which won't be of interest for us. 
Equation \rf{Borelconstmap}
finally allows us to represent the Borel summation $\hat{F}_{\rho}(\lambda,t)$ of the 
full free energy $F(\la,t)=F_0(\lambda)+\widetilde{F}(\lambda,t)$ of the topological string theory by the formula
\begin{equation}\label{Borel-full}
\hat{F}_{\rho}(\la,t)=F_{\rho}(\la,t)-F_{\rho}(\la,0)-\frac{1}{12}\log\check{\la}+C.
\end{equation}

In the following two subsections we will first discuss the interpretation of 
Theorem \ref{theorem1} in the context of topological string theory. 
This will be followed by a discussion of the relation to previous results in 
this direction. 
\subsection{Connection to topological string theory}\label{connST}

In the case of the resolved conifold, non-perturbative definitions of the topological string partition 
functions should be analytic functions of $\la$ and $t$ such that 
(\ref{resconfree}) gives an asymptotic series expansion for $\la\ra 0$ of the corresponding free energy $\widetilde{F}(\lambda,t)$. 

\subsubsection{}

The Gopakumar--Vafa (GV) resummation of the GW potential \cite{GVb} re-organises the non-constant part 
$\widetilde{F}(\lambda,t)$ of the GW potential in the following form:
\begin{equation}\label{GVresum}
\sum_{g\ge 0}  \lambda^{2g-2}
\sum_{\beta\in \Ga}  [\mathrm{GW}]^{}_{\be,g} \,Q^{\beta}
= \sum_{\beta\in\Ga}\sum_{g\ge 0} \,[\mathrm{GV}]^{}_{\be,g}\, \sum_{k\ge 1} \frac{1}{k} \left( 2 \sin \big( {\textstyle \frac{k\lambda}{2}}\big)\right)^{2g-2} Q^{k\beta}\,.
\end{equation}
Equation \rf{GVresum} can be understood as an equality of 
formal power series in $Q^{\beta}$ with coefficients being Laurent series in $\la$.
One can 
thereby regard \rf{GVresum} as a definition of the
GV invariants  $[\mathrm{GV}]^{}_{\be,g}$ in terms of the Gromov--Witten invariants
$[\mathrm{GW}]^{}_{\be,g}$. \\

Using the known results for the invariants $\mathrm{GW}_{\be,g}$ of the conifold, 
one finds that the right-hand side of 
 \rf{GVresum} simplifies  to 
\begin{equation}\label{FGV-def}
F_{\mathrm{GV}}(\lambda,t)
=\sum_{k=1}^\infty\frac{e^{2\pi \I k t}}{k\big(2\sin\big(\frac{\lambda k}{2}\big)\big)^2}\,.
\end{equation}
This has also been derived using the topological vertex formalism \cite{AKMV}.
Assuming $\mathrm{Im}(t)>0$, one may notice that the series defining $F_{\mathrm{GV}}(\lambda,t)$ in 
\rf{FGV-def} is convergent for $\mathrm{Im}(\la)>0$ or $\mathrm{Im}(\la)<0$.
One may regard $F_{\mathrm{GV}}(\lambda,t)$  as a minimal summation of the divergent
series \rf{formal}, in the sense that it is obtained  by a 
rearrangement of the formal series $\widetilde{F}(\lambda,t)$ 
 into a convergent series in powers of $Q=e^{2\pi\ii t}$
that defines functions analytic in $\la$  away from the real line $\BR$.\\

Our results relate $F_{\mathrm{GV}}(\lambda,t)$ to the limits of the Borel summations 
along rays $\rho_k$ for $k\ra\pm\infty$ when the rays $\rho_k$ approach the 
imaginary axis. 

\subsubsection{}

As mentioned above, the function $F_{\mathrm{GV}}(\lambda,t)$ is not well-defined
for $\la\in\BR$.  This is one of the motivations to look for analytic
functions having the same asymptotic expansion, but larger domains of definition, 
as candidates for non-perturbative definitions of the topological string partition
functions.\\

A general proposal has been  made in \cite{HMMO} 
for non-perturbative definitions of topological string partition 
functions. This proposal was motivated by the observation \cite{HMO}
that one can systematically add functions of $e^{(2\pi)^2\frac{\ii}{\la}}$ 
to the function $F_{\mathrm{GV}}(\la,t)$ cancelling all  the singularities that 
$F_{\mathrm{GV}}(\la,t)$ develops on the real $\la$-axis. The function  of 
$e^{(2\pi)^2\frac{\ii}{\la}}$ having this property can be interpreted as 
certain non-perturbative corrections in string theory.\\

Specialised to the conifold, the proposal made in 
\cite{HMMO} yields the following function:
\begin{equation}\label{polecancel}
F_{\mathrm{np}}(\lambda,t):=F_{\mathrm{GV}}(\lambda,t)+\frac{1}{2\pi\I}\frac{\pa}{\pa \la}\la F_{\mathrm{NS}}\big({\textstyle \frac{4\pi^2}{\la},
\frac{2\pi}{\la}(t-\frac{1}{2})}\big),
\end{equation}
using the notations
\begin{equation}
F_{\mathrm{GV}}(\lambda,t):=\sum_{k=1}^\infty\frac{e^{2\pi \I k t}}{k\big(2\sin\big({\textstyle\frac{\lambda k}{2}}\big)\big)^2}\,,\qquad
F_{\mathrm{NS}}(g,t):=\frac{1}{2\I}\sum_{k=1}^\infty\frac{e^{2\pi \I k t}}{k^2\sin\big({\textstyle\frac{g k}{2}}\big)}\,.
\end{equation}
It is easy to see that the right side of \rf{polecancel} 
coincides with the expression on the right of \rf{Fnp-decomp} (i.e. that $F_{\mathbb{R}_{>0}}=F_{\mathrm{np}}$).

\subsubsection{} 
Using Borel 
summation is another natural approach to finding 
non-perturbative definitions of the topological string partition 
functions, as previously investigated in 
\cite{Pasquetti:2009jg} and in \cite{HO}. A formula for the Borel transform 
had been first proposed in \cite{Pasquetti:2009jg}, and in \cite{HO} it was
conjectured that the Borel transform along the real axis
is equal to \rf{polecancel}. Extensive numerical studies provided 
convincing evidence for these proposals. \\

Our Theorem \ref{theorem1} offers a more complete  picture.
It shows that the Borel summations
$F_{\rho}(\la,t)$ interpolate between $F_{\mathrm{GV}}(\la,t)$ and $F_{\mathrm{np}}(\la,t)$.
All the functions $F_{\rho}(\la,t)$ defined by different choices of the 
ray $\rho$ can be regarded as different re-packagings of the same
information, contained in the formal series  (\ref{resconfree}).
Any of these summations can serve as a candidate for 
a non-perturbative definition of the topological string partition
function of the resolved conifold. Additional requirements have 
to be imposed to distinguish a particular choice among  others.\\

Defining the topological string partition 
functions by Borel summation whenever this possibility exists seems
to be the most canonical way to associate actual functions to the 
divergent series (\ref{resconfree}).
The price to pay 
is that the resulting function is only piecewise analytic, 
having jumps across the rays $\pm l_k$. However, as will be
explained in the rest of the paper, there is interesting 
information contained in these jumps. We are going to demonstrate that
the jump functions encode information on the 
spectrum of BPS states on the resolved conifold in 
a particularly simple and transparent way by 
relating them to the Riemann--Hilbert problem formulated 
in \cite{BridgelandCon} which takes as input data the generalised DT invariants 
for the resolved conifold.\\

The Borel summations $F_{\rho_k}(\lambda,t)$ each have natural domains of definition, 
bounded by the rays $l_k$ and $l_{k-1}$. It seems important to note, however,
that the functions $F_{\rho_k}(\lambda,t)$ can be analytically continued in 
$\la$ to larger domains of definition containing $l_k$ and $l_{k-1}$. This suggests
to regard the analytically continued functions $F_{\rho_k}(\lambda,t)$ as local
sections of a line bundle defined by taking exponentials of the
jumps $\phi_{l_k}(\la,t)=F_{\rho_{k+1}}(\lambda,t)-F_{\rho_{k}}(\lambda,t)$ 
as transition functions. This line bundle, together with the collection 
of distinguished local sections $F_{\rho_k}(\lambda,t)$ is a natural geometric object 
canonically associated to the formal series (\ref{resconfree})
by Borel summation. We will see that it is a natural 
analog of the line bundle proposed in \cite{CLT20} for the case of the 
resolved conifold. 

\subsubsection{} Let us note that the differences 
$F_{\mathrm{D},k}(\lambda,t):=F_{\rho_k}(\lambda,t)-F_{\mathrm{GV}}(\lambda,t)$ 
can be
represented as sums of terms which are all 
proportional to an exponential 
function having dependence with respect to the topological 
string coupling $\lambda$ of the form 
$e^{(const.)/\lambda}$.
It is therefore natural to associate the 
differences $F_{\mathrm{D},k}(\lambda,t)$ with 
non-perturbative effects in string theory. They can be represented
as a sum over the Stokes jumps across the rays $l_k$ enclosed by 
$\rho_k$ and $\ii \BR$. We will see that 
these  jumps are in a one-to-one correspondence with 
D-branes in type II string theory on the resolved conifold. \\

A dependence of the form of the form 
$e^{(const.)/\lambda}$ is characteristic
for non-perturbative effects in string theory having a 
world-sheet description through disk amplitudes with boundaries 
associated to D-branes. Such disk amplitudes can represent 
central charge functions of D-branes in type II string theory 
\cite{HIV00}.
We will see that the constants in the exponential 
functions  $e^{(const.)/\lambda}$ appearing in the
differences $F_{\mathrm{D},k}(\lambda,t)$
have a simple relation to the central charge functions of the D-branes
associated to the jumps. The functions $F_{\mathrm{D},k}(\lambda,t)$ 
can be represented as sums over all terms which are exponentially 
suppressed in the wedge of the $\la$-plane bounded by $l_k$ and $l_{k-1}$.\\

These observations suggest that the non-perturbative
effects represented by the functions $F_{\mathrm{D},k}(\lambda,t)$
may have a world-sheet description in terms of 
disk amplitudes with boundaries associated to stable 
D-branes representing states in the BPS-spectrum of the resolved 
conifold. The set of D-branes contributing to the non-perturbatively defined partition
functions would then depend on the phase of $\lambda$, and jump across the rays $l_k$.
It would be interesting to verify this 
interpretation more directly. 

\subsection{Previous results}\label{prevres}

Previous work on this subject had obtained several important partial results. 
The first study of the Borel summability of the series $\widetilde{F}(\lambda,t)$ was performed
in \cite{Pasquetti:2009jg}, where an explicit formula for the Borel transform was found.
While the direct comparison of the formula derived in \cite{Pasquetti:2009jg} with \rf{Borel-sum}
is not completely straightforward, 
it is easy to see that the poles and residues agree.\\

Another approach to the summation of the formal series $\widetilde{F}(\lambda,t)$
has been proposed in \cite{HO}. The summation considered in \cite{HO} is 
an analytic function $F^{\rm resum}_{\rm coni}(\la,t)$ defined
through an explicit integral representation. 
Numerical evidence has been presented for the conjecture that 
$F^{\rm resum}_{\rm coni}(\la,t)$ is equal to the Borel summation $F_{\BR_{>0}}(\la,t)$
along $\rho=\BR_{>0}$ in our notations. 
We will later in Section \ref{FHOrel} explicitly establish the relation between
$F^{\rm resum}_{\rm coni}(\la,t)$ and $F_{\BR_{>0}}(\la,t)$ considered in our paper. 
It was furthermore proposed in \cite{HO} that the function 
$F^{\rm resum}_{\rm coni}(\la,t)$ admits the decomposition
\rf{polecancel}. This conjecture 
has been extensively checked numerically.\\

Interesting relations with spectral determinants of 
finite difference operators along the lines of \cite{GHM} 
have been found in  \cite{BGT}. Further exploration of the 
relations to our results should be illuminating. \\

It has been demonstrated in \cite{BridgelandCon} that a special function closely related to the
triple sine function has \rf{formal} as its asymptotic expansion.  
The relation between the triple sine function and the formal series $\widetilde{F}(\la,t)$ 
has
stimulated the work \cite{alim2020difference,alim2021integrable,Alim:2021ukq}  studying  the function 
defined on the right side of \rf{FR+-def}
as a promising candidate for a non-perturbative definition of the topological string 
partition function. It was identified in \cite{alim2021integrable} as a solution with pleasant analytic properties of a difference equation \cite{alim2020difference} which governs the topological string free energy. 
In \cite{Alim:2021ukq}, the non-perturbative content of this function was extracted demonstrating 
that this function admits the 
decomposition 
\rf{polecancel} and matching in particular with the results of \cite{HO}.\\

Further work on the function  $F_{\rm np}^{}(\la,t)$ 
in connection with the non-perturbative structure of topological strings can be found in \cite{Lockhart:2012vp,Krefl:2015vna}.

\section{Proofs of the results of Section \ref{resultssect}}\label{proofs}

In this section we prove each of the points of Theorem \ref{theorem1}.
Our approach is strongly inspired by the paper
\cite{garoufalidis2020resurgence} which has studied the analogous problem
for the non-compact quantum dilogarithm function.
Each of the four subsections below corresponds to each of the four points of the Theorem.

\begin{subsection}{The Borel transform}\label{boreltransproof}

We start by proving the first part of Theorem \ref{theorem1}, concerning the Borel transform of of $\widetilde{F}(\lambda,t)$. We remark that an alternative expression of the Borel transform was previously given in  \cite{Pasquetti:2009jg}, which we recall in Section \ref{oldBT}. \\

Recall the asymptotic expansion of the topological string free energy for the resolved conifold, which is given by \eqref{resconfree}:
\begin{equation}
\begin{split}
    \widetilde{F}(\lambda,t)&= \sum_{g=0}^{\infty} \lambda^{2g-2} \widetilde{F}^g(t)= \frac{1}{\lambda^2} \textrm{Li}_{3}(Q)+\frac{B_2}{2}\mathrm{Li}_1(Q) + \sum_{g=2}^{\infty} \lambda^{2g-2} \frac{(-1)^{g-1}B_{2g}}{2g (2g-2)!}\, \textrm{Li}_{3-2g} (Q) \,  \\
&=\frac{1}{\lambda^2} \textrm{Li}_{3}(Q)  +\frac{B_2}{2}\mathrm{Li}_1(Q)+ \Phi(\check{\lambda},t)\,, \quad \check{\lambda}=\frac{\lambda}{2\pi}\,, \quad Q=e^{2\pi it}\,,
\end{split}
\end{equation}
We use the property
\begin{equation} \label{polylogder}
\theta_Q \textrm{Li}_s(Q) =\textrm{Li}_{s-1} (Q)\,,  \quad \theta_Q:= Q \,\frac{d}{dQ}\,,
\end{equation}
to write 
\begin{equation}\label{derrep}
\widetilde{F}^g=\frac{(-1)^{g-1}B_{2g}}{2g (2g-2)!} \,\theta_Q^{2g} \textrm{Li}_3(Q)\,, \quad g\ge2\,.
\end{equation}
Furthermore, using that $\theta_Q=\frac{1}{2\pi \I}\partial_t$ we obtain
\begin{equation}
   \widetilde{F}^g=\frac{(-1) B_{2g}}{2g (2g-2)!(2\pi )^{2g}} \,\partial_t^{2g} \textrm{Li}_3(Q)\,, \quad g\ge2\,. 
\end{equation}
We thus have
\begin{equation}\label{formsumborel}
   \Phi(\check{\lambda},t)=-\frac{1}{4\pi^2} \sum_{g=2}^{\infty}\frac{ B_{2g}}{2g (2g-2)!} \,\check{\lambda}^{2g-2} \partial_t^{2g} \textrm{Li}_3(Q)\,. 
\end{equation}

We now wish to compute the Borel transform of $\Phi(\check{\lambda},t)$ 
and specify its domain of convergence.
The Borel transform is defined as the formal power series $G(\xi,t):=\mathcal{B}(\Phi(-,t))(\xi)$, where
\begin{equation}
    \mathcal{B}\colon\check{\lambda}\mathbb{C}[[\check{\lambda}]]\to \mathbb{C}[[\xi]], \;\;\;\; \mathcal{B}(\check{\lambda}^{n+1})=\frac{\xi^n}{n!}.
\end{equation}
Namely, we wish to study

\begin{equation}
     G(\xi,t)= -\frac{1}{4\pi^2}\sum_{g=2}^{\infty} \frac{ B_{2g}}{2g (2g-2)! (2g-3)!} \xi^{2g-3}\,\partial_t^{2g} \textrm{Li}_3(Q)\,.
\end{equation}

In order to do this, it will be convenient to first recall the Hadamard product and a certain integral representation thereof. The techniques used below follow the lines of \cite{garoufalidis2020resurgence}.

\begin{dfn}
Consider two formal power series $\sum_{n=0}^{\infty}a_nz^n,\sum_{n=0}^{\infty}b_nz^n\in \mathbb{C}[[z]]$. Then the Hadamard product $\oast\colon\mathbb{C}[[z]]\times \mathbb{C}[[z]] \to \mathbb{C}[[z]]$ is defined by
\begin{equation}
    \bigg(\sum_{n=0}^{\infty}a_nz^n\bigg)\oast \bigg(\sum_{n=0}^{\infty}b_nz^n\bigg)=\sum_{n=0}^{\infty}a_nb_nz^n\,.
\end{equation}
\end{dfn}

\begin{lem} \label{lemHad} Consider two holomorphic functions near $z=0$ having series expansions
\begin{equation}
    f_1(z)=\sum_{n=0}^{\infty}a_nz^n, \;\;\; f_2(z)=\sum_{n=0}^{\infty}b_nz^n
\end{equation}
with radius of convergence $r_1>0$ and $r_2>0$, respectively. Then $(f_1\oast f_2)(z)$ converges for $|z|<r_1r_2$, and for any $\rho \in (0,r_1)$ the following holds for $|z|<\rho r_2$:
\begin{equation}
    (f_1\oast f_2)(z)=\frac{1}{2\pi i}\int_{|s|=\rho}\frac{\mathrm d s}{s}f_1(s)f_2\Big(\frac{z}{s}\Big)\,.
\end{equation}

\end{lem}
\begin{proof} 

By using the limsup definition of the radius of convergence, one can easily check that the radius of convergence of $(f_1\oast f_2)(z)$ must be bigger or equal than $r_1r_2$. On the other hand, we have that for $|z|<\rho r_2<r_1r_2$,
\begin{equation}
    \begin{split}
   (f_1\oast f_2)(z)&=\sum_{n=0}^{\infty}a_nb_nz^n=\sum_{n=0}^{\infty}\bigg(\frac{1}{2\pi \I}\int_{|s|=\rho}\mathrm d s\,\frac{f_1(s)}{s^{n+1}}\bigg)b_nz^n\\
   &=\frac{1}{2\pi \I}\int_{|s|=\rho}\frac{\mathrm d s}{s}f_1(s)\sum_{n=0}^{\infty}b_n\Big(\frac{z}{s}\Big)^n=\frac{1}{2\pi \I}\int_{|s|=\rho}\frac{\mathrm d s}{s}f_1(s)f_2\Big(\frac{z}{s}\Big)\\
   \end{split}
\end{equation}
where the interchange of sum and integrals is justfied by the Fubini--Tonelli theorem and the absolute convergence of $(f_1\oast f_2)(z)$.
\end{proof}

The idea is to write 
\begin{equation}
    G(\xi,t)=(f_1\oast f_2(-,t))(\xi)
\end{equation}
for two functions $f_1(\xi)$, $f_2(\xi,t)$ which are holomorphic near $\xi=0$, and then use the first part of the previous lemma. 
We will take $f_1(\xi)$, $f_2(\xi,t)$ to be the following:
\begin{equation}
    \begin{split}
     f_1(\xi) &= -\frac{1}{4\pi^2}\sum_{g=2}^{\infty} \frac{ (2g-1)\,B_{2g}}{ (2g)!} \xi^{2g-3}\\
     f_2(\xi,t)&= \sum_{g=2}^{\infty}  \frac{\xi^{2g-3}}{(2g-3)!}\,\partial_t^{2g} \textrm{Li}_3(Q)=\sum_{g=2}^{\infty}  \frac{\xi^{2g-3}}{(2g-3)!}\,(2\pi \I)^{2g} \textrm{Li}_{3-2g}(Q)\,.
     \end{split}
\end{equation}

\begin{prop}\label{boreltransconv} Let $t\in \mathbb{C}^{\times}$ with $|\mathrm{Re}(t)|<\frac12$. Then $G(\xi,t)$ converges for $|\xi|<2\pi |t|$.
\end{prop}

\begin{proof}
Using the fact that 
\begin{equation}
    B_{2g}\sim (-1)^{g+1}4\sqrt{\pi g}\Big(\frac{g}{\pi e}\Big)^{2g} \;\;\; \text{as} \;\;\; g\to \infty,
\end{equation}
we find that the radius of convergence for $f_1(\xi)$ is $2\pi$. On the other hand, using the fact that for $|\text{Re}(t)|<1/2$, we have
\begin{equation}
        \mathrm{Li}_{3-2g}(e^{2\pi \I t})\sim \Gamma(1-3+2g)(-2\pi \I t)^{3-2g-1}\, \;\;\; \text{as} \;\;\; g\to \infty,
\end{equation}
we find that the radius of convergence of $f_2(\xi,t)$ is $r_2(t)=|t|$.\\

By the use of Lemma \ref{lemHad}, we find that provided $t\in \mathbb{C}^{\times}$ satisfies $|\text{Re}(t)|<\frac12$, we have that $G(\xi,t)=(f_1\oast f_2(-,t))(\xi)$ converges for $|\xi|<r_1r_2(t)=2\pi |t|$.
\end{proof}

We now wish to use the integral representation of the Hadamard product to find a more convenient representation of $G(\xi,t)$.

\begin{prop}\label{Borelhadprod}
With the same hypothesis as in Proposition \ref{boreltransconv}, we have 
    \begin{equation}\label{Boreltrans}
        G(\xi,t) =\frac{1}{(2\pi )^2}\sum_{m\in\Z \setminus \{0\}}\frac{1}{m^3}\bigg(1+\frac{\xi}{2}\frac{\partial}{\partial \xi}\bigg)\bigg(\frac{1}{1-e^{-2\pi \I t + \xi/m}}-\frac{1}{1-e^{-2\pi \I t - \xi/m}}\bigg) \end{equation}
In particular, for fixed $t$, the expression on the right allows us to analytically continue $G(\xi,t)$ to a meromorphic function in $\xi$ with poles at $\xi=2\pi \I(t+k)m$ for $k \in \mathbb{Z}$ and $m \in \mathbb{Z}\setminus\{0\}$.
\end{prop}
\begin{proof}
The idea is to now use the integral representation of the Hadamard product in Lemma \ref{lemHad}, together with the results in the proof of Proposition \ref{boreltransconv}. In particular, for $t\in \mathbb{C}^{\times}$ with $|\text{Re}(t)|<1/2$ and $\rho \in (0,2\pi)$, we have for $|\xi|<\rho|t|$
\begin{equation}
    G(\xi,t)=\frac{1}{2\pi \mathrm i}\int_{|s|=\rho}\frac{\mathrm ds}{s}f_1(s)f_2\Big(\frac{\xi}{s},t\Big)
\end{equation}
where $f_1(\xi)$ and $f_2(\xi,t)$ are as in Proposition \ref{boreltransconv}.\\

We have, on the one hand,
\begin{equation}
    \begin{split}
    f_1(\xi) &= -\frac{1}{4\pi^2}\sum_{g=2}^{\infty} \frac{ (2g-1)\,B_{2g}}{ (2g)!} \xi^{2g-3}
    = -\frac{1}{4\pi^2}\frac{1}{\xi} \partial_{\xi} \left( \frac{1}{\xi} \sum_{g=2}^{\infty} \frac{ \,B_{2g}}{ (2g)!} \xi^{2g}\right)\\
    &= -\frac{1}{4\pi^2}\frac{1}{\xi} \partial_{\xi} \left( \frac{1}{\xi} \left( \sum_{g=0}^{\infty} \frac{ \,B_{g}}{ g!} \xi^{g} -1 + \frac{\xi}{2}-\frac{\xi^2}{12}\right)\right) \\
    &= -\frac{1}{4\pi^2}\frac{1}{\xi} \partial_{\xi} \left( \frac{1}{\xi} \left( \frac{\xi}{e^{\xi}-1} -1 + \frac{\xi}{2}-\frac{\xi^2}{12}\right)\right)\\
    &= -\frac{1}{4\pi^2}\left(  \frac{1}{\xi^3}- \frac{1}{\xi (e^{\xi/2}-e^{-\xi/2})^2}-\frac{1}{12\xi}\right)\,,
    \end{split}
\end{equation}
where we have used the expression for the generating function of the Bernoulli numbers
\begin{equation}\label{eq:bern}
\frac{w}{e^w-1} = \sum_{n=0}^{\infty} B_n \frac{w^n}{n!}\,,
\end{equation}
and the fact that except $B_1=-\frac12$, all odd Bernoulli numbers vanish. From the final expression we see that $f_1(\xi)$ admits an analytic continuation to a meromorphic function with double poles at $\xi=2\pi \I \mathbb{Z} \setminus \{0\}$.\\

On the other hand, for $f_2(\xi,t)$, we have
\begin{align*}
        f_2(\xi,t) &= \sum_{g=2}^{\infty}  \frac{\xi^{2g-3}}{(2g-3)!}\,\partial_t^{2g} \textrm{Li}_3(Q) 
    = \partial_{\xi}^3\, \left(\sum_{g=1}^{\infty}  \frac{\partial_t^{2g} \textrm{Li}_3(Q)}{(2g)!} \xi^{2g}\right) \\
    &= \partial_{\xi}^3 \left( \frac{1}{2} (\textrm{Li}_3(e^{2\pi \I(t+\xi)})+\textrm{Li}_3(e^{2\pi \I(t-\xi)})) - \textrm{Li}_3(e^{2\pi \I t}) \right) \\
    &=\frac{(2\pi \I)^3}{2}\left(\mathrm{Li}_0\big(e^{2\pi \I(t + \xi)}\big)-\mathrm{Li}_0\big(e^{2\pi \I(t - \xi)}\big)\right)\\
    &=\frac{(2\pi \I)^3}{2}\bigg(\frac{e^{2\pi \I(t + \xi)}}{1-e^{2\pi \I(t + \xi)}}-\frac{e^{2\pi \I(t - \xi)}}{1-e^{2\pi \I(t - \xi)}}\bigg)\,,   \numberthis 
\end{align*}

so that $f_2(\xi,t)$ admits an analytic continuation in $\xi$ with simple poles at $\pm t + \mathbb{Z}$. The integral representation then becomes
\begin{equation}\label{borelint}
    G(\xi,t)=\frac{1}{2}\int_{|s|=\rho} \frac{\mathrm ds}{s}\left(\frac{1}{s^3}- \frac{e^s}{ s(e^{s}-1)^2}-\frac{1}{12s}\right)\bigg(\frac{e^{2\pi \I(t + \xi/s)}}{1-e^{2\pi \I(t + \xi/s)}}-\frac{e^{2\pi \I(t - \xi/s)}}{1-e^{2\pi \I(t - \xi/s)}}\bigg)\,.
\end{equation}

Notice that $f_2(\xi/s,t)$ as a function of $s$ has simple poles at $s=\pm \xi/(t+k)$ for all $k \in \mathbb{Z}$. By our assumption that $|\xi|<\rho|t|$ and $|\mathrm{Re}(t)|<\frac12$, we have
\begin{equation}
    \Big|\frac{\pm \xi}{t+k}\Big|<\rho\frac{|t|}{|t+k|}\leq \rho\;,
\end{equation}
so that all the poles of $f_2(\xi/s,t)$ lie inside the contour. Furthermore, since $\rho<2\pi$, all the poles of $f_1(s)$ lie outside the contour.  \\

If, for each $k \in \mathbb{Z}$ with $k>1$, we denote as $\gamma_k$ the contour given by $|s|=\pi(2k+1)$, then between $|s|=\rho$ and $\gamma_k$, we have the poles at $\pm 2\pi \I n$ for $n=1,...,k$ due to $f_1(s)$. We can therefore write the following for any $k>1$:
\begin{align}\label{intcontlim}
     G(\xi,t)&=\frac{1}{2}\int_{|s|=\rho} \frac{\mathrm d s}{s}\left(\frac{1}{s^3}- \frac{e^s}{ s(e^{s}-1)^2}-\frac{1}{12s}\right)\bigg(\frac{e^{2\pi \I(t + \xi/s)}}{1-e^{2\pi \I(t + \xi/s)}}-\frac{e^{2\pi \I(t - \xi/s)}}{1-e^{2\pi \I(t - \xi/s)}}\bigg)\\
     &= 2\pi \I\sum_{m\in \mathbb{Z}:-k<m<k , m\neq 0 }\frac{1}{2}\left.\frac{\mathrm d}{\mathrm ds}\left(\frac{e^s(s-2\pi \I m)^2}{(e^s-1)^2s^2}\bigg(\frac{e^{2\pi \I(t + \xi/s)}}{1-e^{2\pi \I(t + \xi/s)}}-\frac{e^{2\pi \I(t - \xi/s)}}{1-e^{2\pi \I(t - \xi/s)}}\bigg)\right)\right|_{s=2\pi \I m}\notag\\
     &\;\;\;+\frac{1}{2}\int_{\gamma_k} \frac{\mathrm ds}{s} \left(\frac{1}{s^3}- \frac{e^s}{ s(e^{s}-1)^2}-\frac{1}{12s}\right)\bigg(\frac{e^{2\pi \I(t + \xi/s)}}{1-e^{2\pi i(t + \xi/s)}}-\frac{e^{2\pi \I(t - \xi/s)}}{1-e^{2\pi \I(t - \xi/s)}}\bigg)\,,
\notag\end{align}
where the terms in the sum come from the contribution of the (clockwise) contours around the poles between the two contours.  \\

One can check that $f_2(\xi/s,t)=\mathcal{O}(1/s)$ as $s\to \infty$, while $f_1|_{\gamma_k}=\mathcal{O}(1/k)$ as $k\to \infty$. Hence, taking the limit $k\to \infty$ in (\ref{intcontlim}) we obtain the following expression:
\begin{align*}
        G(\xi,t)&=2\pi \I\sum_{m \in \mathbb{Z}-\{0\}}\frac{1}{2}\left.\frac{\mathrm d}{\mathrm ds}\left(\frac{e^s(s-2\pi \I m)^2}{(e^s-1)^2s^2}\bigg(\frac{e^{2\pi \I(t + \xi/s)}}{1-e^{2\pi \I(t + \xi/s)}}-\frac{e^{2\pi \I(t - \xi/s)}}{1-e^{2\pi \I(t - \xi/s)}}\bigg)\right)\right|_{s=2\pi \I m}\\
        &=-\sum_{m\in\Z \setminus \{0\}}\frac{1}{(2\pi \I)^2}\left(\frac{1}{m^3}\left(\frac{e^{2\pi \I t + \xi/m}}{1-e^{2\pi \I t + \xi/m}}-\frac{e^{2\pi \I t - \xi/m}}{1-e^{2\pi \I t - \xi/m}}\right)\right.\\ &\qquad\qquad\qquad+\left.\frac{\xi}{2m^4}\left(\frac{e^{2\pi \I t + \xi/m}}{(1-e^{2\pi \I t + \xi/m})^2}+\frac{e^{2\pi \I t - \xi/m}}{(1-e^{2\pi \I t - \xi/m})^2}\right)\right). \numberthis
\end{align*}
The result then follows by a simple rewriting of the summands.
\end{proof}
\end{subsection}

A direct check that $G(\xi,t)$ is the Borel transform of $\widetilde{F}(\la,t)$ can be found in Lemma \ref{lemmaApp} from Appendix \ref{App:TaylorBorel}.

\subsubsection{Previous expression for the Borel transform}\label{oldBT}
In the following we review an expression for the Borel transform of the topological string free energy for the resolved conifold obtained in \cite{Pasquetti:2009jg}, starting again with
\begin{equation}
\begin{split}
    \widetilde{F}(\lambda,t)&= \sum_{g=0}^{\infty} \lambda^{2g-2} \widetilde{F}^g(t)= \frac{1}{\lambda^2} \textrm{Li}_{3}(q)+\frac{B_2}{2}\mathrm{Li}_1(q) + \sum_{g=2}^{\infty} \lambda^{2g-2} \frac{(-1)^{g-1}B_{2g}}{2g (2g-2)!} \textrm{Li}_{3-2g} (q) \,  \\
&=\frac{1}{\lambda^2} \textrm{Li}_{3}(q)  +\frac{B_2}{2}\mathrm{Li}_1(q)+ \Phi(\check{\lambda},t)\,, \quad \check{\lambda}=\frac{\lambda}{2\pi}\,.
\end{split}
\end{equation}
Using the series representation of the polylogarithm
\begin{equation}
    \textrm{Li}_{s}(e^{2\pi i t})= \Gamma(1-s) \, \sum_{k\in \mathbb{Z}} (2\pi i)^{s-1}(k-t)^{s-1}\,,
\end{equation}
valid for $\textrm{Re}(s)<0$ and $t\notin \mathbb{Z}$,
we can write
\begin{equation}
\begin{split}
    \Phi(\check{\lambda},t)&=\sum_{g=2}^{\infty} \lambda^{2g-2} \frac{(-1)^{g-1}B_{2g}}{2g (2g-2)!}\, \textrm{Li}_{3-2g} (q)= \sum_{g=2}^{\infty} \check{\lambda}^{2g-2} \frac{B_{2g}}{2g (2g-2)} \sum_{k\in \mathbb{Z}} (k-t)^{2-2g}. 
\end{split}
\end{equation}
Taking the Borel transform of $\Phi(\check{\lambda},t)$, we find
\begin{equation}\label{diffBoreltrans}
\begin{split}
     G(\xi,t)&=\sum_{g=2}^{\infty} \frac{B_{2g}}{2g\, (2g-2)!} \xi^{2g-3} \sum_{k\in \mathbb{Z}} (k-t)^{2-2g}=\sum_{k\in\Z} \frac{1}{ \xi}\left(\frac{(k-t)^2}{\xi^2} - \frac{e^{\xi/(k-t)}}{(e^{\xi/(k-t)} -1)^2}-\frac{1}{12}\right),\\
\end{split}
\end{equation}
where the second equality follows from
\begin{equation}
    \frac{e^w}{(e^w-1)^2}=\frac{1}{w^2}-\frac{1}{12}-\sum_{g=2}^{\infty}\frac{B_{2g}}{2g (2g-2)!} w^{2g-2}\,,
\end{equation}
which can be obtained by taking a derivative of the generating function of Bernoulli numbers \eqref{eq:bern}
and rearranging the outcome. We note here that this expression for the Borel transform, which was previously obtained in \cite{Pasquetti:2009jg} has the same set of poles at
$$ \xi= 2\pi \I m (t+k)\,, \quad m\in \mathbb{Z}\setminus\{0\} \,,\quad k\in \mathbb{Z}\,,$$ 
as the one we have obtained in Theorem \ref{theorem1}. We will later show that it also has the same Stokes jumps.\\

One advantage of the expression (\ref{Boreltrans}) compared to (\ref{diffBoreltrans}) is that the first gives a well-defined expression for $t \in \mathbb{Z}$. As we will see below, this will allow us to express the Borel transform of constant map contribution of (\ref{freeenergydecomp}) in terms of $G(\xi,0)$ (see Corollary \ref{constmapBorelsum}).

\begin{rem} It seems one might be able to obtain (\ref{diffBoreltrans}) from the integral representation (\ref{borelint}) by deforming the contour to $0$ instead of $\infty$, and summing over the residues of the poles inside the contour. The only technical issue is that it seems harder to show that the contour limiting to $s=0$ limits to a zero contribution.

\end{rem}

\subsection{The Borel sum along \texorpdfstring{$\mathbb{R}_{>0}$}{TEXT}}

We now prove the second point of Theorem \ref{theorem1}. Hence, we wish to study the Borel sum of $\widetilde{F}(\lambda,t)$ along $\mathbb{R}_{>0}$. More generally, we define the following:

\begin{dfn} Given $t\in \mathbb{C}$ with $\text{Im}(t)\neq 0$, and a ray $\rho \subset \mathbb{C}^{\times}$ from $0$ to $\infty$ avoiding the poles $2\pi \I(t+k)m$ of $G(\xi,t)$ and which is different from $\pm \I \mathbb{R}_{>0}$, we define

\begin{equation}
    F_{\rho}(\lambda,t) :=  \frac{1}{\lambda^2} \mathrm{Li}_{3}(Q)+ \frac{B_2}{2}\mathrm{Li}_{1}(Q)+\int_{\rho} \mathrm d \xi \,e^{-\xi/\check{\lambda}} G(\xi,t)\,, \quad Q = e^{2\pi\I t}\,,
\end{equation}
for $\lambda$ in the half-plane $\mathbb{H}_{\rho}$ centered at $\rho$. We call $F_{\rho}(\lambda,t)$ the Borel sum of $\widetilde{F}(\lambda,t)$ along $\rho$.

\end{dfn}

The integral appearing in (\ref{FR+-def}) corresponds to an integral representation of a certain function $G_3(z\,|\, \omega_1,\omega_2)$ related to the triple sine function. Hence, before studying the Borel sum along $\mathbb{R}_{>0}$, we recall how this function is defined and some of its properties. For a convenient review of the special functions appearing below and their properties, see for example \cite[Section 4]{BridgelandCon} and the references cited therein.

\begin{dfn}\label{G3def} For $z\in \mathbb{C}$ and $\omega_1,\omega_2\in \mathbb{C}^{\times}$, we define 
\begin{equation}
G_3(z\,|\,\omega_1,\omega_2):=\exp\Big(\frac{\pi \I}{6}B_{3,3}(z+\omega_1\,|\,\omega_1,\omega_1,\omega_2)\Big)\cdot \sin_3(z+\omega_1\,|\,\omega_1,\omega_1,\omega_2)\,,
\end{equation}
where $\sin_3(z\,|\,\omega_1,\omega_2,\omega_3)$ denotes the triple sine function, and $B_{3,3}(z\,|\, \omega_1,\omega_2,\omega_3)$ is the multiple Bernoulli polynomial. 

\end{dfn}

What will be most important for us are the following properties:

\begin{prop}\cite[Prop.\ 4.2]{BridgelandCon}\cite[Prop.\ 2]{Narukawa}\label{intrep2}
$G_3(z\,|\,\omega_1,\omega_2)$ is a single-valued meromorphic function under the assumption $\omega_1/\omega_2\not\in \mathbb{R}_{<0}$. Furthermore, we have

\begin{itemize}
    \item It is regular everywhere, and vanishes only at the points
    \begin{equation}
        z=a\omega_1+b\omega_2, \quad a,b\in \mathbb{Z}\,,
    \end{equation}
    with $a<0$ and $b\leq 0$, or $a>0$ and $b>0$.
    \item Let $\mathrm{Re}(\omega_i)>0$ and $-\mathrm{Re}(\omega_1)<\mathrm{Re}(z)<\mathrm{Re}(\omega_1+\omega_2)$. Then 
\begin{equation}
    G_3(z\, |\, \omega_1,\omega_2)=\exp\bigg(- \int_{\mathbb{R}+\mathrm i0^+} \frac{\mathrm du}{8u}\frac{e^{u(z-\omega_2/2)}}{\sinh(\omega_2u/2)(\sinh(\omega_1 u/2))^2}\;\bigg).
\end{equation}
\end{itemize}

\end{prop}

\begin{dfn} For $\mathrm{Re}\,\check{\lambda} >0$ and $-\mathrm{Re}(\check{\lambda}) < \mathrm{Re}(t) < \mathrm{Re} (\check{\lambda}+1)$, we define
\begin{equation}\label{eq:F-nonpert}
F_{\mathrm{np}}(\lambda,t) := \log G_3(t\,|\, \check\lambda,1)=- \int_{\mathbb{R}+\mathrm i0^+} \frac{\mathrm du}{8u}\frac{e^{u(t-1/2)}}{\sinh(u/2)(\sinh(\check\lambda u/2))^2}\; ,
\end{equation}
\end{dfn}

We now wish to relate $F_{\mathbb{R}_{>0}}(\lambda,t)$ to $F_{\mathrm{np}}(\lambda,t)$. For this, we will need the following lemma, giving a ``Woronowicz form" for $F_{\mathrm{np}}$. We remark that a similar form for the triple-sine function was conjectured in \cite[Equation B.17] {ConformalTBA}.

\begin{lem}\label{woronowiczlemma} Let $t\in \mathbb{C}$ be such that $0<\mathrm{Re}(t)<1$, $\mathrm{Im}(t)>0$, and let $\lambda $ be in the sector determined by $l_0=\mathbb{R}_{<0}\cdot 2\pi \I t$ and $l_{-1}=\mathbb{R}_{<0}\cdot 2\pi \I (t-1)$. Furthermore, assume that $\mathrm{Re}(t)<\mathrm{Re}(\check\lambda +1)$. Then $F_{\mathrm{np}}(\lambda,t)$
admits the following Woronowicz form:
\begin{equation}\label{altform}
    F_{\mathrm{np}}(\lambda,t)=\frac{1}{(2\pi)^2}\int_{\mathbb{R}+\I0^+}\mathrm dv\frac{v}{1-e^v}\log(1-e^{\check\lambda v +2\pi \I t})\,.
\end{equation}
\end{lem}
\begin{proof}
We will follow the method of \cite{garoufalidis2020resurgence}, based on the unitarity of the 
Fourier transform:
\[
\langle f,g\rangle=\langle Ff,Fg\rangle,\qquad \langle f,g\rangle=\int_\BR \mathrm dx\;f(x)\overline{g(x)}, \qquad (F\psi)(x)=\int_\BR \mathrm dy\;e^{2\pi\ii\,xy}\psi(y).
\]
We start by defining for sufficiently small $\ep>0$,
\begin{equation}
f_\ep(x):=e^{-\ep x}\log\big(1-e^{\check\la x
+2\pi\ii \,t}\big),\;\;\;\;
g_\ep(x):=e^{+\ep x}\frac{1}{1-e^{x+\ii\ep}},\;\;\;\;
G_\ep(x):=e^{+\ep x}\frac{x}{1-e^{x+\ii\ep}}.
\end{equation}
We then easily see that 
\begin{equation}\label{F0-eta0-scalarprod}
\lim_{\ep\ra 0^+}\frac{1}{(2\pi)^2}\langle f_\ep,G_\ep\rangle=\frac{1}{(2\pi)^2}\int_{\mathbb{R}+\I 0^+}dv\frac{v}{1-e^v}\log(1-e^{\check\lambda v +2\pi \I t})\;.
\end{equation}
We now compute the Fourier transform of $f_\ep$, $g_\ep$, and $G_\ep$. Setting $\zeta=2\pi x+\I\epsilon$, we find that 
\begin{equation}
        Ff_\ep(x)=\int_{\mathbb{R}}\mathrm dy\;e^{\I y\zeta}\log(1-e^{\check\lambda y + 2\pi \I t})=\frac{\I\check \lambda}{\zeta}\int_{\mathbb{R}}\mathrm dy\;\frac{e^{\I y \zeta}}{1-e^{-\check\lambda y -2\pi \I t}}\;,
\end{equation}
where we have integrated by parts, and used that the boundary terms vanish. The last integral has simple poles at $y=2\pi \I(k-t)/\check{\lambda}$, and under our assumptions for the parameters $t$ and $\lambda$, it is easy to check that the poles on the upper half-plane correspond to $k>0$, while those in the lower half-plane correspond to $k\leq 0$. If $\text{Re}(x)>0$, by an application of Jordan's lemma and the residue theorem, we can compute $Ff_\ep(x)$ by summing up the residues in the upper half-plane, obtaining
\begin{equation}
        \begin{split}
    \frac{\I\check \lambda}{\zeta}\int_{\mathbb{R}}\mathrm dy\;\frac{e^{\I y \zeta}}{1-e^{-\check\lambda y -2\pi \I t}}&=2\pi \I \frac{\I\check{\lambda}}{\zeta}\sum_{k=1}^{\infty}\frac{e^{\I y\zeta}}{\check{\lambda}}\Big|_{y=2\pi \I (k-t)/\check{\lambda}}=-\frac{2\pi}{\zeta}e^{2\pi t\zeta /\check\lambda}\sum_{k=1}^{\infty}e^{-2\pi k\zeta/\check{\lambda} }\\
    &=-\frac{\pi}{\zeta}e^{\pi \zeta(2t-1)/\check\lambda}\Big(2e^{-\pi \zeta/\check{\lambda}}\sum_{k=0}^{\infty}e^{-2\pi \zeta k/\check{\lambda}}\Big)=-\frac{\pi}{\zeta}\frac{e^{\pi \zeta(2t-1)/\check\lambda}}{\sinh(\pi \zeta/\check{\lambda})}\;,
    \end{split}
\end{equation}
where in the last equality we have used the Dirichlet series representation of $1/\sinh(z)$. Similarly, if $\text{Re}(x)<0$, we can compute $Ff_{\epsilon}$ summing up the residues in the lower half-plane, obtaining
\begin{align*}
   \frac{\I\check \lambda}{\zeta}\int_{\mathbb{R}}\mathrm dy\;\frac{e^{\I y \zeta}}{1-e^{-\check\lambda y -2\pi \I t}}&=-2\pi \I \frac{\I\check{\lambda}}{\zeta}\sum_{k=0}^{-\infty}\frac{e^{\I y\zeta}}{\check{\lambda}}\Big|_{y=2\pi \I (k-t)/\check{\lambda}}=\frac{2\pi}{\zeta}e^{2\pi t\zeta /\check\lambda}\sum_{k=0}^{\infty}e^{2\pi k\zeta/\check{\lambda} }\\ &=-\frac{\pi}{\zeta}e^{\pi \zeta(2t-1)/\check\lambda}\Big(-2e^{\pi \zeta/\check{\lambda}}\sum_{k=0}^{\infty}e^{2\pi \zeta k/\check{\lambda}}\Big)=-\frac{\pi}{\zeta}\frac{e^{\pi \zeta(2t-1)/\check\lambda}}{\sinh(\pi \zeta/\check{\lambda})}\,, \numberthis
\end{align*}
so that $Ff_{\ep}(x)$ exists for $x\neq 0$ and 
\begin{equation}
Ff_\ep(x)=-\frac{\pi}{\zeta}\frac{e^{\pi \zeta(2t-1)/\check\lambda}}{\sinh(\pi \zeta/\check{\lambda})}\;.
\end{equation}

The computation of $Fg_\ep(x)$ is simpler, and follows similar lines. One obtains that for $x\neq 0$,
\begin{equation}
    \overline{Fg_\ep(x)}=\pi \I \frac{e^{(\ep-\pi )\zeta}}{\sinh(\pi \zeta)}\,.
\end{equation}
On the other hand, since $G_\ep(x)=xg_\ep(x)$, we find
\begin{equation}
    \overline{FG_\ep(x)}=-\frac{1}{2\pi \I}\frac{\partial}{\partial x}\overline{Fg_\ep(x)}= \frac{\partial}{\partial \zeta}\Big(\frac{2\pi e^{\ep \zeta}}{1-e^{2\pi \zeta}}\Big)\,.
\end{equation}
We then have that
\begin{align*}
    \lim_{\ep\ra 0^+}\frac{1}{(2\pi)^2}\langle f_\ep,G_\ep\rangle&=\lim_{\ep\ra 0^+}\frac{1}{(2\pi)^2}\langle Ff_\ep,FG_\ep\rangle\\
    &=\frac{1}{(2\pi)^2}\int_{\mathbb{R}+\I 0^+}\mathrm dx \bigg(-\frac{\pi}{2\pi x}\frac{e^{\pi(2\pi x)(2t-1)/\check\lambda}}{\sinh(2\pi^2 x/\check\lambda)}\bigg)\bigg(\frac{(2\pi)^2}{4(\sinh(2\pi^2x))^2}\bigg)\\
    &=-\int_{\check{\lambda}^{-1}\cdot(\mathbb{R}+\I 0^+)}\frac{\mathrm dv}{8v} \frac{e^{v(t-1/2)}}{\sinh(v/2)(\sinh(\check{\lambda}v/2))^2}\\
    &=-\int_{\mathbb{R}+\I 0^+}\frac{\mathrm dv}{8v} \frac{e^{v(t-1/2)}}{\sinh(v/2)(\sinh(\check{\lambda}v/2))^2}=F_{\mathrm{np}}(\lambda,t)\, , \numberthis
\end{align*}
where we used the fact that the range of the parameter $\lambda$ allows us to deform the contour back to $\mathbb{R}+i0^{+}$. The result then follows. 
\end{proof}

We are now ready to prove the second point of Theorem \ref{theorem1}. Another proof is given in Appendix \ref{altproofBorel}.

\begin{prop} \label{BorelsumR} Under the same assumptions as in Lemma \ref{woronowiczlemma}, we have
\begin{equation}\label{F_0BorelS}
    F_{\mathrm{np}}(\lambda,t)=\frac{1}{\lambda^2}\mathrm{Li}_3(Q)+\frac{B_2}{2}\mathrm{Li}_1(Q) + \int_{0}^{\infty}\mathrm d\xi \,e^{-\xi/\check{\lambda}}G(\xi,t)\,,
\end{equation}
where $G(\xi,t)$ is the Borel transform (\ref{Boreltrans}) obtained in the previous section, $Q=e^{2\pi \I t}$, and $\check\lambda=\lambda/2\pi$.

\end{prop}

\begin{proof}

We start by performing the change of variables $y=\lambda v/2\pi$ on (\ref{altform}), obtaining
\begin{equation}
    \begin{split}
    F_{\mathrm{np}}(\lambda,t)&=\frac{1}{\lambda^2}\int_{\lambda(\mathbb{R}+\I 0^+)}\mathrm dy \;\; \frac{y}{1-e^{2\pi y/\lambda }}\log(1-e^{y+2\pi \I t})\\
    &=\frac{1}{\lambda^2}\int_{\mathbb{R}+\I 0^+}\mathrm dy \;\;  \frac{y}{1-e^{2\pi y/\lambda }}\log(1-e^{y+2\pi \I t})\\
    &=\lim_{\epsilon \to 0^+}\frac{1}{\lambda^2}\int_{\mathbb{R}}\mathrm dy \, \frac{y}{1-e^{2\pi y/\lambda -i\epsilon}}\log(1-e^{y+2\pi \I t})\,,
    \end{split}
\end{equation}
where in the second equality we have used that the range of $\lambda$ allows us to deform the contour back to $\mathbb{R}+\I 0^+$. Now using
\begin{equation}
    \frac{\mathrm d}{\mathrm dy}(-\log(1-e^{-2\pi y/\lambda +\I \epsilon}))=\frac{2\pi }{\lambda(1-e^{2\pi y/\lambda -\I \epsilon})}
\end{equation}
and integration by parts, we find
\begin{equation}
    \begin{split}
    F_{\mathrm{np}}(\lambda,t)=\lim_{\epsilon\to 0}\Big[&-\frac{y}{2\pi \lambda}\log(1-e^{-2\pi y/\lambda +i\epsilon})\log(1-e^{y+2\pi it})\Big|_{y=-\infty}^{\infty}\\
    &+\frac{1}{2\pi \lambda}\int_{\mathbb{R}}\mathrm dy\;\; \log(1-e^{-2\pi y/\lambda +i\epsilon})\Big(\log(1-e^{y+2\pi it})+\frac{y}{1-e^{-y-2\pi it}}\Big)\Big]\,.
    \end{split}
\end{equation}
Because $\text{Re}(\lambda)>0$, we obtain that the boundary terms vanish. Furthermore,  
splitting the integration over the left and right half intervals, one then obtains
\begin{equation}\label{p1}
    \begin{split}
    F_{\mathrm{np}}(\lambda,t)=\lim_{\epsilon\to 0}&\Big[\frac{1}{2\pi \lambda}\int_{0}^{\infty}\mathrm dy\;\; \log(1-e^{-2\pi y/\lambda +\I \epsilon})\Big(\log(1-e^{y+2\pi \I t})+\frac{y}{1-e^{-y-2\pi \I t}}\Big)\\
    &+\frac{1}{2\pi \lambda}\int_{0}^{\infty}\mathrm dy\;\; \log(1-e^{2\pi y/\lambda +\I\epsilon})\Big(\log(1-e^{-y+2\pi \I t})-\frac{y}{1-e^{y-2\pi \I t}}\Big)\Big]\\
    =\widetilde{H}(&\lambda,t) + \lim_{\epsilon\to 0}H(\lambda,t,\epsilon)\;,
    \end{split}
\end{equation}
where we have defined
\begin{align*}
    \widetilde{H}(\lambda,t):=\frac{1}{2\pi \lambda^2}\int_{0}^{\infty}\mathrm dy\;\; (2\pi y +\pi \I \lambda) &\Big(\log(1-e^{-y+2\pi \I t})-\frac{y}{1-e^{y-2\pi it}}\Big)\\
    H(\lambda,t,\epsilon):= \frac{1}{2\pi \lambda}\int_{0}^{\infty}\mathrm dy\;\; \log(1-e^{-2\pi y/\lambda +\I\epsilon})\Big(&\log(1-e^{y+2\pi \I t})+\log(1-e^{-y+2\pi \I t})\\
    &+\frac{y}{1-e^{-y-2\pi \I t}}-\frac{y}{1-e^{y-2\pi \I t}}\Big)\;. \numberthis
\end{align*}
One can compute $\widetilde{H}(\lambda,t)$ explicitly by performing an integration by parts to get rid of the log term:
\begin{equation}\label{Gcomp}
    \begin{split}
        \widetilde{H}(\lambda,t)&=\frac{1}{2\pi \lambda^2}\bigg((\pi y^2 +\pi \I \lambda y)\log(1-e^{-y+2\pi \I t})\Big|_{y=0}^{\infty}-\int_{0}^{\infty}\mathrm dy\;\; (\pi y^2 +\pi \I \lambda y) \frac{-1}{1-e^{y-2\pi \I t}}\bigg)\\
        &\;\;\;\;-\frac{1}{2\pi \lambda^2}\int_{0}^{\infty}\mathrm dy\;\; (2\pi y +\pi \I \lambda) \frac{y}{1-e^{y-2\pi \I t}}\\
        &=\frac{1}{2\lambda^2}\int_0^{\infty}\mathrm dy \frac{y^2}{e^{y-2\pi \I t}-1}\;.
    \end{split}
\end{equation}
Since $\text{Im}(t)>0$, we find that $|e^{2\pi \I t}|<1$, so that the last integral in (\ref{Gcomp}) corresponds to an integral representation of $\mathrm{Li}_3(e^{2\pi \I t})/\lambda^2$. Hence, we conclude that
\begin{equation}\label{p2}
    \widetilde{H}(\lambda,t)=\frac{1}{\lambda^2}\mathrm{Li}_3(e^{2\pi \I t})\;.
\end{equation}
On the other hand, by  expanding the first log term of $H$ and applying the Fubini--Tonelli theorem, we find that
\begin{equation}
    \begin{split}
    H(\lambda,t,\epsilon)
    =-\sum_{n=1}^\infty\frac{1}{2\pi \lambda}\int_{0}^{\infty}\mathrm dy\;\; \frac{\e^{-2\pi n y/\lambda +\I n\epsilon}}{n}\Big(&\log(1-e^{y+2\pi \I t})+\log(1-e^{-y+2\pi \I t})\\
    &+\frac{y}{1-e^{-y-2\pi \I t}}-\frac{y}{1-e^{y-2\pi \I t}}\Big)\;.
    \end{split}
\end{equation}
Performing a change of variables in each integral, and interchanging integral and summations again, we obtain
\begin{equation}
\begin{split}
    H(\lambda,t,\epsilon)
     =-\frac{1}{2\pi \lambda}\int_{0}^{\infty}\mathrm dy\;\; \e^{-2\pi  y/\lambda}\sum_{n=1}^\infty\frac{e^{in\epsilon}} {n^2}\Big(&\log(1-e^{y/n+2\pi \I t})+\log(1-e^{-y/n+2\pi \I t})\\
    &+\frac{y/n}{1-e^{-y/n-2\pi \I t}}-\frac{y/n}{1-e^{y/n-2\pi \I t}}\Big)\\
    \end{split}
\end{equation}
Letting $H(\lambda,t):=\lim_{\epsilon \to 0}H(\lambda,t,\epsilon)$, we get
\begin{equation}\label{Hterm}
    \begin{split}
H(\lambda,t) =-\frac{1}{2\pi \lambda}\int_{0}^{\infty}\mathrm dy\;\; \e^{-2\pi  y/\lambda}\sum_{n=1}^\infty\frac{1} {n^2}\Big(&\log(1-e^{y/n+2\pi \I t})+\log(1-e^{-y/n+2\pi \I t})\\
    &+\frac{y/n}{1-e^{-y/n-2\pi \I t}}-\frac{y/n}{1-e^{y/n-2\pi \I t}}\Big)\;.
\end{split}
\end{equation}
Finally, using that $-\frac{2\pi}{\lambda}e^{-2\pi y/\lambda}=\frac{\mathrm d}{\mathrm dy}e^{-2\pi y/\lambda}$ and integrating by parts yields
\begin{align*}
&H(\lambda,t) =
\Bigg[\frac{e^{-2\pi y/\lambda}}{(2\pi)^2}\sum_{n=1}^\infty\frac{1} {n^2}\Big(\log(1-e^{y/n+2\pi \I t})+\log(1-e^{-y/n+2\pi \I t})\\[-1ex]
&\hspace{5.5cm} +\frac{y/n}{1-e^{-y/n-2\pi \I t}}-\frac{y/n}{1-e^{y/n-2\pi \I t}}\Big)\Bigg]\Bigg|_{y=0}^{\infty}\\
&\hspace{1.5cm}-\int_{0}^{\infty}dy\;\; \frac{e^{-2\pi  y/\lambda}}{(2\pi)^2}\frac{\mathrm d}{\mathrm dy}\Bigg[\sum_{n=1}^\infty\frac{1} {n^2}\Big(\log(1-e^{y/n+2\pi \I t})+\log(1-e^{-y/n+2\pi \I t})\\[-1ex] &\hspace{5.5cm}+\frac{y/n}{1-e^{-y/n-2\pi \I t}}-\frac{y/n}{1-e^{y/n-2\pi \I t}}\Big)\Bigg]\,. \numberthis
\end{align*}
Using that the boundary term at $\infty$ vanishes, and interchanging the derivative with the sum, we obtain
\begin{equation}\label{p3}
    \begin{split}
H(\lambda,t) &=-\frac{2}{(2\pi)^2}\log(1-Q)\sum_{n=1}^{\infty}\frac{1}{n^2}+\int_{0}^{\infty}\mathrm dy\;\; \e^{-2\pi  y/\lambda}G(y,t)\\
    &=\frac{1}{2}\mathrm{Li}_1(Q)B_2+\int_{0}^{\infty}\mathrm dy\;\; \e^{-2\pi  y/\lambda}G(y,t)\;,
\end{split}
\end{equation}
where we used that $\sum_{n=1}^{\infty}\frac{1}{n^2}=\pi^2B_2$, $\mathrm{Li}_1(Q)=-\log(1-Q)$.\\

Hence, putting (\ref{p1}), (\ref{p2}) and (\ref{p3}) together gives us (\ref{F_0BorelS}).
\end{proof}

We finish this section with the following corollaries:

\begin{cor}  Let $l_k=\mathbb{R}_{<0}\cdot 2\pi \I(t+k)$, and let $\rho_k$ be a ray between $l_k$ and $l_{k-1}$. Then the following holds for $n \in \mathbb{Z}$:
\begin{equation}
    F_{\rho_{k-n}}(\lambda,t+n)= F_{\rho_k}(\lambda,t)\,. 
\end{equation}
In particular, if $0<\mathrm{Re}(t)<1$ and $\mathrm{Im}(t)>0$, we have on their common domains of definition
\begin{equation}
    F_{\rho_{-n}}(\lambda,t+n)=\log(G_3(t\,|\, \check\lambda,1))
\end{equation}
\end{cor}

\begin{proof}

Note that the labels $l_k$ (and hence also $\rho_k$) depend on $t$. In the following, we denote $l_k(t)$ and $\rho_k(t)$ to emphasize the $t$ dependence. In particular, we have the relations $l_k(t+n)=l_{k+n}(t)$ and $\rho_k(t+n)=\rho_{k+n}(t)$ for $n \in \mathbb{Z}$. \\

Using the fact that $G(\xi,t)=G(\xi,t+n)$ for any $n\in \mathbb{Z}$, we thus obtain
\begin{align*}
    F_{\rho_{k-n}} (\lambda,t+n) &=\frac{1}{\lambda^2}\mathrm{Li}_3(e^{2\pi \I(t+n)})+\frac{B_2}{2}\mathrm{Li}_1(e^{2\pi \I(t+n)})+ \int_{\rho_{k-n}(t+n)}\mathrm d\xi \;e^{-\xi/\check{\lambda}} G(\xi,t+n)\\
    &=\frac{1}{\lambda^2}\mathrm{Li}_3(e^{2\pi \I t})+\frac{B_2}{2}\mathrm{Li}_1(e^{2\pi \I t})+ \int_{\rho_{k-n}(t+n)}\mathrm d\xi \;e^{-\xi/\check{\lambda}} G(\xi,t)\\ &=\frac{1}{\lambda^2}\mathrm{Li}_3(e^{2\pi \I t})+\frac{B_2}{2}\mathrm{Li}_1(e^{2\pi \I t})+ \int_{\rho_{k}(t)}\mathrm d\xi \;e^{-\xi/\check{\lambda}} G(\xi,t)\\
    &=F_{\rho_k}(\lambda,t)\,. \numberthis
\end{align*}
The final result then follows from Proposition \ref{BorelsumR} and \ref{intrep2}.
\end{proof}

\begin{cor}\label{constmapBorelsum} Let $\rho$ be a ray different from $\pm i\mathbb{R}_{>0}$.
The Borel-transform $F_{0,\rho}(\la)$ of the formal series $F_0(\la)$ defined in (\ref{constmapterms})
can be represented in the form 
\begin{equation}
 F_{0,\rho}(\la)=-\frac{1}{\la^2}\zeta(3)+F_0^1-\int_\rho\mathrm d\xi \;e^{-\xi/\check{\lambda}}\left(G(\xi,0)+\frac{1}{12\xi}\right).  
\end{equation}
It is related to the Borel transform $F_\rho(\la,t)$ of $\tilde{F}(\la,t)$ by the equation \rf{Borelconstmap}.
\end{cor}
\begin{proof} We will consider a limit of $F_{\rho}(\lambda,t)$ as $t\to 0$, where $t$ is taken to satisfy $\text{Re}(t)>0$, $\text{Im}(t)>0$; and such that along the limit, $\rho$ is always between $l_{-1}$ and $l_0$ (resp.\ $-l_{-1}$ and $-l_0$) if $\rho$ is on the right (resp.\ left) Borel half-plane. Let us first assume that $\rho$ is taken on the right half plane. The limit requires some care as 
$G(\xi,0)$ has a simple pole with residue $-\frac{B_2}{2}=-\frac{1}{12}$ at $\xi=0$. We may observe, however, that
\[
F_{\rho}(\la,t)=\frac{1}{\lambda^2}\mathrm{Li}_3(e^{2\pi \I t})+\int_\rho\mathrm d\xi
\;\Big(e^{-\xi/\check{\lambda}}G(\xi,t)+\frac{1}{12}\frac{1}{e^{\xi-2\pi\I t}-1}\Big)
\]
has a well-defined limit for $t\ra 0$, as the poles at $\xi=0$ in the integrand cancel each other.
This limit can be represented as 
\begin{equation}\label{Frho-t0}
F_{\rho}(\la,0)=\frac{1}{\lambda^2}\zeta(3)+\int_\rho\mathrm d\xi
\;e^{-\xi/\check{\lambda}}\bigg(G(\xi,0)+\frac{1}{12\xi}
\bigg)
+\frac{1}{12}\int_\rho \mathrm d\xi\;\bigg(\frac{1}{e^\xi-1}-\frac{e^{-\xi/\check{\lambda}}}{\xi}\bigg).
\end{equation}
The derivative of the second integral with respect to $\check{\la}$ is easily found to be equal to $-\frac{1}{12}\frac{1}{\check{\la}}$.
It follows that this integral is equal to 
$-\frac{1}{12}\log{\check{\la}}+C$, with $C$ being an undetermined constant. 
The integrand of the first integral in \rf{Frho-t0}, on the other hand, is analytic at
$\xi=0$. By straightforward computation of the Taylor series one may check that
$G_0(\xi):=-G(\xi,0)-\frac{1}{12\xi}$ is the Borel transform of 
the formal series $F_0(\la) + \zeta(3)/\lambda^2-F_0^1$ (see Lemma \ref{constboreltrans} for the computation), so the relation (\ref{Borelconstmap}) then follows.
On the other hand, if $\rho$ is on the left half plane, one can apply the same argument from before by using the relation $F_{\rho}(\lambda,t)=F_{-\rho}(-\lambda,t)$ (see Lemma \ref{reflectionlemma} below).
\end{proof}

\subsection{Stokes phenomena of the Borel sum}\label{Stokesjumpssect}

In the previous section, we studied  $F_{\rho}(\lambda,t)$ for $\rho=\mathbb{R}_{>0}$. However, the ray $\mathbb{R}_{>0}$ is a choice, and any other ray $\rho$ that avoids the poles of $G(\xi,t)$ in principle is an equally valid choice to perform the Borel sum. In this section we study the dependence on this choice.

\begin{prop}\label{Stokesjumps} Assume that $\mathrm{Im}(t)> 0$ and for $k \in \mathbb{Z}$ let $l_k=\mathbb{R}_{<0}\cdot 2\pi \I(t+k)$. Furthermore let $\rho$ be a ray in the sector determined by the Stokes rays $l_{k+1}$ and $l_{k}$, and $\rho'$ a ray in the sector determined by $l_{k}$ and $l_{k-1}$. Then for $\lambda \in \mathbb{H}_{\rho}\cap \mathbb{H}_{\rho'}$ (resp.\ $\lambda \in \mathbb{H}_{-\rho}\cap \mathbb{H}_{-\rho'}$)  we have
\begin{equation}
F_{\pm \rho}(\lambda,t)
    -F_{\pm \rho'}(\lambda,t) =\frac{1}{2\pi i}\partial_{\check\lambda}\Big(\check{\lambda}\,\mathrm{Li}_2\big(e^{\pm 2\pi \I(t+k)/\check \lambda})\Big)\,.
\end{equation}
If $\mathrm{Im}(t)<0$, then the previous jumps also hold provided $\rho$ is interchanged with $\rho'$ in the above formulas.
\end{prop}

\begin{proof}

Notice that
\begin{equation}
    F_{\rho}(\lambda,t)-F_{\rho'}(\lambda, t)=\int_{\mathcal{H}(l_k)}\mathrm  d\xi\,e^{-\xi/\check \lambda}G(\xi,t)\,,
\end{equation}
where $\mathcal{H}(l_k)$ is a Hankel contour around $l_k =\mathbb{R}_{<0}\cdot 2\pi \I (t + k)$.\\

To compute this, notice that for our range of parameters, $G(\xi,t)$ has double poles at $\xi=2\pi \I m(t+k)$ for all $k\in\Z$ and $m\in \Z\setminus \{0\}$ with generalized residues
\begin{equation}
    \left.\frac{\mathrm d}{\mathrm d\xi}(e^{-\xi/\check \lambda}(\xi-2\pi \I m(t+k))^2G(\xi,t))\right|_{\xi=2\pi \I m(t+k)} =  -\frac{e^{-2\pi \I m (t + k)/\check \lambda}}{(2\pi)^2 m^2}\left(1+\frac{2\pi \I m(t+k)}{\check \lambda}\right).
\end{equation}
In particular, we have that
\begin{align*}
    \int_{\mathcal{H}(l_k)}\mathrm d\xi\,e^{-\xi/\check \lambda}G(\xi,t)&=2\pi \I \sum_{m=-1}^{-\infty}\left.\frac{\mathrm d}{\mathrm d\xi}(e^{-\xi/\check \lambda}(\xi-2\pi \mathrm im(t+k))^2G(\xi,t))\right|_{\xi=2\pi \I m(t+k)}\\
    &=-\frac{\I}{2\pi}\sum_{m=1}^{\infty}\frac{ e^{2\pi \I m (t + k)/\check \lambda}}{ m^2}\left(1-\frac{2\pi \I m(t+k)}{\check \lambda}\right)\\
    &=-\frac{\I }{2\pi} \Big(\mathrm{Li}_2\big(e^{2\pi \I(t+k)/\check \lambda}\big)+\frac{2\pi \I(t+k)}{\check \lambda}\log\big(1-e^{2\pi \mathrm i(t+k)/\check \lambda}\big)\Big)\\
    &=\frac{1}{2\pi \mathrm i}\partial_{\check\lambda}\Big(\check{\lambda}\mathrm{Li}_2\big(e^{ 2\pi \mathrm i(t+k)/\check \lambda})\Big)\;, \numberthis
    \end{align*}
where we have used the series representation of $\mathrm{Li}_s(z)$ for $s=1,2$; and the fact that for $\lambda$ in a sufficiently small sector containing $l_k$, we have $|e^{2\pi \mathrm i(t+k)/\check{\lambda}}|<1$.\\

A similar argument follows for the rest of the cases in the statement of the proposition.
\end{proof}

\subsubsection{Stokes jumps of the other Borel transform}
Recall the expression for the Borel transform which was obtained previously in \cite{Pasquetti:2009jg}, given in \eqref{diffBoreltrans}:
\begin{equation}
     G(\xi,t)=\sum_{k\in\Z} \frac{1}{ \xi}\left(\frac{(k-t)^2}{\xi^2} - \frac{e^{\xi/(k-t)}}{(e^{\xi/(k-t)} -1)^2}-\frac{1}{12}\right),\\
\end{equation}

Similarly to the previous discussion, if $y_m=\xi-2\pi \mathrm i m(t+k)$ near the pole given by $\xi=2\pi \mathrm i m (t+k)$, then by Taylor expanding the integrand near $y_m=0$ we obtain
\begin{equation}
   e^{-\xi/\check{\lambda}} G(\xi,t)= e^{-2\pi \mathrm i m(t+k)/\check{\lambda}} \left( \frac{1}{y_m^2} \frac{\mathrm i (t+k)}{2\pi m}+\frac{1}{y_m}\left( - \frac{1}{4\pi^2 m^2 }- \frac{\mathrm  i}{2\pi m\check{\lambda}} (t+k)\right) + \mathcal{O}(1)\right) \,.
\end{equation}
Hence, by following the argument of Proposition \ref{Stokesjumps}, we obtain the same Stokes jumps. 

\subsection{The limits  \texorpdfstring{$\lim_{k\to \pm \infty}F_{\rho_k}(\lambda,t)$}{TEXT} and \texorpdfstring{$F_{\mathrm{GV}}(\lambda,t)$}{Text}}
To finish the proof of Theorem \ref{theorem1}, we study the limits of $F_{\rho}(\lambda,t)$, discussed in point (iv). 

\begin{prop}\label{rblimit1} Let $\rho_k$ denote any ray between the Stokes rays $l_{k}$ and $l_{k-1}$. Furthermore, assume that $0<\mathrm{Re}(t)<1$, $\mathrm{Im}(t)>0$, $\mathrm{Re}(\lambda)>0$, $\mathrm{Im}(\lambda)<0$, and $\mathrm{Re}(t) < \mathrm{Re} (\check{\lambda}+1)$. Then
\begin{equation}
    \lim_{k\to \infty}F_{\rho_k}(\lambda,t)=\sum_{k=1}^\infty\frac{e^{2\pi \mathrm i k t}}{k\big(2\sin\big(\frac{\lambda k}{2}\big)\big)^2}=F_{\mathrm{GV}}(\lambda,t)\,.
\end{equation}

Furthermore, we can write the sum of the Stokes jumps along $l_k$ for $k\geq 0$ as
\begin{equation}
    \sum_{k=0}^{\infty}\phi_{l_k}(\lambda,t)=\frac{1}{2\pi i }\partial_{\lambda}\Big(\lambda \sum_{l=1}^{\infty}\frac{w^l}{l^2(1-\widetilde{q}^l)}\Big), \;\;\;\; w:=e^{2\pi \mathrm it/\check\lambda},\;\; \widetilde{q}:=e^{2\pi \mathrm i/\check{\lambda}}\;.
\end{equation}

\end{prop}

\begin{proof}
By Proposition \ref{Stokesjumps} we find that
\begin{equation}
    F_{\rho_k}-F_{\rho_{k+1}}=\frac{\mathrm i}{2\pi} \Big(\mathrm{Li}_2\big(e^{2\pi \mathrm i(t+k)/\check \lambda}\big)+\log(e^{2\pi \mathrm i\frac{t+k}{\check \lambda }})\log\big(1-e^{2\pi \mathrm i(t+k)/\check \lambda}\big)\Big)\,.
\end{equation}
Denoting $w=e^{2\pi \mathrm it/\check{\lambda}}$ and $\widetilde{q}=e^{2\pi \mathrm i/\check{\lambda}}$, we find
\begin{equation}
    \begin{split}
    F_{\rho_0}(\lambda,t)-\lim_{k\to \infty}F_{\rho_k}(\lambda,t)&=\sum_{k=0}^{\infty}F_{\rho_k}(\lambda,t)-F_{\rho_{k+1}}(\lambda,t)\\
    &=\frac{\mathrm i}{2\pi}\sum_{k=0}^{\infty}\Big(\mathrm{Li}_2\big(w\widetilde{q}^k\big)+\log\big(w\widetilde{q}^k\big)\log\big(1-w\widetilde{q}^k\big)\Big)\,.
    \end{split}
\end{equation}
We now use the following identities:
\begin{subequations}
\begin{align}
&\sum_{k=0}^\infty \log(1-w\widetilde{q}^k)=-\sum_{l=1}^\infty\frac{1}{l}\frac{w^l}{1-\widetilde{q}^l} \label{id1}\,,\\
&\sum_{k=0}^\infty k\log(1-w\widetilde{q}^k)=-\sum_{l=1}^\infty\frac{\widetilde{q}^l}{l}\frac{w^l}{(1-\widetilde{q}^l)^2}
\label{id2}\,,\\
&\sum_{k=0}^\infty \mathrm{Li}_2(w\widetilde{q}^k)=\sum_{l=1}^\infty\frac{1}{l^2}\frac{w^l}{1-\widetilde{q}^l}\,.
\label{id3}
\end{align}
\end{subequations}
The first two identities (\ref{id1}) and (\ref{id2}) are easily established by 
using the Taylor expansion of the logarithm function; using that $|\widetilde{q}|<1$ and $|w|<1$; and exchanging the two 
summations. In order to verify (\ref{id3}), one can first act on it with $w\frac{\mathrm d}{\mathrm dw}$. 
The left side of the resulting equation is easily seen to be equal to 
\[
-\sum_{k=0}^\infty\log(1-w\widetilde{q}^k)=\sum_{l=1}^\infty\frac{1}{l}\frac{w^l}{1-\widetilde{q}^l}\,,
\]
using (\ref{id1}). It follows that (\ref{id3}) holds up to addition of a term which is 
constant with respect to  $w$.  In order to fix this freedom, it 
suffices to note that (\ref{id3}) holds for $w=0$.\\

Using the previous identities, we obtain
\begin{equation}
    \begin{split}
        F_{\rho_0}(\lambda,t)-\lim_{k\to \infty}F_{\rho_k}(\lambda,t)&=\frac{\mathrm i}{2\pi}\sum_{l=1}^\infty\frac{w^l}{l(1-\widetilde{q}^l)}\bigg(\frac{1}{l}-\frac{\widetilde{q}^l\log \widetilde{q}}{1-\widetilde{q}^l}-\log w\bigg)\\
        &=-\frac{\mathrm i}{2\pi}\sum_{l=1}^\infty\frac{\partial}{\partial l}\bigg(\frac{w^l}{l(1-\widetilde{q}^l)}\bigg)\,.
    \end{split}
\end{equation}
Now notice that under our assumptions on $t$ and $\lambda$, we have $F_{\rho_0}=F_{\mathrm{np}}$ by Proposition \ref{BorelsumR}. We now show that $F_{\rho_0}$ admits the following representation as sum over residues:
\begin{equation}
 F_{\rho_0}(\lambda,t)=\frac{1}{2\pi \mathrm i} \sum_{l=1}^\infty \frac{\partial}{\partial l}\bigg(\frac{w^l}{l(1-\widetilde{q}^l)}\bigg)
 +\sum_{k=1}^\infty\frac{e^{2\pi \mathrm i k t}}{k\big(2\sin\big(\frac{\lambda k}{2}\big)\big)^2}\,.
\end{equation}
In order to see this, let us recall that by Proposition \ref{intrep2}, we have
\begin{align*}
F_{\mathrm{np}}(\lambda,t)&=-\int_{\mathbb{R}+\mathrm i0^+}\frac{\mathrm du}{8u}\,\frac{e^{u(t-\frac{1}{2})}}{\sinh(u/2)(\sinh(\lambda  u/4\pi))^2}=
\int_{\mathbb{R}+\mathrm i0^+}\frac{\mathrm du}{u}\,\frac{e^{ut}}{1-e^u}\frac{1}{(2\sinh(\lambda u/4\pi))^2}\,.
\end{align*}
The integrand  has two series of poles, one 
at $u=u_l:=(2\pi)^2\frac{\mathrm i}{\lambda}l$, $l\in \mathbb{Z}$ and the other at
$u=\widetilde{u}_k:=2\pi \mathrm i k$, $k\in \mathbb{Z}$.  We can compute the previous integral by closing the contour in the upper
half-plane. The contributions from  the poles at $u=u_l$ are calculated as 
\[
\begin{aligned}
2\pi \mathrm i \left(\frac{4\pi}{2\lambda}\right)^2
\frac{\partial}{\partial u}\frac{e^{ut}}{(1-e^u)u}\bigg|_{u=(2\pi)^2\frac{\mathrm i}{\lambda}l}
&=2\pi \mathrm i \,\frac{(2\pi)^2}{\lambda^2}
\bigg(\frac{\lambda}{\mathrm i(2\pi)^2}\bigg)^2\frac{\partial}{\partial l}\frac{w^l}{(1-\widetilde{q}^l)l}\\
&=\frac{1}{2\pi \mathrm i}
\frac{\partial}{\partial l}\frac{w^l}{(1-\widetilde{q}^l)l}\,,
\end{aligned}
\]
while the contributions of the poles at $u=2\pi \mathrm i k$ give the remaining term.\\

In particular, we conclude that
\begin{align*}
    \lim_{k\to \infty}F_{\rho_k}(\lambda,t)&=\lim_{k\to \infty}(F_{\rho_k}(\lambda,t)-F_{\rho_0}(\lambda,t))+F_{\rho_0}(\lambda,t)\\
    &=-\frac{1}{2\pi \mathrm i}\sum_{l=1}^\infty\frac{\partial}{\partial l}\bigg(\frac{w^l}{l(1-\widetilde{q}^l)}\bigg) +\frac{1}{2\pi \mathrm i} \sum_{l=1}^\infty \frac{\partial}{\partial l}\bigg(\frac{w^l}{l(1-\widetilde{q}^l)}\bigg)
 +\sum_{k=1}^\infty\frac{e^{2\pi \mathrm i k t}}{k\big(2\sin\big(\frac{\lambda k}{2}\big)\big)^2}\\
 &=\sum_{k=1}^\infty\frac{e^{2\pi \mathrm i k t}}{k\big(2\sin\big(\frac{\lambda k}{2}\big)\big)^2}\,. \numberthis
\end{align*}
The last statement follows easily by noticing that 
\begin{equation}
    \frac{\mathrm i}{2\pi}\sum_{l=1}^\infty\frac{\partial}{\partial l}\bigg(\frac{w^l}{l(1-\widetilde{q}^l)}\bigg)=\frac{1}{2\pi \mathrm i }\partial_{\lambda}\Big(\lambda \sum_{l=1}^{\infty}\frac{w^l}{l^2(1-\widetilde{q}^l)}\Big)\,.
\end{equation}
\end{proof}

To study the other limits to the imaginary rays, we use the following lemma:

\begin{lem}\label{reflectionlemma}
For $\rho$ any ray not in $\{\pm l_k\}\cup \{\pm l_{\infty}\}$ and $\lambda \in \mathbb{H}_\rho$, we have
\begin{equation}
    F_{\rho}(\lambda,t)=F_{-\rho}(-\lambda,t)\,.
\end{equation}
\end{lem}
\begin{proof}
The main thing to notice is that
\begin{equation}
    G(\xi,t)=-G(-\xi,t)\,.
\end{equation}
Using the earlier relation, we obtain
\begin{equation}
    \begin{split}
    F_{\rho}(\lambda,t)&=\frac{1}{\lambda^2} \mathrm{Li}_{3}(Q)+ \frac{B_2}{2}\mathrm{Li}_{1}(Q)-\int_{\rho}\mathrm  d\xi\, e^{-\xi/\check{\lambda}} G(-\xi,t)\\
    &=\frac{1}{\lambda^2} \mathrm{Li}_{3}(Q)+ \frac{B_2}{2}\mathrm{Li}_{1}(Q)+\int_{-\rho}\mathrm  d\xi\, e^{\xi/\check{\lambda}} G(\xi,t)\\
    &=F_{-\rho}(-\lambda,t)\,.
    \end{split}
\end{equation}
\end{proof}
As an immediate corollary from Proposition \ref{rblimit1} and Lemma \ref{reflectionlemma}, we obtain:

\begin{cor}\label{lblimit2} With the same notation as in Proposition \ref{rblimit1}, assume that $0<\text{Re}(t)<1$, $\mathrm{Im}(t)>0$, $\mathrm{Re}(\lambda)<0$, $\mathrm{Im}(\lambda)>0$ and $\mathrm{Re}( t) < \mathrm{Re} (-\check{\lambda}+1)$. Then
\begin{equation}
    \lim_{k\to \infty}F_{-\rho_k}(\lambda,t)=F_{\mathrm{GV}}(\lambda,t)\,.
\end{equation}
\end{cor}

\begin{prop}\label{lblimit1}
With the same notation as in Proposition \ref{rblimit1}, assume that $0<\text{Re}(t)<1$, $\text{Im}(t)>0$, $\text{Re}(\lambda)<0$, $\text{Im}(\lambda)<0$, $\mathrm{Re}\, t < \mathrm{Re} (-\check{\lambda}+1)$ and that $|w^{-1}|<1$. Then 
\begin{equation}
    \begin{split}
\lim_{k\to -\infty}F_{-\rho_k}(\lambda,t)=F_{\mathrm{GV}}(\lambda,t)
    \end{split}
\end{equation}
\end{prop}
\begin{rem}
For fixed $\lambda$ satisfying $\mathrm{Re}(\lambda)<0$, $\mathrm{Im}(\lambda)<0$, the condition $|w^{-1}|<1$ can be satisfied by picking $t$ such that $0<\mathrm{Re}(t)<1$, $\mathrm{Im}(t)>0$, and $\mathrm{Im}(t)$ is sufficiently large. Similarly, for fixed $t$ with $0<\mathrm{Re}(t)<1$, $\mathrm{Im}(t)>0$, $|w^{-1}|<1$ can be satisfied by picking $\lambda$ such that $\mathrm{Re}(\lambda)<0$, $\mathrm{Im}(\lambda)<0$ and $|\mathrm{Im}( \lambda)|<<|\mathrm{Re}(\lambda)|$ .
\end{rem}

\begin{proof}
Using the jumps along the Stokes rays $l_{k}$ for $k<0$, we find that
\begin{align*}
    F_{-\rho_0}-\lim_{k\to -\infty}F_{-\rho_k}&=\sum_{k=-1}^{-\infty}F_{-\rho_{k+1}}-F_{-\rho_{k}}\\
    &=-\frac{\mathrm i}{2\pi}\sum_{k=-1}^{-\infty} \Big(\mathrm{Li}_2\big(e^{-2\pi \mathrm i(t+k)/\check \lambda}\big)+\log(e^{-2\pi \mathrm i\frac{t+k}{\check \lambda }})\log\big(1-e^{-2\pi \mathrm i(t+k)/\check \lambda}\big)\Big)\\
    &=-\frac{\mathrm i}{2\pi}\sum_{k=0}^{\infty} \Big(\mathrm{Li}_2\big(w^{-1}\widetilde{q}^k\big)+\log(w^{-1}\widetilde{q}^k)\log\big(1-w^{-1}\widetilde{q}^k\big)\Big)\\
    &\;\;\;\;+\frac{\mathrm i}{2\pi}\Big(\mathrm{Li}_2\big(w^{-1}\big)+\log(w^{-1})\log\big(1-w^{-1}\big)\Big)\,. \numberthis
\end{align*}
Using the constraints on $t$ and $\lambda$, we find that $|w^{-1}|<1$ and $|\widetilde{q}|<1$, so that we can  expand in series as in Proposition \ref{rblimit1} and write
\begin{align*}
    F_{-\rho_0}-\lim_{k\to -\infty}F_{-\rho_k}&=-\frac{\mathrm i}{2\pi}\sum_{l=1}^\infty\frac{w^{-l}}{l(1-\widetilde{q}^l)}\bigg(\frac{1}{l}-\frac{\widetilde{q}^l\log \widetilde{q}}{1-\widetilde{q}^l}-\log w^{-1}\bigg)\\
     &\;\;\;\;+\frac{\mathrm i}{2\pi}\Big(\mathrm{Li}_2\big(w^{-1}\big)+\log(w^{-1})\log\big(1-w^{-1})\Big)\\ \numberthis
     &=-\frac{1}{2\pi \mathrm i}\sum_{l=1}^\infty\frac{\partial}{\partial l}\bigg(\frac{w^{-l}}{l(1-\widetilde{q}^l)}\bigg)+\frac{\mathrm  i}{2\pi}\Big(\mathrm{Li}_2\big(w^{-1}\big)+\log(w^{-1})\log\big(1-w^{-1})\Big)\,.
\end{align*}

On the other hand, under our conditions on the parameters $t$ and $\lambda$, and by Lemma \ref{reflectionlemma},  we have that 
\begin{equation}
    F_{-\rho_0}(\lambda,t)=F_{\mathrm{np}}(-\lambda,t)\,.
\end{equation}
Following the same argument as in Proposition \ref{rblimit1} using the integral representation of $F_{\mathrm{np}}$, we find that
\[
\begin{aligned}
 F_{\mathrm{np}}(-\lambda,t)=\frac{1}{2\pi \mathrm i} \sum_{l=1}^\infty \frac{\partial}{\partial l}\bigg(\frac{w^{-l}}{l(1-\widetilde{q}^{-l})}\bigg)
 +\sum_{k=1}^\infty\frac{e^{2\pi \mathrm i k t}}{k\big(2\sin\big(\frac{\lambda k}{2}\big)\big)^2}.
\end{aligned}
\]
 Hence, we find that for $\lambda$ and $t$ as in the hypothesis 
\begin{equation}
    F_{-\rho_0}(\lambda,t)=\frac{1}{2\pi \mathrm i} \sum_{l=1}^\infty \frac{\partial}{\partial l}\bigg(\frac{w^{-l}}{l(1-\widetilde{q}^{-l})}\bigg)
 + F_{\mathrm{GV}}(\lambda,t)
\end{equation}
Joining our results together, we conclude that
\begin{equation}
    \begin{split}
    \lim_{k\to -\infty}F_{-\rho_k}(\lambda,t)&=\frac{1}{2\pi \mathrm i}\sum_{l=1}^{\infty}\Big(\frac{\partial}{\partial l}\bigg(\frac{w^{-l}}{l(1-\widetilde{q}^{-l})}\bigg) + \frac{\partial}{\partial l}\bigg(\frac{w^{-l}}{l(1-\widetilde{q}^{l})}\bigg)\Big)+F_{\mathrm{GV}}(\lambda,t)\\
    &\;\;\;\;+\frac{1}{2\pi \mathrm i}\Big(\mathrm{Li}_2\big(w^{-1}\big)+\log(w^{-1})\log\big(1-w^{-1})\Big)\,.
    \end{split}
\end{equation}
Finally, notice that
\begin{equation}
    \begin{split}
    \sum_{l=1}^{\infty}\Big(\frac{\partial}{\partial l}\bigg(\frac{w^{-l}}{l(1-\widetilde{q}^{-l})} +\frac{w^{-l}}{l(1-\widetilde{q}^{l})}\bigg)\Big)&=\sum_{l=1}^{\infty}\Big(\frac{\partial}{\partial l}\bigg(\frac{w^{-l}}{l}\bigg)\Big)\\
    &=\log(w^{-1})\sum_{l=1}^{\infty}\frac{w^{-l}}{l}-\sum_{l=1}^{\infty}\frac{w^{-l}}{l^2}\\
    &=-(\mathrm{Li}_2(w^{-1})+\log(w^{-1})\log(1-w^{-1}))\,,\\
    \end{split}
\end{equation}
where in the last equality we used that $|w^{-1}|<1$ under our hypotheses. Hence, we conclude that 
\begin{equation}
    \lim_{k\to -\infty}F_{-\rho_k}(\lambda,t)=F_{\mathrm{GV}}(\lambda,t)\,.
\end{equation}

\end{proof}

By using Lemma \ref{reflectionlemma} and Proposition \ref{lblimit1}, we get the following immediate corollary:

\begin{cor} With the same notation as in Proposition \ref{rblimit1}, assume that $0<\mathrm{Re}(t)<1$, $\mathrm{Im}(t)>0$, $\mathrm{Re}(\lambda)>0$, $\mathrm{Im}(\lambda)>0$, $\mathrm{Re}(t)<\mathrm{Re}(\check{\lambda}+1)$ and such that $|w|<1$. Then
\begin{equation}
    \lim_{k\to -\infty}F_{\rho_k}(\lambda,t)=F_{\mathrm{GV}}(\lambda,t)\,.
\end{equation}
\end{cor}
The limits studied above can be informally 
interpreted as the relation between  $F_{\mathrm{GV}}(\lambda,t)$ 
and the Borel summations along the imaginary axes. 

\subsubsection{Relation between $F_{\mathbb{R}_{>0}}$ and $ F_{\mathrm{coni}}^{\mathrm{resum}}$}\label{FHOrel}

In this subsection, we briefly explain how $F_{\mathbb{R}_{>0}}=F_{np}$ relates to $ F_{\mathrm{coni}}^{\mathrm{resum}}$ from \cite{HO}.\\

On the one hand, from part (iv) of Theorem \ref{theorem1} together with the comments of Section \ref{connST}, we find that
\begin{equation}
    F_{\mathbb{R}_{>0}}(\lambda,t)=F_{\mathrm{GV}}(\lambda,t)+\frac{1}{2\pi i}\frac{\partial}{\partial \lambda }\lambda F_{NS}\Big(\frac{4\pi^2}{\lambda},\frac{2\pi}{\lambda}\Big(t-\frac{1}{2}\Big)\Big)\,.
\end{equation}
On the other hand, in \cite{HO} the following function is considered
\begin{equation}\label{FHO}
     F_{\mathrm{coni}}^{\mathrm{resum}}(\lambda,t)=\frac{\mathrm{Li}_3(Q)}{\lambda^2}+ \int_{0}^{\infty}dv \, \frac{v}{1-e^{2\pi v-i0^+}}\log(1+Q^2-2Q\cosh(\lambda v)), \;\;\;\; Q=e^{2\pi it}\,,
\end{equation}
and it is conjectured that
\begin{equation}
     F_{\mathrm{coni}}^{\mathrm{resum}}(\lambda,t)=F_{\mathrm{GV}}(\lambda,t)+\frac{1}{2\pi i}\frac{\partial}{\partial \lambda }\lambda F_{NS}\Big(\frac{4\pi^2}{\lambda},\frac{2\pi}{\lambda}\Big(t-\frac{1}{2}\Big)\Big)\,,
\end{equation}
as explained in Section \ref{prevres}.\\

We show that this in indeed the case, by the use of the Woronowicz form of $F_{\mathbb{R}_{>0}}$ of Lemma \ref{woronowiczlemma}.

\begin{prop} 
Let $t\in \mathbb{C}$ be such that $0<\mathrm{Re}(t)<1$, $\mathrm{Im}(t)>0$, and let $\lambda $ be in the sector determined by $l_0=\mathbb{R}_{<0}\cdot 2\pi \I t$ and $l_{-1}=\mathbb{R}_{<0}\cdot 2\pi \I (t-1)$. Then $F_{\mathbb{R}_{>0}}=F_{\mathrm{coni}}^{\mathrm{resum}}$ on their common domains of definition. 
\end{prop}

\begin{proof}
First notice that since
\begin{equation}
    1+Q^2-2Q\cosh(\lambda x)=(1-e^{\lambda x}Q)(1-e^{-\lambda x}Q)\,,
\end{equation}
we can rewrite (\ref{FHO}) as follows
\begin{equation}
    \begin{split}
    F_{\mathrm{coni}}^{\mathrm{resum}}&=\frac{\mathrm{Li}_3(Q)}{\lambda^2}+\int_{\mathbb{R}+\mathrm i0^+}\mathrm dv \frac{v}{1-e^{2\pi v}}\log(1-e^{\lambda v}Q)\\
    &\quad -\int_{-\infty}^0\mathrm dv \frac{v}{1-e^{2\pi v-\mathrm i0^+}}\log(1-e^{\lambda v}Q)+\int_{0}^{\infty}\mathrm dv \frac{v}{1-e^{2\pi v-\mathrm i0^+}}\log(1-e^{-\lambda v}Q)\\
    &=\frac{\mathrm{Li}_3(Q)}{\lambda^2}+\int_{\mathbb{R}+\mathrm i0^+}\mathrm dv \frac{v}{1-e^{2\pi v}}\log(1-e^{\lambda v}Q) +\int_{0}^{\infty}\mathrm dv \; v\log(1-e^{-\lambda v}Q)\,.
    \end{split}
\end{equation}

On the other hand, notice that we can rewrite the last term in the above expression as follows:
\begin{align*} 
    \int_{0}^{\infty}\mathrm dv \; v\log(1-e^{-\lambda v}Q)&=\frac{1}{\lambda^2}\int_{\lambda \cdot \mathbb{R}_{>0}}\mathrm dv \; v\log(1-e^{-v}Q)\\
    &=\frac{1}{\lambda^2}\int_{0}^{\infty}\mathrm dv \; v\log(1-e^{- v}Q)\\
    &=\frac{1}{2\lambda^2}\int_{0}^{\infty}\mathrm dv\,\frac{v^2}{1-e^{v}Q^{-1}}=-\frac{1}{\lambda^2}\mathrm{Li}_3(Q) \numberthis
\end{align*}
where in the second equality we have used that the range of $\lambda$ allows us to deform back the contour to $\mathbb{R}_{>0}$; in the third equality we have integrated by parts; and in the last one we have used that $\text{Im}(t)>0$ implies that $|e^{2\pi\mathrm  it}|<1$, and hence we can use the integral representation of $\mathrm{Li}_3$.\\

Hence,
\begin{equation}
    F_{\mathrm{coni}}^{\mathrm{resum}}=\int_{\mathbb{R}+\mathrm i0^+}\mathrm dv \frac{v}{1-e^{2\pi v}}\log(1-e^{\lambda v}Q)=\frac{1}{(2\pi)^2}\int_{\mathbb{R}+\mathrm i0^+}\mathrm dv \frac{v}{1-e^{ v}}\log(1-e^{\frac{\lambda}{2\pi} v+2\pi\mathrm  it})
\end{equation}
so the result follows from Lemma \ref{woronowiczlemma} and Proposition \ref{BorelsumR}.
\end{proof}


\section{Relation to the Riemann--Hilbert problem and line bundles}
\label{hyperholo}

In Sections \ref{mainresults} and \ref{proofs} we discussed the Borel sum $F_{\rho}(\lambda,t)$ of $\widetilde{F}(\lambda,t)$ along the ray $\rho$, and its dependence on $\rho$ in terms of the Stokes jumps. Our objectives in this section are the following:

\begin{itemize}
    \item On one hand, in \cite{BridgelandCon} a Riemann--Hilbert problem is associated to the BPS spectrum of the resolved conifold. This involves finding piecewise holomorphic functions $\CX_{\ga}$ on $\BC^{\times}\times M$, with $M$ being called the space of stability structures,
    related by certain Stokes jumps along rays in $\BC^\times$. Introducing a coordinate $\lambda_B$ for $\BC^\times$ called
    twistor variable one may interpret the family of functions $\CX_{\ga}(\lambda_B,-)$ on $M$ as complex coordinates
    defining a family of complex structures on $M$.
    We will show that the jumps of $F_{\rho}(\lambda,t)$ serve as ``potentials" for the Stokes jumps associated to the Riemann--Hilbert problem. 
    
    \item On the other hand, in thinking of $F_{\rho}(\lambda,t)$ more geometrically, it is natural to consider the partition functions 
    \begin{equation}
        Z_{\rho}(\lambda,t)=\exp(F_{\rho}(\lambda,t))\;,
    \end{equation}
    and interpret them as defining a section of a line bundle $\mathcal{L}$, having transition functions equal to the exponentials of the Stokes jumps. This perspective follows the ideas of \cite{CLT20}, specialized to the case of the resolved conifold.
    
     \item We will furthermore demonstrate that the
    the line bundle $\mathcal{L}$ is related to certain hyperholomorphic line bundles previously considered in \cite{Neitzke_hyperhol,APP}. These hyperholomorphic line bundles are canonically defined by 
    a given BPS spectrum and represented by transition functions 
    defined from the Rogers dilogarithm function. We will show that the hyperholomorphic line bundles
    considered in \cite{Neitzke_hyperhol,APP} are in the case of  the resolved conifold related to the line bundle $\mathcal{L}$ by performing a  certain ``conformal limit" previously considered in \cite{Gaiotto:2014bza}.
\end{itemize}

   In order to facilitate comparison with \cite{BridgelandCon}, we will represent the parameters $\lambda$ and $t$
   used in this paper in the following form
    \begin{equation}
         t=v/w, \;\;\;\; \lambda = 2\pi \lambda_{\mathrm B}/w\,,
    \end{equation}
    where $\lambda_{\mathrm B}$ is the notation used here for the variable 
    denoted $t$ in \cite{BridgelandCon},
    and consider a projectivized partition function 
    \begin{equation}
        Z_{\rho}(\lambda_{\mathrm B},v,w):=\exp\Big(F_{w^{-1}\cdot \rho}\Big(\frac{2\pi \lambda_{\mathrm B}}{w},\frac{v}{w}\Big)\Big)\,.
    \end{equation}
    We will show that after appropriately normalizing the partition functions $Z_{\rho} \to \widehat{Z}_{\rho}$, the BPS spectrum of the resolved conifold will be neatly encoded in the transition functions of the line bundle defined by $\widehat{Z}_{\rho}$. Furthermore, we will see that the normalization reintroduces the constant map constribution \ref{constmapterms}, giving a partition function whose free energy has asymptotic expansion equal to \ref{freeenergydecomp}.

\subsection{Bridgeland's Riemann--Hilbert problem and its solution}\label{RHproblem}

We begin by recalling the Riemann--Hilbert problem considered in \cite{BridgelandCon}. 
The initial data for such Riemann--Hilbert problems is the following:

\begin{dfn}\label{varBPS} A variation of BPS structures is given by a tuple $(M,\Gamma,Z,\Omega)$, where 

\begin{itemize}
    \item $M$ is a complex manifold. 
    \item Charge lattice: $\Gamma \to M$ is a local system of lattices with a skew-symmetric, covariantly constant paring $\langle - , - \rangle\colon  \Gamma \times \Gamma \to \mathbb{Z}$. 
    \item Central charge: $Z$ is a holomorphic section of $\Gamma^*\otimes \mathbb{C} \to M$. 
    \item BPS indices: $\Omega \colon  \Gamma \to \mathbb{Z}$ is a function satisfying $\Omega(\gamma)=\Omega(-\gamma)$ and the Kontsevich--Soibelman wall-crossing formula \cite{KS, BridgelandDT}. 
\end{itemize}

The tuple $(M,\Gamma,Z,\Omega)$ should also satisfy the following conditions:

\begin{itemize}
\item Support property: Let $\text{Supp}(\Omega):=\{\gamma \in \Gamma \;\; | \;\; \Omega(\gamma)\neq 0\}$. Given a compact set $K\subset M$ and a choice of covariantly constant norm $|\cdot |$ on $\Gamma|_K \otimes_{\mathbb{Z}}\mathbb{R} $, there is a constant $C>0$ such that for any  $\text{Supp}(\Omega)\cap \Gamma|_K$:
    \begin{equation} \label{supportproperty}
        |Z_{\gamma}|>C|\gamma| \,.
    \end{equation}
    
    \item Convergence property: for each $p\in M$, there is an $R>0$ such that
    \begin{equation}
        \sum_{\gamma \in \Gamma_p}|\Omega(\gamma)|e^{-R|Z_{\gamma}|}<\infty \,.
    \end{equation}
\end{itemize}
\end{dfn}

The variation of BPS structures associated to the resolved conifold is then taken to be the tuple $(M,\Gamma,Z,\Omega)$, where:
\begin{itemize}
    \item $M$ is the complex $2$-dimensional manifold
    \begin{equation}
        M:=\{(v,w)\in \mathbb{C}^2 \;\; | \;\; w\neq 0, \;\;\; v+nw\neq 0 \;\;\; \text{for all} \;\;\; n \in \mathbb{Z}\}\,.
    \end{equation}
    \item $\Gamma\to M$ is given by the trivial local system
    \begin{equation}
        \Gamma=\mathbb{Z}\cdot \delta \oplus \mathbb{Z}\cdot \beta \oplus \mathbb{Z} \cdot \delta^{\vee} \oplus \mathbb{Z} \cdot \beta^{\vee},
    \end{equation}
    with pairing defined such that $(\beta^{\vee},\beta,\delta^{\vee},\delta)$ is a Darboux frame. Namely,
    \begin{equation}
        \langle \delta^{\vee},\delta \rangle= \langle \beta^{\vee},\beta \rangle=1\,,
    \end{equation}
    with all other pairings equal to $0$.
    \item If for $\gamma \in \Gamma$, we denote $Z_{\gamma}:=Z(\gamma)$, $Z$ is defined by
    \begin{equation}
        Z_{n\beta +m\delta +p\beta^{\vee}+q\delta^{\vee}}=2\pi \I(nv+mw), \;\;\;\; \text{for} \;\;\; n,m,p,q\in \mathbb{Z}\,.
    \end{equation}
    \item $\Omega$ is given by the BPS spectrum of the resolved conifold \cite{JS}, see also \cite{Banerjee:2019apt}:
    \begin{equation}
    \Omega(\gamma) =
  \begin{cases}
    1       & \quad \text{if } \gamma = \pm \beta + n \delta\quad \text{for }\;\; n\in \mathbb{Z}\,,\\
    -2  & \quad \text{if } \gamma= k \delta \quad \text{for} \;\;k\in \mathbb{Z}\setminus \{ 0\}\,,\\
    0 & \quad \text{otherwise}.
  \end{cases}
\end{equation}
\end{itemize}

To this data, the following Riemann--Hilbert problem is associated\footnote{We will follow slightly different conventions from \cite{BridgelandCon}. In particular, what we call $\mathcal{L}_k$ corresponds in Bridgeland's convention to $-\mathcal{L}_k$.}. First, we define $\mathcal{L}_k:=\mathbb{R}_{<0}\cdot 2\pi \I(v+kw)$ and $\mathcal{L}_{\infty}:=\mathbb{R}_{<0}\cdot2 \pi \I w$, and assume that $(v,w)\in M$ satisfies $\text{Im}(v/w)>0$ (the case $\text{Im}(v/w)\leq 0$ is also considered in \cite{BridgelandCon}, but we restrict to $\text{Im}(v/w)>0$ for simplicity). Then, for each ray $\rho$ from $0$ to $\infty$ not in $\{\pm \mathcal{L}_k\}_{k\in \mathbb{Z}}\cup \{\pm \mathcal{L}_{\infty}\}$, we should find a holomorphic function $\mathcal{X}_{\gamma,\rho}(v,w,-)\colon \mathbb{H}_{\rho} \to \mathbb{C}^{\times}$ labeled by $\gamma \in \Gamma$ such that they satisfy the following:

\begin{itemize}
    \item Twisted homomorphism property: for $\gamma, \gamma'\in \Gamma$ we have
    \begin{equation}
        \mathcal{X}_{\gamma+\gamma',\rho}(v,w,-)=(-1)^{\langle \gamma, \gamma' \rangle}\mathcal{X}_{\gamma,\rho}(v,w,-)\mathcal{X}_{\gamma',\rho}(v,w,-)\,.
    \end{equation}
    \item Stokes jumps: we denote by $\rho_k$ a ray between $\mathcal{L}_k$ and $\mathcal{L}_{k-1}$. We then have
        \begin{equation}\label{SJ1}
        \mathcal{X}_{\gamma, \pm \rho_{k+1}}(v,w,\lambda_{\mathrm B})=\mathcal{X}_{\gamma,\pm \rho_{k}}(v,w,\lambda_{\mathrm B})(1-\mathcal{X}_{\pm(\beta+k\delta)}(v,w,\lambda_{\mathrm B}))^{ \langle\gamma, \pm(\beta +k\delta) \rangle\Omega(\beta+k\delta)}\,.
    \end{equation}
    On the other hand, consider $\rho_{k_1}$ and $\rho_{k_2}$ with $k_1\neq k_2$, and let $[\rho_{k_1},-\rho_{k_2}]$ denote the smallest of the two sectors determined by $\rho_{k_1}$ and $-\rho_{k_2}$. In the case $\mathcal{L}_{\infty} \subset [\rho_{k_1},-\rho_{k_2}]$, for $\lambda_{\mathrm B}$ in the corresponding common domains we have
            \begin{align}\label{SJ2}
        \mathcal{X}_{\gamma, -\rho_{k_2}}=\mathcal{X}_{\gamma, \rho_{k_1}}\cdot \Bigg(\prod_{k\geq k_1}&(1-\mathcal{X}_{ \beta+k\delta})^{ \langle\gamma, \beta +k\delta \rangle\Omega(\beta+k\delta)}\prod_{k>-k_2}(1-\mathcal{X}_{-\beta+k\delta})^{ \langle\gamma, -\beta +k\delta \rangle\Omega(\beta-k\delta)}\notag\\
        &\prod_{k\geq 1}(1-\mathcal{X}_{ k\delta})^{ \langle \gamma, k\delta \rangle \Omega(k\delta)}\Bigg)\;,
    \end{align}
     while if $-\mathcal{L}_{\infty} \subset [\rho_{k_1},-\rho_{k_2}]$, we have
    \begin{align}\label{SJ3}
        \mathcal{X}_{\gamma, \rho_{k_1}}=\mathcal{X}_{\gamma, -\rho_{k_2}}\cdot \Bigg(\prod_{k\geq k_2}&(1-\mathcal{X}_{ -\beta-k\delta})^{ \langle\gamma, -\beta -k\delta \rangle\Omega(\beta+k\delta)}\prod_{k>-k_1}(1-\mathcal{X}_{\beta-k\delta})^{ \langle\gamma, \beta -k\delta \rangle\Omega(\beta-k\delta)}\notag\\
        &\prod_{k\geq 1}(1-\mathcal{X}_{- k\delta})^{ \langle \gamma, -k\delta \rangle \Omega(k\delta)}\Bigg)\;.
    \end{align}
    \item Asymptotics as $\lambda_{\mathrm B} \to 0$: For each ray $\rho$ and $\gamma \in \Gamma$, we have
    \begin{equation}
        \mathcal{X}_{\gamma,\rho}(v,w,\lambda_{\mathrm B})e^{-Z_{\gamma}(z,w)/\lambda_{\mathrm B}}\to 1, \;\;\; \text{as} \;\;\;\; \lambda_{\mathrm B} \to 0,\;\;\; \lambda_{\mathrm B}\in \mathbb{H}_{\rho}.
    \end{equation}
    \item Polynomial growth as $\lambda_{\mathrm B} \to \infty$: for each ray $\rho$ and $\gamma \in \Gamma$, we have the following for some $k>0$:
    \begin{equation}
        |\lambda_{\mathrm B}|^{-k}<|\mathcal{X}_{\gamma,\rho}(v,w,\lambda_{\mathrm B})|<|\lambda_{\mathrm B}|^k, \;\;\;\; \text{for} \;\;\; |\lambda_{\mathrm B}|\gg0\,.
    \end{equation}
\end{itemize}

Such a problem is shown to admit a unique solution \cite[Lemma 4.9]{BridgelandDT}, and the solution is given as follows. By the twisted homomorphism property, it is enough to describe $\mathcal{X}_{\gamma,\rho}$ for $\gamma\in \{ \beta^{\vee},\beta,\delta^{\vee},\delta\}$. The solutions for $\gamma=\beta$ and $\gamma=\delta$ have trivial Stokes jumps, and they are given by
\begin{equation}
    \mathcal{X}_{\beta,\rho}(v,w,\lambda_{\mathrm B})=e^{2\pi \I v/\lambda_{\mathrm B}}, \;\;\;\; \mathcal{X}_{\delta,\rho}(z,w,\lambda_{\mathrm B})=e^{2\pi \I w/\lambda_{\mathrm B}},
\end{equation}
for any ray $\rho$. On the other hand, for a ray $\rho_k$ between $\mathcal{L}_k$ and $\mathcal{L}_{k-1}$, the functions $\mathcal{X}_{\beta^{\vee},-\rho_k}(v,w,\lambda_{\mathrm B})$ and $\mathcal{X}_{\delta^{\vee},-\rho_k}(v,w,\lambda_{\mathrm B})$ are given by (see \cite[Equation (67)]{BridgelandCon})
\begin{equation}
    \begin{split}
    \mathcal{X}_{\beta^{\vee},-\rho_k}(v,w,\lambda_{\mathrm B})&=F^*(v+kw\,|\,w,-\lambda_{\mathrm B})\\ \mathcal{X}_{\delta^{\vee},-\rho_k}(v,w,\lambda_{\mathrm B})&=H^*(v+kw\,|\,w,-\lambda_{\mathrm B})(F^*(v+kw\,|\,w,-\lambda_{\mathrm B}))^k\,,
    \end{split}
\end{equation}
where $F^*$ is defined in terms of double sine function, and $H^*$ in terms of the triple sine function (see \cite[Section 5]{BridgelandCon} for more details on the definitions of $F^*$ and $H^*$). On the other hand, $\mathcal{X}_{\beta^{\vee},\rho_k}$ and $\mathcal{X}_{\delta^{\vee},\rho_k}$ are determined by the relation
\begin{equation}\label{realitycondition}
    \mathcal{X}_{\gamma,\rho_k}(v,w,\lambda_{\mathrm B})=1/\mathcal{X}_{\gamma,-\rho_k}(v,w,-\lambda_{\mathrm B})\,,
\end{equation}
that follows from uniqueness of the solutions of the Riemann--Hilbert problem. 

\subsection{Relation of the partition function to the Riemann--Hilbert problem}\label{normpartfunct}

In this section we wish to relate the Stokes jumps of $F_{\rho}(\lambda,t)$ with the Stokes jumps of the Riemann--Hilbert problem. More specifically, we will first consider a normalization of the partition function $\exp(F_{\rho}(\lambda,t))$ by a factor proportional to $\lambda^{1/12}\exp(F_{\rho}(\lambda,0))$. This normalization will not only capture the required Stokes jumps at $\pm l_{\infty}$, but will also allow us to recover the constant map contribution of (\ref{freeenergydecomp}). Indeed, by Corollary \ref{constmapBorelsum}, this normalization introduces back the Borel sum of the constant map constribution to the free energy. We will then relate the Stokes jumps of the normalized partition function to the jumps of the Riemann--Hilbert problem of Section \ref{RHproblem}. This will in turn show us how the BPS spectrum of the resolved conifold  is encoded in the Stokes jumps of $F_{\rho}(\lambda,t)$. We will use the notation of Section \ref{RHproblem} throughout. \\

To establish the link with Section \ref{RHproblem} more clearly, it will be convenient to consider the projectivized parameters
    \begin{equation}
        t=v/w, \;\;\;\; \check\lambda =\lambda_{\mathrm B}/w,
    \end{equation}
where we recall that $\check{\lambda}=\lambda/2\pi$. We think of the tuple $(v,w)$ as a point of $M$.\\

We consider the rays $\mathcal{L}_k=\mathbb{R}_{<0}\cdot 2\pi \I(v+kw)=\mathbb{R}_{<0}\cdot Z_{\beta+k\delta}$ for $k\in \mathbb{Z}$ and $\mathcal{L}_{\infty}=\mathbb{R}_{<0}\cdot 2\pi \I w=\mathbb{R}_{<0}\cdot Z_{k\delta}$. The relation to the old Stokes rays is then given by $\mathcal{L}_k=w\cdot l_k$ and $\mathcal{L}_{\infty}=w \cdot l_{\infty}$ .

\begin{dfn} Given a ray $\rho$ different from $\{\pm \mathcal{L}_{k}\}_{k\in \mathbb{Z}}\cup \{\pm \mathcal{L}_{\infty}\}$ we define the for $\lambda_{\mathrm B}\in \mathbb{H}_{\rho}$ and $(v,w)\in M$ with $\text{Im}(v/w)\neq 0$,
    \begin{equation}
        \mathcal{F}_{\rho}(\lambda_{\mathrm B},v,w):=  F_{ w^{-1}\cdot\rho}\bigg( \frac{2\pi\lambda_{\mathrm B}}{ w},\frac{v}{w}\bigg)=F_{ w^{-1}\cdot\rho}(\lambda,t)\,.
    \end{equation}
Notice that
\begin{equation}
    \mathcal{F}_{\rho}(\lambda_{\mathrm B},v,1)= F_{\rho}(\lambda,t)\,.
\end{equation}    
\end{dfn} 

Following the same argument as in Proposition \ref{Stokesjumps}, it is easy to check the following:

\begin{prop}
 Let $\rho_k$ be a ray in the sector determined by the rays $\mathcal{L}_{k}$ and $\mathcal{L}_{k-1}$. Then, if $\mathrm{Im}(v/w)>0$, on the overlap of their domains of definition in the $\lambda_{\mathrm B}$ variable we have
\begin{equation}\label{STproj}
    \Phi_{\pm \mathcal{L}_k}(\lambda_{\mathrm B},v,w):=\mathcal{F}_{\pm \rho_{k+1}}(\lambda_{\mathrm B},v,w)
    -\mathcal{F}_{\pm \rho_k}(\lambda_{\mathrm B},v,w) =\frac{1}{2\pi \I}\partial_{\lambda_{\mathrm B}}\Big(\lambda_{\mathrm B}\mathrm{Li}_2\big(e^{\pm 2\pi \I(v+kw)/\lambda_{\mathrm B}}\big)\Big)\;.
\end{equation}
If $\mathrm{Im}(v/w)<0$, then the previous jumps also hold provided $\rho_{k+1}$ is interchanged with $\rho_{k}$ in the above formulas.\\
\end{prop}

\begin{proof}
After a change of integration variables, we obtain
\begin{equation}
    \mathcal{F}_{ \rho_{k+1}}(\lambda_{\mathrm B},v,w)
    -\mathcal{F}_{ \rho_k}(\lambda_{\mathrm B},v,w)=\frac{1}{w}\int_{\mathcal{H}(\mathcal{L}_k)}\mathrm d\xi \; e^{-\xi/\lambda_{\mathrm B}}G(\xi/w,v/w)\;,
\end{equation}
where $\mathcal{H}(\mathcal{L}_k)$ is a Hankel contour around $\mathcal{L}_k$. The result then follows by summing over residues along the poles in $\mathcal{L}_k$, as in Proposition \ref{Stokesjumps}.
\end{proof}

\subsubsection{Normalizing the partition function and the Stokes jumps at $\pm l_{\infty}$}\label{-2BPS}

 We would like to first make sense of a limit of the form
\begin{equation}
    F_{\rho}(\lambda,0):= \lim_{t\to 0}F_{\rho}(\lambda,t)\,,
\end{equation}
where $t$ is taken to satisfy $\text{Re}(t)>0$, $\text{Im}(t)>0$; and such that along the limit, $\rho$ is always between $l_{-1}$ and $l_0$ (resp.\ $-l_{-1}$ and $-l_0$) if $\rho$ is on the right (resp.\ left) Borel half-plane.

\begin{lem}\label{t0limit} The limit $F_{\rho}(\lambda,0)$ from above exists for $\lambda \in \mathbb{H}_{\rho}$. In fact, if $\rho$ is on the right (resp.\ left) Borel plane, is can be analytically continued to $\lambda \in \mathbb{C}^{\times}\setminus\mathbb{R}_{\leq 0}$ (resp.\ $\lambda \in \mathbb{C}^{\times}\setminus\mathbb{R}_{\geq 0}$). 

\end{lem}

\begin{proof}
Let us assume that $\rho$ is on the right Borel plane. Then by Propositions \ref{BorelsumR} and \ref{intrep2}, and our condition on $t$, we have
\begin{equation}
    F_{\rho}(\lambda,t)=F_{\mathbb{R}_{>0}}(\lambda,t)=\log(G_3(t\,|\,\check{\lambda},1))\,,
\end{equation}
where we recall that $G_3$ is the function defined in terms of the triple sine in Definition \ref{G3def}.\\

By Proposition \ref{intrep2}, the function $G_3(t\,|\,\check{\lambda},1)$ has a well defined value at $t=0$; and $G_3(0 \, | \, \check\lambda,1)$ is everywhere regular, vanishing only at the points $\check{\lambda} \in \mathbb{Q}_{\leq 0}$.\\

In particular, for $\check{\lambda} \in \mathbb{H}_{\rho}$, we have
\begin{equation}\label{t=0triplesine}
    F_{\rho}(\lambda,0)=\lim_{t \to 0}F_{\rho}(\lambda,t)=\log(G_3(0\,|\,\check \lambda,1))\,,
\end{equation}
and we can analytically continue $F_{\rho}(\lambda,0)$ to $\lambda \in \mathbb{C}^{\times}\setminus\mathbb{R}_{\leq 0}$.\\

If $\rho$ is on the left Borel plane, the statement follows from the previous case, together with the relation $F_{\rho}(\lambda,t)=F_{-\rho}(-\lambda,t)$ from Lemma \ref{reflectionlemma}.
\end{proof}
Notice that by Lemma \ref{t0limit}, we can also write
\begin{equation}
           F_{\rho}(\lambda,0)=\frac{1}{\lambda^2} \mathrm{Li}_{3}(1)+ \lim_{t\to 0}\bigg(\frac{B_2}{2}\mathrm{Li}_{1}(e^{2\pi \I t})+\int_{\rho} \mathrm d\xi\,e^{-\xi/\check{\lambda}} G(\xi,t)\bigg)\,,
\end{equation}
where the limit in $t$ is assumed to satisfy the constraints from above.\\

The following proposition suggests that one can obtain the appropriate Stokes jumps at $\pm l_{\infty}$ by considering a normalization involving $F_{\rho}(\lambda,0)$:
\begin{prop}
Let $\rho$ (resp.\ $\rho'$) be a ray close to $l_{\infty}=\I \mathbb{R}_{<0}$ from the left (resp.\ right). Then for $\lambda$ in their common domain of definition
\begin{equation}
    F_{\pm \rho}(\lambda,0)-F_{\pm \rho'}(\lambda,0)=\frac{1}{\pi \I}\sum_{k \geq 1}\partial_{\check\lambda}\Big(\check\lambda \,\mathrm{Li}_2\big(e^{\pm 2\pi \I k/\check \lambda}\big)\Big) - \frac{\I\pi}{12}\;.
\end{equation}

Furthermore $F_{\rho}(\lambda,0)$ only has Stokes jumps along $\pm l_{\infty}$.
\end{prop}
\begin{proof}
First, notice that by our definition of the limit in $t$, we have
\begin{equation}
    \begin{split}
    F_{\rho}(\lambda,0)-F_{\rho'}(\lambda,0)
    &=\lim_{t\to 0}\Big(\int_{\rho}\mathrm d\xi\,  e^{-\xi/\check{\lambda}} G(\xi,t)-\int_{\rho'}\mathrm d\xi\, e^{-\xi/\check{\lambda}} G(\xi,t)\Big)\\
    &=\lim_{t\to 0}\int_{\mathcal{H}}\mathrm d\xi\, e^{-\xi/\check{\lambda}}G(\xi,t)\,,
    \end{split}
\end{equation}
where $\mathcal{H}=\rho-\rho'$ denotes a Hankel contour along $\I\mathbb{R}_{<0}$, containing $l_k$ for $k\geq 0$ and $-l_{k}$ for $k<0$. Hence, for $\lambda$ close to $l_{\infty}$, the Hankel contour just gives the contribution of these rays that we previously computed:
\begin{equation}
    \begin{split}
    F_{\rho}(\lambda,0)-F_{\rho'}(\lambda,0)
    &=\lim_{t\to 0}\int_{\mathcal{H}}\mathrm d\xi\, e^{-\xi/\check{\lambda}}G(\xi,t)\\
    &=\lim_{t\to 0} \bigg(\frac{1}{2\pi \I}\sum_{k\geq 1}\Big[ \partial_{\check\lambda}\Big(\check\lambda\, \mathrm{Li}_2(e^{2\pi \I(t+k)/\check\lambda})\Big)+\partial_{\check\lambda}\Big(\check\lambda \,\mathrm{Li}_2(e^{-2\pi \I(t-k)/\check\lambda})\Big)\Big]\\
    & \quad\quad\quad\quad +\frac{1}{2\pi \I}\partial_{\check\lambda}\Big( \mathrm{Li}_2(e^{2\pi \I t/\check\lambda})\Big)\bigg)\\
    &=\frac{1}{\pi \I}\sum_{k\geq 1}\partial_{\check\lambda}\Big(\check\lambda \,\mathrm{Li}_2(e^{2\pi \I k/\check\lambda})\Big) + \frac{1}{2\pi \I}\mathrm{Li}_2(1)\\
    &=\frac{1}{\pi \I}\sum_{k\geq 1}\partial_{\check\lambda}\Big(\check\lambda\, \mathrm{Li}_2(e^{2\pi \I k/\check\lambda})\Big) -\frac{\pi \I}{12}\;.
    \end{split}
\end{equation}

A similar argument follows for $-l_{\infty}=\I\mathbb{R}_{>0}$. Furthermore, the fact that there are no other Stokes jumps follows from the way we have defined the limit $F_{\rho}(\lambda,0)$. For example, if $\rho$ and $\rho'$ are both on the right Borel plane, then along the limit $\rho$ and $\rho'$ are both between $l_0$ and $l_{-1}$, and hence $F_{\rho}(\lambda,0)=F_{\rho'}(\lambda,0)$.
\end{proof}

We can as before projectivize, and define
\begin{equation}
    \mathcal{F}_{ \rho}(\lambda_\mathrm B,0,w):=\lim_{v\to 0}\mathcal{F}_{\rho}(\lambda_\mathrm B,v,w)\,,
\end{equation}
where the limit in $v$ is such that $t=v/w$ satisfies the conditions of the previous limit in $F_{\rho}(\lambda,0)$.

\begin{prop}\label{projinfinityjumps}
Let $\rho$ (resp.\ $\rho'$) be a ray close to $\mathcal{L}_{\infty}=\I\mathbb{R}_{<0}$ from the left (resp.\ right). Then for $\lambda_\mathrm B$ in their common domain of definition
\begin{equation}\label{inftyjump}
    \mathcal{F}_{\pm \rho}(\lambda_\mathrm B,0,w)-\mathcal{F}_{\pm \rho'}(\lambda_\mathrm B,0,w)=\sum_{k \geq 1}\Phi_{\pm\mathcal{L}_{\infty},k}-\frac{\pi \I}{12}\,,
\end{equation}
where 
\begin{equation}\label{SJ2proj}
    \Phi_{\pm\mathcal{L}_{\infty},k}(\lambda_{\mathrm B},v,w):=\frac{1}{\pi \I}\partial_{\lambda_{\mathrm B}}\Big(\lambda_{\mathrm B} \mathrm{Li}_2\big(e^{\pm 2\pi \I k w/ \lambda_{\mathrm B}}\big)\Big)
\end{equation}
Furthermore, $\mathcal{F}_{\pm \rho}(\lambda_{\mathrm B},0,w)$ only has Stokes jumps along $\pm \mathcal{L}_{\infty}$.
\end{prop}
On the other hand, Corollary \ref{constmapBorelsum} and \rf{Borelconstmap} suggest to not only normalize by $\mathcal{F}_{\rho}(\lambda_B,0,w)$ but to also normalize by a $\frac{1}{12}\log(\lambda_B/w)$ term as follows:

\begin{dfn}  Given a ray $\rho$ different from $\{\pm \mathcal{L}_{k}\}_{k\in \mathbb{Z}}\cup \{\pm \mathcal{L}_{\infty}\}$ we define the following for $\lambda_{\mathrm B}\in \mathbb{H}_{\rho}$ and $(v,w)\in M$ with $\text{Im}(v/w)>0$:
\begin{equation}\label{normalizationdef}
    \widehat{\mathcal{F}}_{\rho}(\lambda_{\mathrm B},v,w):=\mathcal{F}_{\rho}(\lambda_{\mathrm B},v,w)-\mathcal{F}_{\rho}(\lambda_{\mathrm B},0,w) -\frac{1}{12}\log\Big(\frac{\lambda_B}{w}\Big)\,,
\end{equation}
where we place the branch cut of the $\log$ term at $l_{\infty}=\mathrm i\mathbb{R}_{<0}$. We also define the normalized partition function as
\begin{equation}
    \widehat{Z}_{\rho}(\lambda_{\mathrm B},v,w):=\exp(\widehat{\mathcal{F}}_{\rho}(\lambda_{\mathrm B},v,w))\,.
\end{equation}
We remark that due to (\ref{Borelconstmap}) and Corollary \ref{constmapBorelsum}, asymptotic expansion of the the normalized free energy $\hat{\mathcal{F}}_{\rho}(\lambda,t,1)$ as $\lambda \to 0$ recovers the full $\mathcal{F}(Q,\lambda)$ of (\ref{freeenergydecomp}) up to an overall shift by a constant, which can be absorbed in a redefinition of $\widehat{\mathcal{F}}_{\rho}(\lambda_{\mathrm B},v,w)$. 

\end{dfn}

We immediately obtain the following:

\begin{cor} The Stokes jumps of $\hat{\mathcal{F}}(\lambda_B,v,w)$ at $\pm \mathcal{L}_k$ are given by (\ref{STproj}), while (using the notation of Proposition \ref{projinfinityjumps}) we find that the Stokes jumps at $\pm \mathcal{L}_{\infty}$ are determined by
\begin{equation}\label{inftyjumpslog}
     -\mathcal{F}_{\pm \rho}(\lambda_\mathrm B,0,w)+\mathcal{F}_{\pm \rho'}(\lambda_\mathrm B,0,w)-\lim_{\lambda_B \to \pm \mathcal{L}_{\infty}^+}\frac{1}{12}\log\Big(\frac{\lambda_B}{w}\Big)-\lim_{\lambda_B \to \pm \mathcal{L}_{\infty}^{-}}\frac{1}{12}\log\Big(\frac{\lambda_B}{w}\Big)=-\sum_{k \geq 1}\Phi_{\pm\mathcal{L}_{\infty},k}\mp \frac{\pi \I}{12}\,,
\end{equation}
where $\lim_{\lambda_B \to \pm \mathcal{L}_{\infty}^{+}}$ (resp. $\lim_{\lambda_B \to \pm \mathcal{L}_{\infty}^{-}}$) denotes the anti-clockwise (resp. clockwise) limit to $\pm \mathcal{L}_{\infty}$. 
\end{cor}

\begin{proof}
The only new thing to notice is that because of the $\log(\check{\lambda})$ term with branch cut at $l_{\infty}$ of the normalization, we find that 
\begin{equation}
    -\lim_{\lambda_B \to  \mathcal{L}_{\infty}^+}\frac{1}{12}\log\Big(\frac{\lambda_B}{w}\Big)-\lim_{\lambda_B \to \mathcal{L}_{\infty}^{-}}\frac{1}{12}\log\Big(\frac{\lambda_B}{w}\Big)=-\frac{2\pi \mathrm i}{12}=-\frac{\pi\mathrm i}{6}\,,
\end{equation}
while 
\begin{equation}
    -\lim_{\lambda_B \to  -\mathcal{L}_{\infty}^+}\frac{1}{12}\log\Big(\frac{\lambda_B}{w}\Big)-\lim_{\lambda_B \to -\mathcal{L}_{\infty}^{-}}\frac{1}{12}\log\Big(\frac{\lambda_B}{w}\Big)=0\,.
\end{equation}
Combining these jumps with (minus) the jumps of (\ref{inftyjump}) gives the desired result.

\end{proof}

\subsubsection{Relation to the Riemann--Hilbert problem}

In this subsection, we wish to relate the Stokes jumps of $\widehat{Z}_{\rho}$ with the Stokes jumps of the Riemann--Hilbert problem in Section \ref{RHproblem}. We will see that the jumps of $\widehat{Z}_{\rho}$ serve as ``potentials" for the jumps of the RH problem. \\

For the discussion below, it will be useful to note that
the Stokes jumps can be represented in terms of the double sine function, revealing some important properties. A useful review of definition and relevant properties of the 
double sine function $\sin_2(z\,|\,\omega_1,\omega_2)$ can be found in \cite[Section 4]{BridgelandCon} and references therein. 

\begin{dfn}
For $z\in \mathbb{C}$ and $\omega_1,\omega_2\in \mathbb{C}^{\times}$, let 
\begin{equation}
    F(z\,|\,\omega_1,\omega_2):=\exp\Big(-\frac{\pi \I}{2}B_{2,2}(z\,|\, \omega_1,\omega_2)\Big)\cdot \sin_2(z\,|\,\omega_1,\omega_2)
\end{equation}
where  $B_{2,2}(z\,|\, \omega_1,\omega_2)$ is a multiple Bernoulli polynomial.  
\end{dfn}

We will use the following properties of the function $F$.

\begin{prop}(See \cite[Proposition 4.1]{BridgelandCon})\label{propertiesF} The function $ F(z\,|\,\omega_1,\omega_2)$ is a single valued meromorphic function of the variables $z\in \mathbb{C}$ and $\omega_1,\omega_2\in \mathbb{C}^{\times}$ under the assumption $\omega_1/\omega_2\not\in \mathbb{R}_{<0}$. It has the following properties:
\begin{itemize}
    \item The function is regular and non-vanishing except at the points
    \begin{equation}
        z=a\omega_1+b\omega_2, \quad a,b\in \mathbb{Z}\,.
    \end{equation}
    \item It is invariant under simultaneous rescaling of the three arguments, and symmetric in $\omega_1,\omega_2$.
    \item It satisfies the difference equation
    \begin{equation}
        \frac{F(z+\omega_1\, |\, \omega_1,\omega_2)}{F(z\, |\, \omega_1,\omega_2)}=(1-e^{2\pi \I z/\omega_2})^{-1}\,.
    \end{equation}
    \item When $\mathrm{Re}(\omega_i)>0$ and $0<\mathrm{Re}(z)<\mathrm{Re}(\omega_1+\omega_2)$ there is an integral representation 
    \begin{equation}
        F(z\, |\, \omega_1,\omega_2)=\exp\Big(\int_{\mathbb{R}+\I0^+}\frac{\mathrm du}{u}\frac{e^{zu}}{(e^{\omega_1u}-1)(e^{\omega_2 u}-1)}\Big)\,.
    \end{equation}
    
\end{itemize}

\end{prop}

\begin{dfn} Let $\Phi_{\pm \mathcal{L}_k}(\lambda_{\mathrm B},v,w)$ and $\Phi_{\pm\mathcal{L}_{\infty},k}(\lambda_{\mathrm B},v,w)$ be as in (\ref{STproj}) and (\ref{SJ2proj}), respectively. We then define\footnote{The minus sign in the exponent of the second expression in (\ref{expSJ}) is due to the normalization of the partition function.} 
\begin{equation}\label{expSJ}
        \Xi_{\pm \mathcal{L}_k}(v,w,\lambda_{\mathrm B}):=e^{\Phi_{\pm \mathcal{L}_k}(\lambda_{\mathrm B},v,w)},\;\;\;\;\;  \Xi_{\pm \mathcal{L}_{\infty},k}(v,w,\lambda_{\mathrm B}):=e^{-\Phi_{\pm \mathcal{L}_\infty,k}(\lambda_{\mathrm B},v,w)}\;.
\end{equation}
\end{dfn}
We now give a proposition relating the exponentials of the Stokes jumps (\ref{expSJ}) of $\widehat{Z}_{\rho}$ to the function $F$.

\begin{prop}\label{tranfuncdoublesine} 
Assuming $\mathrm{Im}(v/w)>0$, we can write on their common domains of definition
\begin{equation}\label{Frel}
    \begin{split}
    \Xi_{\pm \mathcal{L}_k}(v,w,\lambda_{\mathrm B})&=(F(\pm(v+kw)/\lambda_{\mathrm B}+1\,|\, 1,1 ))^{-1}\\
    \Xi_{\pm \mathcal{L}_{\infty},k}(v,w,\lambda_{\mathrm B})&=(F(\pm kw/\lambda_{\mathrm B}+1\,|\,1,1))^2\,.
    \end{split}
\end{equation}
Furthermore, the functions on the right of (\ref{Frel}) are holomorphic and non-vanishing for  $\lambda_{\mathrm B}\in \mathbb{H}_{\pm \mathcal{L}_k}$ and $\lambda_{\mathrm B}\in \mathbb{H}_{\pm \mathcal{L}_{\infty}}$, respectively. 
\end{prop}

\begin{proof}
First, notice that 
\begin{equation}
    \begin{split}
    \Phi_{ \mathcal{L}_k}(\lambda_{\mathrm B},v,w)&=\frac{1}{2\pi \I}\partial_{\lambda_{\mathrm B}}\Big(\lambda_{\mathrm B}\mathrm{Li}_2\big(e^{ 2\pi \I(v+kw)/\lambda_{\mathrm B}})\Big)\\
    &=\frac{1}{2\pi \I}\partial_{\check\lambda}\Big(\check\lambda \mathrm{Li}_2(e^{2\pi \I(t+k)/\check\lambda})\Big)\\
    &=\phi_{l_k}(\lambda,t)\,.
    \end{split}
\end{equation}
In particular, if $w=e^{2\pi \I t/\check{\lambda}}$ and $\widetilde{q}=e^{2\pi \I k/\check{\lambda}}$, then for $\lambda$ near $l_k$, we can expand 
\begin{equation}\label{comp1}
    \Phi_{ \mathcal{L}_k}(\lambda_{\mathrm B},v,w)=\frac{1}{2\pi \I}\sum_{l=1}^{\infty}  \frac{(w\widetilde{q}^k)^l}{l}\Big(\frac{1}{l} -\log(w\widetilde{q}^k)\Big)\,.
\end{equation}

On the other hand, assuming for the moment that $0<\text{Re}(t+k)<\text{Re}(\check\lambda)$,  $\text{Re}(\check{\lambda})>0$, and using the scaling invariance of $F$, we find by Proposition \ref{propertiesF} that
\begin{equation}
    \begin{split}
    F((v+kw)/\lambda_{\mathrm B}+1\,|\, 1,1 )&=F((t+k)/\check\lambda+1\,|\, 1,1 )\\
    &=F(t+k+\check\lambda\,|\, \check\lambda,\check\lambda )\\
    &=\exp\Big(\int_{\mathbb{R}+i0^+}\frac{\mathrm du}{u}\frac{e^{u(t+k)}}{(2\sinh(\check\lambda u/2))^2}\Big)\,.
    \end{split}
\end{equation}
Using that $\text{Im}(v/w)=\text{Im}(t)>0$, we can compute the previous integral by closing the contour in the upper half-plane and sum over the residues at $2\pi \I l/\check{\lambda}$ for $l\geq 1$, obtaining
\begin{equation}\label{comp2}
    \begin{split}
    \int_{\mathbb{R}+i0^+}\frac{\mathrm du}{u}\frac{e^{u(t+k)}}{(2\sinh(\check\lambda u/2))^2}&=2\pi \I\sum_{l=1}^{\infty}\Big(\frac{1}{\check{\lambda}}\Big)^2 \frac{\partial}{\partial u} \frac{e^{u(t+k)}}{u}\Big|_{u=2\pi \I l/\check{\lambda}}\\
    &=-\frac{1}{2\pi \I}\sum_{l=1}^{\infty}\frac{(w\widetilde{q}^k)^l}{l}\Big(\frac{1}{l} -\log(w\widetilde{q}^k)\Big)\,.
    \end{split}
\end{equation}
Comparing (\ref{comp1}) with (\ref{comp2}), the first result then follows, with the other cases being analogous.\\

To check the second statement, we use the fact from Proposition \ref{propertiesF} that the function $F(z\,|\, \omega_1,\omega_2)$ is regular and non-vanishing except at the points 
\begin{equation}
    z=a\omega_1+b\omega_2, \quad a,b\in \mathbb{Z}\,.
\end{equation}

We then find that $F(\pm(v+kw)/\lambda_{\mathrm B}+1\,|\, 1,1 )$ is regular except at the points where
\begin{equation}
    (v+kw)/\lambda_{\mathrm B}\in \mathbb{Z} \quad \iff \lambda_{\mathrm B}=\frac{v+kw}{n}, \quad n\in \mathbb{Z}\,.
\end{equation}
In particular, $F(\pm(v+kw)/\lambda_{\mathrm B}+1\,|\, 1,1 )$ is regular for $\lambda_{\mathrm B}$ in the half-plane centered at $\pm \mathcal{L}_k=\pm \mathbb{R}_{<0}\cdot 2\pi \I(v+nw)$. The other case follows similarly. 
\end{proof}
From the previous proposition we obtain the following corollary:

\begin{cor}\label{Stokespotential}
Let $(v,w)\in M$ such that $\mathrm{Im}(v/w)>0$. Then $\Xi_{\pm \mathcal{L}_k}$ and $\Xi_{\pm \mathcal{L}_{\infty},k}$ serve as potentials for the jumps of the Riemann--Hilbert problem of Section \ref{RHproblem}, in the sense that
\begin{equation}
    \begin{split}
    \frac{\Xi_{\pm \mathcal{L}_k}(v+\lambda_{\mathrm B},w,\lambda_{\mathrm B})}{\Xi_{\pm \mathcal{L}_k}(v,w,\lambda_{\mathrm B})}&=(1-e^{\pm 2\pi \I(v+kw)/\lambda_{\mathrm B}})^{\pm 1}=(1-\mathcal{X}_{\pm (\beta+k\delta)})^{\pm \langle \beta^{\vee}, \beta+k\delta \rangle \Omega(\beta+k\delta)},\\
    \frac{\Xi_{\pm \mathcal{L}_k}(v,w+\lambda_{\mathrm B},\lambda_{\mathrm B})}{\Xi_{\pm \mathcal{L}_k}(v,w,\lambda_{\mathrm B})}&=(1-e^{\pm 2\pi \I(v+kw)/\lambda_{\mathrm B}})^{\pm k}=(1-\mathcal{X}_{\pm (\beta+k\delta)})^{\pm \langle \delta^{\vee}, \beta+k\delta \rangle \Omega(\beta+k\delta)},\\
        \frac{\Xi_{\pm \mathcal{L}_{\infty},k}(0,w+\lambda_{\mathrm B},\lambda_{\mathrm B})}{\Xi_{\pm \mathcal{L}_{\infty},k}(0,w,\lambda_{\mathrm B})}&=(1-e^{\pm 2\pi \I kw/\lambda_{\mathrm B}})^{\mp 2k}=(1-\mathcal{X}_{\pm k\delta})^{\pm \langle \delta^{\vee}, k\delta \rangle \Omega(k\delta)}.\\
    \end{split}
\end{equation}

\end{cor}
\begin{proof}
We use the fact from Proposition \ref{propertiesF} that $F(z\,|\, 1,1)$ satisfies the difference equation
\begin{equation}
    \frac{F(z+1\,|\, 1,1)}{F(z\,|\, 1,1)}=\frac{1}{1-e^{2\pi \I z}}\,.
\end{equation}
Then by Proposition \ref{tranfuncdoublesine} we find that 
\begin{equation}
    \frac{\Xi_{ \mathcal{L}_k}(v+\lambda_{\mathrm B},w,\lambda_{\mathrm B})}{\Xi_{ \mathcal{L}_k}(v,w,\lambda_{\mathrm B})}=\bigg(\frac{F( (v+kw)/\lambda_{\mathrm B} +1 +1|1,1)}{F((v+kw)/\lambda_{\mathrm B}+1,|1,1)}\bigg)^{-1}=1-e^{2\pi \I(v+kw)/\lambda_{\mathrm B}}\,.
\end{equation}
The other identities follow similarly.

\end{proof}

\begin{rem} From corollary \ref{Stokespotential} we see how the BPS spectrum is encoded in the jumps of the normalized partition function. Namely, the BPS spectrum of the resolved conifold appears in the expressions of the jumps of $\widehat{Z}_{\rho}(\lambda_B,v,w)$ as follows:
\begin{equation}
    \begin{split}
    \Xi_{\pm \mathcal{L}_k}(v,w,\lambda_\mathrm B)&=\exp \bigg(\frac{\Omega(\beta+k\delta)}{2\pi \I}\partial_{\lambda_\mathrm B}\Big(\lambda_\mathrm B\mathrm{Li}_2(e^{\pm Z_{\beta+k\delta}/\lambda_\mathrm B})\Big)\bigg)\\
    \Xi_{\pm \mathcal{L}_{\infty},k}(v,w,\lambda_\mathrm B)&=\exp\bigg(\frac{\Omega(k\delta)}{2\pi \I}\partial_{\lambda_{\mathrm B}}\Big(\lambda_{\mathrm B} \mathrm{Li}_2\big(e^{\pm Z_{k\delta}/ \lambda_{\mathrm B}}\big)\Big)\bigg)\,,
    \end{split}
\end{equation}
making explicit how  the DT-invariant are encoded in the jumps.

\end{rem}
\subsection{The line bundle defined by the normalized partition function}\label{linebundlenormalized}

We would like to discuss how to define a line bundle  $\mathcal{L}\to \mathbb{C}^{\times}\times M$, such that the partition functions $\widehat{Z}_{\rho}(\lambda_{\mathrm B},v,w)$ define a section of $\mathcal{L}$.\\

To concretely define the line bundle, we restrict for simplicity to\footnote{See Remark \ref{otherM} for the other points of $M$.} 
\begin{equation}
    M_{+}:=\{(v,w)\in M \;|\;  \text{Im}(v/w)>0\}\,.
\end{equation} 
Furthermore, let $\rho_k$ be a ray between $\mathcal{L}_k$ and $\mathcal{L}_{k-1}$. For definiteness, we pick $\rho_k$ to be always in the middle of  $\mathcal{L}_k$ and $\mathcal{L}_{k-1}$, and consider the open sets
\begin{equation}\label{opencov}\begin{split}
U_{k}^{\pm}&:=\{(\lambda_{\mathrm B},v,w)\in \mathbb{C}^{\times}\times M_+ \; | \; \lambda_{\mathrm B} \in \mathbb{H}_{\pm \rho_k}\}\,.
\end{split}
\end{equation}
We remark that the condition on $\lambda_{\mathrm B}$ actually depends on $(v,w)$, since the latter specifies the rays $\mathcal{L}_k$ (and hence also $\rho_k$). We then clearly have that $\{U_{k}^{+}\}_{k \in \mathbb{Z}}\cup \{U_{k}^{-}\}_{k \in \mathbb{Z}}$ forms an open cover of $\mathbb{C}^{\times}\times M_+$.\\

If $U_{k_1}^{+}\cap U_{k_2}^{+}\neq \emptyset$ for $k_1<k_2$, we then define for $(\lambda_{\mathrm B},v,w)\in U_{k_1}^{+}\cap U_{k_2}^{+}$,
\begin{equation}
     g^+_{k_1,k_2}(\lambda_{\mathrm B},v,w):=\prod_{k_1\leq k < k_2}\Xi_{\mathcal{L}_k}(\lambda_{\mathrm B},v,w)\,.
\end{equation}
Notice that if $(\lambda_{\mathrm B},v,w)\in U_{k_1}^{+}\cap U_{k_2}^{+}$, then $(\lambda_{\mathrm B},v,w)\in \mathbb{H}_{\mathcal{L}_k}$ for $k_1\leq k \leq k_2$, so by Proposition \ref{tranfuncdoublesine} we have that $g^+_{k_1,k_2}$ is $\mathbb{C}^{\times}$-valued:
\begin{equation}
    g^+_{k_1,k_2}\colon U_{k_1}^{+}\cap U_{k_2}^{+} \to \mathbb{C}^{\times}\,.
\end{equation}
With the assumptions $U_{k_1}^{+}\cap U_{k_2}^{+}\neq \emptyset$ for $k_1<k_2$, we also define $g^+_{k_2,k_1}:=(g^+_{k_1,k_2})^{-1}$ and $g^+_{k,k}:=1$ for any $k\in \mathbb{Z}$.\\

If $U_{k_1}^{-}\cap U_{k_2}^{-}\neq \emptyset$ for $k_1<k_2$, then we similarly define
\begin{equation}
    g^{-}_{k_1,k_2}\colon  U_{k_1}^{-}\cap U_{k_2}^{-} \to \mathbb{C}^{\times}\,.
\end{equation}
by
\begin{equation}
    g^{-}_{k_1,k_2}(\lambda_{\mathrm B},v,w):= \prod_{k_1\leq k < k_2}\Xi_{-\mathcal{L}_k}(\lambda_{\mathrm B},v,w)\,,
\end{equation}
and $g_{k_2,k_1}^{-}:=(g^{-}_{k_1,k_2})^{-1}$, $g_{k,k}^{-}:=1$.\\

 On the other hand, if for some $k_1,k_2\in \mathbb{Z}$ we have $U_{k_1}^{+}\cap U_{k_2}^{-}\neq \emptyset$, then $\rho_{k_1}\neq \rho_{k_2}$ and hence out of the two sectors determined by $\rho_{k_1}$ and $-\rho_{k_2}$ there is a smallest one, which we denote by $[\rho_{k_1},-\rho_{k_2}]$. For all $(\lambda_{\mathrm B},v,w)\in U_{k_1}^{+}\cap U_{k_2}^{-}$ we must either have that  $\mathcal{L}_{\infty} \subset [\rho_{k_1},-\rho_{k_2}]$ or $- \mathcal{L}_{\infty} \subset [\rho_{k_1},-\rho_{k_2}]$. In the first case we define\footnote{Recall that the $e^{\pm\pi \I/12}$ factors are due to the jumps (\ref{inftyjumpslog}) of the normalization of the partition function.}
\begin{equation}\label{inftytrans}
    g_{k_1,k_2}^{\infty}(\lambda_{\mathrm B},v,w):=e^{-\pi\mathrm  i/12}\prod_{k\geq k_1}\Xi_{\mathcal{L}_k}(\lambda_{\mathrm B},v,w)\prod_{k< k_2}\Xi_{-\mathcal{L}_{k}}(\lambda_{\mathrm B},v,w)\prod_{k\geq 1}\Xi_{\mathcal{L}_{\infty},k}(\lambda_{\mathrm B},v,w)\,,
\end{equation}
and $g_{k_2,k_1}^{\infty}:=(g_{k_1,k_2}^{\infty})^{-1}$.\\

Notice that in this first case we have $\lambda_{\mathrm B}\in \mathbb{H}_{\mathcal{L}_k}$ for $k\geq n_1$, $\lambda_{\mathrm B}\in \mathbb{H}_{-\mathcal{L}_k}$ for $k<n_2$, and $\lambda_{\mathrm B}\in \mathbb{H}_{\mathcal{L}_{\infty}}$. Hence, by Proposition \ref{tranfuncdoublesine} and the convergence of the above product\footnote{Here we use that the corresponding infinite sums of $\Phi_{\mathcal{L}_{k}}$, $\Phi_{-\mathcal{L}_{k}}$ and $\Phi_{\mathcal{L}_{\infty,k}}$, converge for $(\lambda_B,v,w)\in U_{k_1}^+\cap U_{k_2}^{-}$.}, we find that 
\begin{equation}
    g_{k_1,k_2}^{\infty}(\lambda_{\mathrm B},v,w)\colon  U_{k_1}^{+}\cap U_{k_2}^{-} \to \mathbb{C}^{\times}\,.
\end{equation}
On the other hand, in the second case we define
\begin{equation}\label{-inftytrans}
    g_{k_2,k_1}^{-\infty}(\lambda_{\mathrm B},v,w):=e^{\pi \I/12}\prod_{k\geq k_2}\Xi_{-\mathcal{L}_{k}}(\lambda_{\mathrm B},v,w)\prod_{k< k_1}\Xi_{\mathcal{L}_{k}}(\lambda_{\mathrm B},v,w)\prod_{k\geq 1}\Xi_{\mathcal{L}_{\infty},-k}(\lambda_{\mathrm B},v,w)\,,
\end{equation}
and $g_{k_1,k_2}^{-\infty}:=(g_{k_2,k_1}^{-\infty})^{-1}$.\\

With the previous results, the following proposition then follows:

\begin{prop}
The functions $g^{\pm}_{k_1,k_2}$, $g^{\pm \infty}_{k_1,k_2}$ associated to the cover $\{U_{k}^{+}\}_{k \in \mathbb{Z}}\cup \{U_{k}^{-}\}_{k \in \mathbb{Z}}$ define a $1$-\v{C}ech cocycle over $\mathbb{C}^{\times}\times M_{+}$, and hence a line bundle $\mathcal{L}\to \mathbb{C}^{\times}\times M_{+}$. Furthermore, assuming $\mathrm{Im}(v/w)>0$, the normalized partition functions $\widehat{Z}_{\rho}(\lambda,z,w)$ glue together into a section of $\mathcal{L}$.
\end{prop}

\begin{proof}
The fact that $g^{\pm}_{k_1,k_2}$ and $g^{\pm \infty}_{k_1,k_2}$ define a $1$-\v{C}ech cocycle follows directly from their definitions. Furthermore, the fact that the $\widehat{Z}_{\rho}(\lambda,z,w)$ glue together into a section follows from our previous discussions on the Stokes jumps of $\widehat{Z}_{\rho}(\lambda,z,w)$.
\end{proof}

\begin{rem} \label{otherM}

Let 
\begin{equation}
    M_{-}:=\{(v,w)\in M \; | \; \mathrm{Im}(v/w)<0\}, \quad M_0:= \{(v,w)\in M \; | \; \mathrm{Im}(v/w)=0\}\,,
\end{equation}
so that $M=M_{+}\cup M_{-}\cup M_{0}$. By using the Stokes jumps for the case $(u,v) \in M_{-}$ (resp. $(u,v)\in M_{0}$, where all the Stokes rays collapse to either $\mathcal{L}_{\infty}$ or $-\mathcal{L}_{\infty}$)  we can as before define a line bundle over $\mathbb{C}^{\times}\times M_{-}$ (resp. $\mathbb{C}^{\times}\times M_0$) having the normalized partition function as a section. Since the Borel summations $F_{\rho}(\lambda_B,v,w)$ make sense on sufficiently small open subsets of $\mathbb{C}^{\times}\times M$, and furthermore depend holomorphically on the parameters, these line bundles glue together into a holomorphic line bundle $\mathcal{L}\to \mathbb{C}^{\times}\times M$ having the normalized partition function as a section. 
\end{rem}

\subsection{Relation to the Rogers dilogarithm and hyperholomorphic line bundles}

Notice that one can write the Stokes jumps of $\widehat{Z}_{\rho}(\lambda_{\mathrm B},v,w)$ along $\pm \mathcal{L}_k=\pm \mathbb{R}_{<0}\cdot Z_{\beta+k\delta}$ as
\begin{equation}\label{normaltranfunc}
    \Phi_{\mathcal{L}_k}(\lambda_{\mathrm B},v,w)=\frac{\Omega(\beta+k\delta)}{2\pi \I}\Big(\mathrm{Li}_2(\mathcal{X}_{\pm (\beta+k\delta)})+\log(\mathcal{X}_{\pm (\beta+k\delta)})\log(1-\mathcal{X}_{\pm (\beta+k\delta)})\Big)\,,
\end{equation}
where
\begin{equation}
    \log(\mathcal{X}_{\pm (\beta+k\delta)})=\pm 2\pi \I(v+kw)/\lambda_{\mathrm B}\,.
\end{equation}
Up to a factor of $\frac12$ in the second summand, this matches 
\begin{equation}\label{Rdilog}
    \frac{\Omega(\beta+k\delta)}{2\pi \I}L(\mathcal{X}_{\pm (\beta+k\delta)})\,,
\end{equation}
where $L(x)$ denotes the Rogers dilogarithm
\begin{equation}
L(x):=\mathrm{Li}_2(x)+\frac{1}{2}\log(x)\log(1-x)\,.
\end{equation}

In previous works \cite{Neitzke_hyperhol,APP}, hyperholomorphic line bundles with transition functions having the form of the exponentials of (\ref{Rdilog}) have been discussed in the context of instanton-corrected hyperk\"{a}hler and quaternionic-K\"{a}hler geometries. Our goal in the rest of this section is then two-fold:

\begin{itemize}
    \item We would first like to explain how (\ref{normaltranfunc}) and (\ref{Rdilog}) are related by changes of local trivialization involving the solutions of the RH problem of Section \ref{RHproblem}.
    \item This will then be used to relate the line bundle $\mathcal{L}\to \mathbb{C}^{\times}\times M_+$ constructed in Section \ref{linebundlenormalized} with a certain ``conformal limit" of the line bundles constructed in \cite{Neitzke_hyperhol,APP}.  
\end{itemize}

\subsubsection{Relation to the Rogers dilogarithm}
In order to relate to the Rogers dilogarithm, we follow the idea  suggested in \cite[Appendix H]{CLT20} (see also Lemma \ref{scaledif}, below). \\

We start by considering the solutions of the RH problem from Section \ref{RHproblem}. Notice that since  $\mathcal{X}_{\gamma,\rho}(v,w,-):\mathbb{H}_{\rho}\to \mathbb{C}^{\times}$, then there must exist $x_{\gamma,\rho}(v,w,-)\colon\mathbb{H}_{\rho}\to \mathbb{C}$ such that 
\begin{equation}
    \mathcal{X}_{\gamma,\rho}(v,w,\lambda_\mathrm{B})=\exp(x_{\gamma,\rho}(v,w,\lambda_\mathrm{B})).
\end{equation}
We then define for $(\lambda_{\mathrm B},v,w)\in U_{\rho}=\{(v,w,\lambda_B)\in M_{+}\times \mathbb{C}^{\times}\, |\, \lambda_B \in \mathbb{H}_{\rho}\}$,
\begin{equation}
    x_{\beta,\rho}:=2\pi \I v/\lambda_{\mathrm B}, \;\;\;\; x_{\delta,\rho}:=2\pi \I w/\lambda_{\mathrm B}, \;\;\;\;  x_{\beta^{\vee},\rho}:=\log \mathcal{X}_{\beta^{\vee},\rho}, \;\;\;\; x_{\delta^{\vee},\rho}:=\log \mathcal{X}_{\delta^{\vee},\rho}\,.
\end{equation}
In taking the logs in the last two coordinates, we do as the following lemma:

\begin{lem}
The branches of the logs in $x_{\beta^{\vee},\rho}$ and $x_{\delta^{\vee},\rho}$ can be taken such that the following relations are satisfied on the common domains of definition:

\begin{itemize}
    \item Along $\pm \mathcal{L}_{k}$:
    \begin{equation}\label{logSJ1}
    \begin{split}
    x_{\beta^{\vee},\pm \rho_{k+1}}&=x_{\beta^{\vee},\pm \rho_k}\pm \log(1-\mathcal{X}_{\pm (\beta+k\delta)})\,,\\
    x_{\delta^{\vee},\pm \rho_{k+1}}&=x_{\delta^{\vee},\pm \rho_k}\pm k\log(1-\mathcal{X}_{\pm (\beta+k\delta)})\,,
    \end{split}
\end{equation}
\item If $\rho_{k_1}$, and $\rho_{k_2}$ are such that $\mathcal{L}_{\infty} \subset [\rho_{k_1},-\rho_{k_2}]$:
            \begin{equation}\label{logSJ2}
    \begin{split}
        x_{\beta^{\vee}, -\rho_{k_2}}=x_{\beta^{\vee}, \rho_{k_1}} + \Big(\sum_{k\geq k_1}& \log(1-\mathcal{X}_{ \beta+k\delta})-\sum_{k>-k_2} \log(1-\mathcal{X}_{-\beta+k\delta})\Big)\,,
    \end{split}
    \end{equation}
and 
            \begin{equation}\label{logSJ3}
    \begin{split}
        x_{\delta^{\vee},- \rho_{k_2}}=x_{\delta^{\vee},\rho_{k_1}} + \Big(\sum_{k\geq k_1}& k\log(1-\mathcal{X}_{ \beta+k\delta})+\sum_{k>-k_2} k\log(1-\mathcal{X}_{-\beta+k\delta})\\
        &\quad \quad -\sum_{k\geq 1} 2k \log (1-\mathcal{X}_{ k\delta})\Big)\,,
    \end{split}
    \end{equation}
    \item  If $\rho_{k_1}$, and $\rho_{k_2}$ are such that $-\mathcal{L}_{\infty} \subset [\rho_{k_1},-\rho_{k_2}]$:
            \begin{equation}\label{logSJ4}
    \begin{split}
        x_{\beta^{\vee}, \rho_1}=x_{\beta^{\vee}, -\rho_{k_2}} +\Big(-\sum_{k\geq k_2}& \log(1-\mathcal{X}_{- \beta-k\delta})+\sum_{k>-k_1} \log(1-\mathcal{X}_{\beta-k\delta})\Big)\,,
    \end{split}
    \end{equation}
and 
            \begin{equation}\label{logSJ5}
    \begin{split}
        x_{\delta^{\vee}, \rho_{k_1}}=x_{\delta^{\vee},-\rho_{k_2}} + \Big(-\sum_{k\geq k_2}& k\log(1-\mathcal{X}_{ -\beta-k\delta})-\sum_{k>-k_1} k\log(1-\mathcal{X}_{\beta-k\delta})\\
        &\quad \quad +\sum_{k\geq 1} 2k \log (1-\mathcal{X}_{ -k\delta})\Big)\,.
    \end{split}
    \end{equation}
\end{itemize}
\end{lem}

\begin{proof}
We do the argument for $\beta^{\vee}$, since for the $\delta^{\vee}$ is the same.\\

First we pick $\rho_k$ for some $k$ and fix a branch for $\log \mathcal{X}_{\beta^{\vee},\rho_k}$. We then fix the branches of $\log \mathcal{X}_{\beta^{\vee},\rho_n}$ with $n\in \mathbb{Z}$ by enforcing the jumps (\ref{logSJ1}). On the other hand, we set $x_{\beta^{\vee},-\rho_k}(\lambda_B):=-\log \mathcal{X}_{\beta^{\vee},\rho_n}(-\lambda_B)$. This indeed gives a log of $\mathcal{X}_{\beta^{\vee},-\rho_k}$ due to (\ref{realitycondition}). It is then easy to check that the $x_{\beta^{\vee},-\rho_n}$ for $n \in \mathbb{Z}$ must satisfy the corresponding jumps in (\ref{logSJ1}).\\

With such choices, the jumps (\ref{logSJ2}) and (\ref{logSJ4}) must be satisfied up to summands of $2\pi i n_1$ and $2\pi in_2$ respectively. Furthermore, it is easy to check that $n_1$ and $n_2$ must be independent of the rays $\rho_{k_1}$ and $\rho_{k_2}$ satisfying $\mathcal{L}_{\infty} \subset [\rho_{k_1},-\rho_{k_2}]$ or $-\mathcal{L}_{\infty} \subset [\rho_{k_1},-\rho_{k_2}]$, respectively. Furthermore, by the condition  $x_{\beta^{\vee},-\rho_n}(\lambda_B)=-\log \mathcal{X}_{\beta^{\vee},\rho_n}(-\lambda_B)$, one finds that $n_1=-n_2$. It is then easy to check that by setting $x_{\beta^{\vee},\rho_n}:=\log \mathcal{X}_{\beta^{\vee},\rho_n}+2\pi in_1$ for all $n\in \mathbb{Z}$, the jumps (\ref{logSJ1}), (\ref{logSJ2}) and (\ref{logSJ4}) are satisfied. 
\end{proof}

\begin{lem}\label{scaledif} For $\rho_k$ between $\mathcal{L}_{k}$ and $\mathcal{L}_{k-1}$, consider the following holomorphic function on $U_{k}^+$ (resp.\ $U_{k}^{-}$):
\begin{equation}
     f_{\pm \rho_k}:=x_{\beta,\pm \rho_{k}}\cdot x_{\beta^{\vee},\pm \rho_{k}}+x_{\delta,\pm \rho_{k}}\cdot x_{\delta^{\vee},\pm \rho_{k}}\,.
\end{equation}
We then have the following relations:

\begin{itemize}
    \item On $U_{k_1}^{\pm}\cap U_{k_2}^{\pm}$ with $k_1<k_2$:
    \begin{equation}
        f_{\pm  \rho_{k_2}}-f_{\pm  \rho_{k_1}}=\sum_{k_1\leq k< k_2}\log(\mathcal{X}_{\pm (\beta+k\delta)})\log(1-\mathcal{X}_{\pm (\beta+k\delta)})\,.
    \end{equation}
    \item If $U_{k_1}^{+}\cap U_{k_2}^{-}\neq \emptyset$, then recall that for all $(\lambda_{\mathrm B},v,w)\in U_{k_1}^{+}\cap U_{k_2}^{-}$ we either have $\mathcal{L}_{\infty}\subset [\rho_{k_1},-\rho_{k_2}]$ or $-\mathcal{L}_{\infty}\subset [\rho_{k_1},-\rho_{k_2}]$. In the first case, we have
    \begin{equation}
        \begin{split}
        f_{-\rho_{k_2}}-f_{\rho_{k_1}}&=\sum_{k \geq k_1}\log(\mathcal{X}_{ \beta+k\delta})\log(1-\mathcal{X}_{\beta+k\delta}) +\sum_{k > -k_2}\log(\mathcal{X}_{ -\beta+k\delta})\log(1-\mathcal{X}_{-\beta+k\delta})\\
        &\quad -2 \sum_{k\geq 1}\log(\mathcal{X}_{k\delta})\log(1-\mathcal{X}_{k\delta})\,,
        \end{split}
    \end{equation}
    while in the second case 
    \begin{equation}
        \begin{split}
        f_{\rho_{k_1}}-f_{-\rho_{k_2}}&=\sum_{k \geq k_2}\log(\mathcal{X}_{ -\beta-k\delta})\log(1-\mathcal{X}_{-\beta-k\delta}) + \sum_{k > -k_1}\log(\mathcal{X}_{ \beta-k\delta})\log(1-\mathcal{X}_{\beta-k\delta})\\
        &\quad -2\sum_{k\geq 1}\log(\mathcal{X}_{-k\delta})\log(1-\mathcal{X}_{-k\delta})\,.
        \end{split}
    \end{equation}
\end{itemize}
\end{lem}

\begin{proof}
From the Stokes jumps (\ref{logSJ1}) we obtain that on $U_{k_1}^+\cap U_{k_2}^+$, we have
\begin{align*}
    f_{\rho_{k_2}}-f_{\rho_{k_1}}
    &=\sum_{k_1\leq k< k_2}(x_{\beta,\rho'}+kx_{\delta,\rho'})\log(1-\mathcal{X}_{\beta+k\delta})\\
    &=\sum_{k_1\leq k< k_2}\Big(\frac{2\pi \I(v+kw)}{\check{\lambda}}\Big)\log(1-\mathcal{X}_{\beta+k\delta})\\
    &=\sum_{k_1\leq k< k_2}\log(\mathcal{X}_{\beta+k\delta})\log(1-\mathcal{X}_{\beta+k\delta})\,.\\ \numberthis
\end{align*}
The other cases follow similarly by using the other jumping relations (\ref{logSJ1}), (\ref{logSJ2}), (\ref{logSJ3}), (\ref{logSJ4}), (\ref{logSJ5}). For example, if $U_{k_1}^{+}\cap U_{k_2}^{-}\neq \emptyset$ and $\mathcal{L}_{\infty}\subset [\rho_{k_1},-\rho_{k_2}]$, we then have
\begin{align*}
        f_{-\rho_{k_2}}-f_{\rho_{k_1}}&=x_{\beta,\rho_{k_1}}\cdot \Big(\sum_{k\geq k_1} \log(1-\mathcal{X}_{ \beta+k\delta})-\sum_{k>-k_2} \log(1-\mathcal{X}_{-\beta+k\delta})\Big)\\
        &+x_{\delta,\rho_{k_1}}\cdot \Big(\sum_{k\geq k_1} k\log(1-\mathcal{X}_{ \beta+k\delta})+\sum_{k>-k_2} k\log(1-\mathcal{X}_{-\beta+k\delta}) -\sum_{k\geq 1} 2k \log (1-\mathcal{X}_{ k\delta})\Big)\\
        &=\sum_{k \geq k_1}\log(\mathcal{X}_{ \beta+k\delta})\log(1-\mathcal{X}_{\beta+k\delta}) + \sum_{k > -k_2}\log(\mathcal{X}_{ -\beta+k\delta})\log(1-\mathcal{X}_{-\beta+k\delta})\\
        &\quad -2 \sum_{k\geq 1}\log(\mathcal{X}_{k\delta})\log(1-\mathcal{X}_{k\delta})\,. \numberthis
\end{align*}
\end{proof}

We therefore obtain the following:

\begin{prop} \label{transRoger}In terms of the local trivializations of $\mathcal{L} \to \mathbb{C}^{\times}\times M_+$ given by
\begin{equation}
   \exp \Big(\frac{\I}{4\pi}f_{\pm \rho_k}\mp \frac{\mathrm i\pi}{24}\Big)\widehat{Z}_{\pm \rho_k}\colon U_{k}^{\pm}\to \mathcal{L}\,.
\end{equation}
The transition functions of $\mathcal{L}$ are given as follows:

\begin{itemize}
    \item If $U_{k_1}^{\pm}\cap U_{k_2}^{\pm}\neq \emptyset$ for $k_1<k_2$, we have
    \begin{equation}
        \widetilde{g}^{\pm }_{k_1,k_2}=\prod_{k_1\leq k < k_2}\exp\Big(\frac{\Omega(\beta+k\delta)}{2\pi \I}L(\mathcal{X}_{\pm (\beta+k\delta)})\Big)\,.
    \end{equation}
    \item If $U_{k_1}^{+}\cap U_{k_2}^{-}\neq \emptyset$, and $\mathcal{L}_{\infty}\subset [\rho_{k_1},-\rho_{k_2}]$, then 
\begin{align*}
    \widetilde{g}_{k_1,k_2}^{\infty}(\lambda_{\mathrm B},v,w)&=\prod_{k\geq k_1}\exp\Big(\frac{\Omega(\beta+k\delta)}{2\pi \I}L(\mathcal{X}_{ \beta+k\delta})\Big)\prod_{k< k_2}\exp\Big(\frac{\Omega(\beta+k\delta)}{2\pi \I}L(\mathcal{X}_{- (\beta+k\delta)})\Big)\\
    &\quad \quad \cdot \prod_{k\geq 1}\exp\Big(\frac{\Omega(k\delta)}{2\pi \I}L(\mathcal{X}_{k\delta})\Big)\,, \numberthis
\end{align*}
and the analogous relation for the case $\mathcal{L}_{-\infty}\subset [\rho_{k_1},-\rho_{k_2}]$. In particular, one can get rid of the $e^{\pm i\pi/12}$ factors appearing in (\ref{inftytrans}) and (\ref{-inftytrans}).
\end{itemize}

\end{prop}

\begin{proof}
For simplicity, we compute the new transition functions in the case $U_{k+1}^{+}\cap U_{k}^{+}$, with all the others following the same type of argument using the rest of the identities in Lemma \ref{scaledif}. We have
\begin{equation}
    \begin{split}
        \widetilde{g}^{+}_{k,k+1}&=  \exp\Big(\frac{\I}{4\pi }(f_{\rho_{k+1}}-f_{\rho_k})\Big)\cdot g^{+}_{k,k+1}\\
        &=\exp\Big(-\frac{1}{4\pi \I}\log(\mathcal{X}_{\beta+k\delta})\log(1-\mathcal{X}_{\beta+k\delta})\Big)\cdot \exp\Big(\frac{1}{2\pi \I}\partial_{\lambda_{\mathrm B}}\Big(\lambda_{\mathrm B}\mathrm{Li}_2\big(\mathcal{X}_{\beta+k\delta})\Big)\Big)\\
        &=\exp\Big(\frac{\Omega(\beta+k\delta)}{2\pi \I}L(\mathcal{X}_{\beta+k\delta})\Big)\,,
    \end{split}
\end{equation}
where on the second line we have used Lemma \ref{scaledif}. 
\end{proof}

\subsubsection{Relation to hyperholomorphic line bundles}

We briefly recall the setting of \cite{Neitzke_hyperhol,APP}. The starting point is again a certain type of variations of BPS structures $(M,\Gamma,Z,\Omega)$. The variation of BPS structures should be such that the data $(M,\Gamma,Z)$ defines a (possibly indefinite) affine special K\"{a}hler (ASK) geometry on $M$ (or conical ASK in the case of \cite{APP}). More precisely, the pair $(M,\Gamma,Z)$ should satisfy:

\begin{itemize}
    \item The pairing $\langle - , - \rangle$ admits local Darboux frames $(\widetilde{\gamma}_i,\gamma^i)$\footnote{One can relax this condition by allowing ``flavor charges" (i.e. charges $\gamma$ such that $\langle\gamma,- \rangle=0$). For a description of this more general case see for example \cite{Neitzkenotes} or \cite[Section 2.2]{AST21}.}, and the $1$-forms $dZ_{\gamma^i}$ give a local frame of $T^*M$. 
    \item By using the identification $\Gamma \cong \Gamma^*$ given by $\gamma \to \langle \gamma, - \rangle$, we can induce a $\mathbb{C}$-bilinear pairing on $\Gamma^*\otimes \mathbb{C}$. With respect to this pairing, we should have
    \begin{equation}\label{sym}
        \langle\mathrm dZ\wedge\mathrm dZ \rangle=0\,.
    \end{equation}
    \item The two-form 
    \begin{equation}
        \omega:=\frac{1}{4}\langle \mathrm dZ \wedge \mathrm d\overline{Z} \rangle
    \end{equation}
    is non-degenerate. 
\end{itemize}
    Under the above conditions, $\{Z_{\gamma^i}\}_{i=1}^{\mathrm{dim}_{\mathbb{C}}(M)}$ give local coordinates on $M$, and it is not hard to check that $\tau_{ij}$ defined by $dZ_{\widetilde{\gamma}_i}=\tau_{ij}dZ_{\gamma^j}$ must be symmetric by (\ref{sym}), and 
    \begin{equation}
        \omega= \frac{1}{4}(\mathrm dZ_{\widetilde{\gamma}_i}\wedge \mathrm d\overline{Z}_{\gamma^i}- \mathrm dZ_{\gamma^i}\wedge \mathrm d\overline{Z}_{\widetilde{\gamma}_i})=\frac{\I}{2}\mathrm{Im}(\tau_{ij})\mathrm dZ_{\gamma^i}\wedge \mathrm d\overline{Z}_{\gamma^j}\,.
    \end{equation}
    In particular, $\omega$ gives a K\"{a}hler form of a (possibly indefinite) ASK geometry. The functions $\{Z_{\gamma^i}\}$ then define special holomorphic coordinates, while $\{Z_{\widetilde{\gamma}_i}\}$ define a conjugate system of special holomorphic coordinates.\\ 
    
    By following the prescription in \cite{GMN1,Neitzkenotes}, one can then define an ``instanton-corrected"\footnote{What this means is that the data of the BPS indices is used to obtain a new HK geometry from the semi-flat HK geometry associated to the ASK geometry via the rigid c-map. In the context of 4d $\mathcal{N}=2$ theories compactified on $S^1$, the modifications to the semi-flat HK geometry correspond to instanton corrections of the HK geometry associated to the corresponding low-energy effective theory.} hyperk\"{a}hler (HK) structure on the total space of a torus fibration $\pi \colon \mathcal{M}\to M$, where 
    \begin{equation}
        \mathcal{M}|_p:=\{\theta\colon \Gamma|_p \to \mathbb{R}/2\pi \mathbb{Z}\; |\; \theta_{\gamma+\gamma'}=\theta_{\gamma}+\theta_{\gamma'}+\pi \langle \gamma, \gamma' \rangle \}\,.
    \end{equation}
    The hyperk\"{a}hler geometry is encoded in terms of certain $\mathbb{C}^{\times}$-valued local functions  $\mathcal{Y}_{\gamma}(\zeta,\theta)$ on the twistor space $\mathcal{Z}:=\mathbb{C}P^1\times \mathcal{M}$ of $\mathcal{M}$. The functions are labeled by local sections $\gamma$ of $\Gamma$, and must satisfy the GMN integral equations \cite{GMN1,Neitzkenotes}.\\
    
    Given such data, a certain holomorphic line bundle $\mathcal{L_{\mathcal{Z}}}\to \mathcal{Z}$ is constructed in \cite{Neitzke_hyperhol,APP}, descending to a hyperholomorphic line bundle $\mathcal{L}_{\mathcal{M}}\to \mathcal{M}$ (that is, a hermitian bundle having a unitary connection with curvature of type $(1,1)$ in all the complex structures of the HK structure of $\mathcal{M}$). Our concern in the following will not be the hyperholomorphic structure of $\mathcal{L}_{\mathcal{M}}$ itself, but the topology of $\mathcal{L}_{\mathcal{Z}}$. \\
    
    We would like to now describe $\mathcal{L}_{\mathcal{Z}}$ in the case of the resolved conifold. We need, however, to solve the following issue: the variation of BPS structures $(M,\Gamma,Z,\Omega)$ we have discussed in Section \ref{RHproblem} does not satisfy the ASK geometry condition, since for that case $Z_{\beta^{\vee}}=Z_{\delta^{\vee}}=0$, and hence
    \begin{equation}
         \langle\mathrm dZ \wedge\mathrm d\overline{Z} \rangle=\mathrm dZ_{\beta^{\vee}}\wedge\mathrm d\overline{Z}_{\beta}+\mathrm dZ_{\delta^{\vee}}\wedge\mathrm d\overline{Z}_{\delta}-\mathrm dZ_{\beta}\wedge\mathrm d\overline{Z}_{\beta^{\vee}}-\mathrm dZ_{\delta}\wedge\mathrm d\overline{Z}_{\delta^{\vee}}= 0\,.
    \end{equation}
    By possibly restricting to an open set $M'\subset M$, and taking $\Gamma':=\Gamma|_{M'}$, we can assume we have a central charge $Z'\colon M'\to \Gamma'\otimes \mathbb{C}$ satisfying the ASK geometry property and such that $Z'|_{\mathbb{Z}\beta\oplus \mathbb{Z}\delta}=Z|_{\mathbb{Z}\beta\oplus \mathbb{Z}\delta}$. Setting $(\gamma^1,\gamma^2)=(\beta,\delta)$, such a central charge can be found by picking a  holomorphic function $\mathfrak{F}(Z_{\gamma^i})\colon M'\to \mathbb{C}$ such that the matrix $\mathrm{Im}(\partial^2\mathfrak{F}/\partial{Z_{\gamma^i}}\partial{Z_{\gamma^j}})$ is non-degenerate, and taking $Z'_{\beta^{\vee}}:=\partial_{Z_{\beta}}\mathfrak{F}$, $Z'_{\delta^{\vee}}:=\partial_{Z_{\delta}}\mathfrak{F}$ (for example, we can just pick $\mathfrak{F}(Z_{\beta},Z_{\delta})=\I(Z_{\beta}^2 +Z_{\delta}^2)/2$). For simplicity, we will assume in the following that have chosen $Z'$ satisfying the ASK condition and such that $M'=M$.
    We can therefore consider the HK manifold $\mathcal{M}$ and line bundle $\mathcal{L}_{\mathcal{Z}}\to \mathcal{Z}$ associated to $(M,\Gamma,Z',\Omega)$.\\
    
    To describe $\mathcal{L}_{\mathcal{Z}}\to \mathcal{Z}$ in this case, we do as follows: recall that a quadratic refinement for $(\Gamma|_p,\langle -, - \rangle)$ is a function $\sigma:\Gamma|_p \to \mathbb{Z}_2$ such that
      \begin{equation}
        \sigma(\gamma)\sigma(\gamma')=(-1)^{\langle \gamma,\gamma'\rangle}\sigma(\gamma+\gamma')\,.
    \end{equation}
    In our particular case, we can make a global choice of quadratic refinement $\sigma:\Gamma \to \mathbb{Z}_2$ determined by $\sigma(\beta)=\sigma(\delta)=\sigma(\beta^{\vee})=\sigma(\delta^{\vee})=1$. With such a choice, we can identify
        \begin{equation}
        \mathcal{M}\cong\{\theta\colon \Gamma \to \mathbb{R}/2\pi \mathbb{Z}\; |\; \theta_{\gamma+\gamma'}=\theta_{\gamma}+\theta_{\gamma'}\}\cong M\times (S^1)^4\,,
    \end{equation}
    via $e^{\I\theta_{\gamma}} \to \sigma(\gamma) e^{\I\theta_{\gamma}}$.\\ 
    
    We consider the bundle  $\widetilde{\pi}\colon\widetilde{\mathcal{M}}\to M$, whose fibers are the universal cover of the fibers of $\pi\colon \mathcal{M}\to M$. Namely,
      \begin{equation}
        \widetilde{\mathcal{M}}:=\{\theta: \Gamma \to \mathbb{R} \; |\; \theta_{\gamma+\gamma'}=\theta_{\gamma}+\theta_{\gamma'}\}\cong M\times \mathbb{R}^4\;.
    \end{equation}
    The main reason for going to the universal cover is to avoid certain issues regarding the domains of definitions of the transition functions involving the Rogers dilogarithm expressions, as we will see below.\\

    Since our end-goal is to compare with $\mathcal{L}\to \mathbb{C}^{\times}\times M_+$ from Section \ref{normpartfunct}, we will restrict to $\widetilde{\mathcal{M}}_{+}:=\widetilde{\pi}^{-1}(M_+)$. However, a similar argument follows for the line bundle over $\mathbb{C}^{\times}\times M_{-}$ and $\mathbb{C}^{\times}\times M_{0}$ (recall Remark \ref{otherM}).\\
    
    It is easy to see that the HK structure on $\mathcal{M}_+$ lifts to $\widetilde{\mathcal{M}}_+$, and we will denote the corresponding twistor space by $\widetilde{\mathcal{Z}}_+=\mathbb{C}P^1\times \widetilde{\mathcal{M}}_+$. We consider the rays $\mathcal{L}_{k}=\mathbb{R}_{<0}\cdot 2\pi \I(v+kw)$, and pick $\rho_k$ between $\mathcal{L}_k$ and $\mathcal{L}_{k-1}$. We furthermore consider the cover $\{V_{k}^{\pm}\}_{k\in \mathbb{Z}}$ of $\mathbb{C}^{\times}\times \widetilde{\mathcal{M}}_{+}$ given by
     \begin{equation}
        V_{k}^{\pm}:=\{(\zeta,\theta)\in \mathbb{C}^{\times}\times \widetilde{\mathcal{M}}_{+} \; | \; \zeta \in \mathbb{H}_{\pm \rho_k}\}\,.
    \end{equation}
    Notice that the condition on $\zeta$ depends on $\widetilde{\pi}(\theta)=(u,v)$, since the latter determines the rays $\mathcal{L}_{k}$, and hence $\rho_k$. Furthermore, notice that $V_k^{\pm}=\widetilde{\pi}^{-1}(U_k^{\pm})$, where $U_k^{\pm}$ was defined in  (\ref{opencov}).\\
    
    We define a line bundle $\mathcal{L}_{\widetilde{\mathcal{Z}}_{+}}\to  \mathbb{C}^{\times}\times \widetilde{\mathcal{M}}_+$ via the following cocycle associated to the cover $\{V_{k}^{\pm}\}_{k\in \mathbb{Z}}$ (compare with \cite[Equation 4.8]{Neitzke_hyperhol} or \cite[Equation 3.29]{APP}):
    
    \begin{itemize}
        \item If $V_{k_1}^{\pm}\cap V_{k_2}^{\pm}\neq \emptyset$ for $k_1<k_2$, then
        \begin{equation}
            h_{k_1,k_2}^{\pm }(\zeta,\theta):= \prod_{k_1\leq k < k_2}\exp\bigg(\frac{\Omega(\beta+k\delta)}{2\pi \I}L(\mathcal{Y}_{\pm (\beta+k\delta)}(\zeta,\theta))\bigg)\,,
        \end{equation}
        where\footnote{Formula (\ref{sfcoord}) gives the so--called semi-flat coordinate labeled by $\pm(\beta+k\delta)$. In the case of the resolved conifold, only the coordinates of the form $\mathcal{Y}_{n\beta+m\delta +p\beta^{\vee}+q\delta^{\vee}}$ with $p\neq 0$ or $q\neq 0$ get ``instanton corrected" away from the semi-flat form.} 
        \begin{equation}\label{sfcoord}
            \mathcal{Y}_{\pm(\beta+k\delta)}(\zeta,\theta)=\exp \Big (\zeta^{-1}Z_{\pm(\beta+k\delta)}(\widetilde{\pi}(\theta)) +\I\theta_{\pm(\beta+k\delta)} +\zeta \overline{Z}_{\pm(\beta+k\delta)}(\widetilde{\pi}(\theta))\Big)\,.
        \end{equation}
        Notice that 
        \begin{equation}
            L(\mathcal{Y}_{ \pm(\beta+k\delta)})=\mathrm{Li}_2(\mathcal{Y}_{ \pm(\beta+k\delta)})+\frac{1}{2}\log(\mathcal{Y}_{ \pm(\beta+k\delta)})\log(1-\mathcal{Y}_{\pm(\beta+k\delta)})
        \end{equation}
        with 
        \begin{equation}\label{logy}
            \log(\mathcal{Y}_{ \pm(\beta+k\delta)})=\zeta^{-1}Z_{\pm(\beta+k\delta)}(\widetilde{\pi}(\theta)) +\theta_{\pm(\beta+k\delta)} +\zeta \overline{Z}_{\pm(\beta+k\delta)}(\widetilde{\pi}(\theta))
        \end{equation}
        is well defined for $\zeta \in \mathbb{H}_{\mathbb{R}_{<0}\cdot Z_{\pm(\beta+k\delta)}}$, since for such $\zeta$ we have $|\mathcal{Y}_{ \pm(\beta+k\delta)}|<1$, and hence $\mathrm{Li}_2(\mathcal{Y}_{ \pm(\beta+k\delta)})$ and $\log(1-\mathcal{Y}_{ \pm(\beta+k\delta)})$ make sense with their principal branches. \\
        
        We also set $h_{k_2,k_1}^{\pm}:=(h_{k_1,k_2}^{\pm})^{-1}$.
        \item   If $V_{k_1}^{+}\cap V_{k_2}^{-}\neq \emptyset$, and $\mathcal{L}_{\infty}\subset [\rho_{k_1},-\rho_{k_2}]$, then 
\begin{equation}
    \begin{split}
    h_{k_1,k_2}^{\infty}(\zeta,\theta)&:=\prod_{k\geq k_1}\exp\bigg(\frac{\Omega(\beta+k\delta)}{2\pi \I}L(\mathcal{\mathcal{Y}}_{ \beta+k\delta})\bigg)\prod_{k< k_2}\exp\bigg(\frac{\Omega(\beta+k\delta)}{2\pi \I}L(\mathcal{\mathcal{Y}}_{- (\beta+k\delta)})\bigg)\\
    &\quad \quad \cdot \prod_{k\geq 1}\exp\bigg(\frac{\Omega(k\delta)}{2\pi \I}L(\mathcal{Y}_{k\delta})\bigg)\,,\\
    \end{split}
\end{equation}
while for the case $\mathcal{L}_{-\infty}\subset [\rho_{k_1},-\rho_{k_2}]$
\begin{equation}
    \begin{split}
    h_{k_2,k_1}^{-\infty}(\zeta,\theta)&:=\prod_{k\geq k_2}\exp\bigg(\frac{\Omega(\beta+k\delta)}{2\pi \I}L(\mathcal{\mathcal{Y}}_{-( \beta+k\delta)}))\bigg)\prod_{k< k_1}\exp\bigg(\frac{\Omega(\beta+k\delta)}{2\pi \I}L(\mathcal{\mathcal{Y}}_{ \beta+k\delta})\bigg)\\
    &\quad \quad \cdot \prod_{k\geq 1}\exp\bigg(\frac{\Omega(k\delta)}{2\pi \I}L(\mathcal{Y}_{-k\delta})\bigg)\,.\\
    \end{split}
\end{equation}
We also set as before $h_{k_2,k_1}^{ \infty}:=(h_{k_1,k_2}^{\infty})^{-1}$ and $h_{k_1,k_2}^{- \infty}:=(h_{k_2,k_1}^{- \infty})^{-1}$.
\end{itemize}

In \cite{Neitzke_hyperhol,APP}, it is argued that such a bundle extends to a holomorphic bundle $\mathcal{L}_{\widetilde{\mathcal{Z}}_+}\to \widetilde{\mathcal{Z}}_+$, and that it descends to a hyperholomorphic line bundle $\mathcal{L}_{\widetilde{\mathcal{M}}_{+}}\to \widetilde{\mathcal{M}}_{+}$. The corresponding line bundle $\mathcal{L}_{\mathcal{M}_{+}}\to \mathcal{M}_+$ can then be obtained by a quotient by a certain action of $\Gamma^*\to M_+$, acting fiberwise on both $\mathcal{L}_{\widetilde{\mathcal{M}}_+}\to M_+$ and $\widetilde{\mathcal{M}}_+\to M_+$, and equivariantly with respect to $\mathcal{L}_{\widetilde{\mathcal{M}}_{+}}\to \widetilde{\mathcal{M}}_{+}$ (see for example \cite[Equation 3.7]{Neitzke_hyperhol}). The pullback of $\mathcal{L}_{\mathcal{M}_+}$ to the twistor space then gives $\mathcal{L}_{\mathcal{Z}}\to \mathbb{C}P^1\times \mathcal{M}_+$.\\

We now wish to relate the bundle $\mathcal{L}_{\mathcal{Z}}|_{\mathbb{C}^{\times}\times M_+}$ with the previous bundle $\mathcal{L}\to \mathbb{C}^{\times}\times M_+$ defined by the normalized partition functions $\widehat{Z}_\rho$ in Section \ref{normpartfunct}. 
We will focus on the complex Lagrangian submanifold $L\subset \mathcal{M}$ (with respect to one of the complex symplectic structures of the HK structure) given by
    \begin{equation}
        L:=\{\theta\in \mathcal{M} \; | \; \theta_{\beta}=\theta_{\delta}=\theta_{\beta^{\vee}}=\theta_{\delta^{\vee}}=0\}\,.
    \end{equation}
    The fact that this defines a complex Lagrangian submanifold $L$ of $\mathcal{M}$, can be seen for example from formula \cite[equation 3.10]{CT} of the instanton corrected holomorphic symplectic form (see also \cite{Gaiotto:2014bza}). Since $L$ can be identified with $M$ as complex manifolds, we will do so in the following. \\
    
    The line bundle $\mathcal{L}_{\mathcal{Z}}|_{\mathbb{C}^{\times}\times M_+}$ can be described by the transition functions $h_{k_1,k_2}^{\pm}|_{\mathbb{C}^{\times}\times M_+}$ and $h_{k_1,k_2}^{\pm \infty}|_{\mathbb{C}^{\times}\times M_+}$ associated to the cover $\{U_{k}^{\pm}\}_{k\in \mathbb{Z}}$ of $\mathbb{C}^{\times}\times M_{+}$ given in (\ref{opencov}). It is now easy to see how to obtain $\mathcal{L}\to \mathbb{C}^{\times} \times M_+$ from $\mathcal{L}_Z|_{\mathbb{C}^{\times}\times M_+}\to \mathbb{C}^{\times}\times M_+$. Namely, one considers the following conformal limit, studied in \cite{Gaiotto:2014bza}:

\begin{itemize}
    \item First, one introduces a scaling parameter $Z \to RZ$ for $R>0$. 
    \item One then considers the limit of the transition functions as $R\to 0$, while keeping the quotient $\lambda_{\mathrm B}=\zeta/R$ fixed.
\end{itemize}

After taking the conformal limit, we see that 
\begin{equation}
    \mathcal{Y}_{n\beta +m\delta}|_{\mathbb{C}^{\times}\times M_+} \to \mathcal{X}_{n\beta +m\delta}\,,
\end{equation}
and hence 
\begin{equation}
    h^{\pm}_{k_1,k_2}|_{\mathbb{C}^{\times}\times M_+} \to \widetilde{g}^{\pm}_{k_1,k_2}, \quad\quad   h^{\pm \infty }_{k_1,k_2}|_{\mathbb{C}^{\times}\times M_+} \to \widetilde{g}^{\pm \infty }_{k_1,k_2}\,,
\end{equation}
where $\widetilde{g}^{\pm}_{k_1,k_2}$ and $\widetilde{g}^{\pm \infty }_{k_1,k_2}$ correspond to the 1-cocycle associated to the cover $\{U_{k}^{\pm}\}_{k\in \mathbb{Z}}$ describing $\mathcal{L}\to \mathbb{C}^{\times}\times M_+$ from Proposition \ref{transRoger}.\\

From the previous discussion, we obtain the following:
\begin{prop}
Consider the $1$-cocycles associated to the cover $\{U_{k}^{\pm}\}$ of $\mathbb{C}^{\times}\times M_+$ and given by $\{\widetilde{g}^{\pm}_{k_1,k_2},\widetilde{g}^{\pm \infty}_{k_1,k_2}\}$ and $\{ h^{\pm}_{k_1,k_2}|_{\mathbb{C}^{\times}\times M_+}, h^{\pm \infty}_{k_1,k_2}|_{\mathbb{C}^{\times}\times M_+} \}$, respectively.
Then the $1$-cocycles are related by the conformal limit from above.  
\end{prop}


\section{The strong-coupling expansion and its Stokes phenomena}

In this section, we will demonstrate that the topological string partition 
function has a Borel summable strong-coupling expansion for $\la\ra\infty$.
The Stokes jumps of the strong-coupling expansion are found to reproduce 
the wall-crossing behaviour of counting functions for the {\it framed} BPS states representing 
composites of D0 and D2 branes bound to a heavy D6 brane in string theory on the resolved conifold.
This wall-crossing
has previously been studied by Jafferis and Moore \cite{JM}. This work in particular gave a physical derivation of the results of Szendr\"oi on the generating function of non-commutative DT invariants \cite{Szendroi}, see also \cite{NN} for related work in mathematics and \cite{Dimofte:2009bv,Aganagic:2009kf} in physics. \\

Because the techniques required in this section are the same as in the previous sections, we will give less details of the intermediate computations. 

\subsection{Borel summation of the strong-coupling expansion}

In order to derive the strong-coupling expansion we shall start with
the Woronowicz form of $F_{\mathrm{np}}(\lambda,t)$ given in (\ref{altform}), now rewritten as
\begin{equation}\label{worono-ints-strong}
F_{\mathrm{np}}(\la,t)=\frac{1}{(2\pi)^2}
\int_{\BR+\ii 0^+}\mathrm dv\;\frac{v+\alpha}{1-e^{v+\alpha}}
\log\big(1-e^{\check\la v}\big). 
\end{equation}
using the notations $\alpha=-2\pi\ii t'$ and $t'=t/\check\la $.
As before, we may rewrite this in terms of a  Laplace transformation, 
\begin{align}\label{strong-Laplace}
F_{\mathrm{np}}(\la,t)&=
\int_{0}^\infty \!\! \frac{\mathrm dv}{(2\pi)^2}\bigg(\frac{(v+\alpha)\log\big(1-e^{\check\la v-\ii0^+}\big)}{1-e^{v+\alpha}}-\frac{(v-\alpha)\log\big(1-e^{-\check\la v-\ii0^+}\big)}{1-e^{\alpha-v}}
\bigg)\\
&=\frac{1}{(2\pi)^2}
\int_{0}^{\infty}\mathrm dv\;\bigg[\frac{(v+\alpha)(\check\la v+\pi\ii)}{1-e^{v+\alpha}}+\bigg(\frac{v+\alpha}{1-e^{v+\alpha}}-\frac{v-\alpha}{1-e^{\alpha-v}}
\bigg)\log\big(1-e^{-\check\la v-\ii0^+}\big)\bigg]\notag\\
&=-\frac{\la}{(2\pi)^3}(2\mathrm{Li}_3(Q')+\alpha\,\mathrm{Li}_2(Q'))-\frac{\ii}{4\pi}(\mathrm{Li}_2(Q')+\alpha\,\mathrm{Li}_1(Q'))
+\int_{0}^{\infty}\mathrm dv\;e^{-\check\la v}G_s(v,t')\,,
\notag\end{align}
using the notations $Q'=e^{2\pi\ii t'}$
 and 
\begin{align}
G_\mathrm s(v,t')
&=\frac{1}{(2\pi)^2}\sum_{n\in\BZ\setminus \{0\}}\frac{1}{n^3}
\frac{v+2\pi\ii\,nt'}{1-e^{-v/n-2\pi\ii t'}}\,.
\notag\end{align}
Having represented the function $F_{\mathrm{np}}(\la,t)$ as a Laplace transform makes it straightforward to 
derive  an asymptotic series in inverse powers of $\la$ for which 
\rf{strong-Laplace} represents a Borel transform.\\

$G_\mathrm s(v,t')$ has poles  at $v=v_{kn}^\pm:=\mp 2\pi \ii n(t'+k)$, $k\in\BZ\setminus\{0\}$, 
$n\in\BZ_{>0}$. In the case $\mathrm{Im}(t')>0$, the poles
$v_{kn}^+$ and $v_{kn}^-$ are in the right and left half-planes, respectively.
Assuming $\mathrm{Re}(t')<1$,  one finds that  the strings of 
poles $\{v_{kn}^+\mid n\in\BZ_{>0}\}$ with $k\in\BZ_{< 0}$ are located in 
the upper half-plane. \\

We may decompose the complex plane representing
values of the integration variable $v$ into a union of rays $\pm l_k':= \pm  \mathbb{R}_{<0}\cdot 2\pi \I(t'+k)$ and wedges $[\pm l_k',\pm l_{k-1}']$ bounded by 
$\pm l_k'$ and 
$\pm l_{k-1}'$. Letting $\lambda':=1/\check\lambda$, for $\rho_k$ in the wedge $[l_k',l_{k-1}']$ and $\lambda'\in \mathbb{H}_{\pm \rho_k}$, we define 
\begin{align}
F'_{\pm \rho_k}(\la',t'):=&-\frac{1}{(2\pi)^2\lambda'}(2\mathrm{Li}_3(Q')-2\pi\ii t'\mathrm{Li}_2(Q')\big)-\frac{\ii}{4\pi}\big(\mathrm{Li}_2(Q')-2\pi\ii t'\mathrm{Li}_1(Q')\big) \notag\\
&+\int_{\pm \rho_k}\mathrm dv\;e^{-\frac{v}{\lambda'}}G_\mathrm s(v,t')\,.
\end{align}
The wedges
$[\pm l_k',\pm l_{k-1}']$
are natural domains of definition (in the $\lambda'$ variable) for the functions $F'_{\pm \rho_k}(\la',t')$,
differing by Stokes jumps from the strings of 
poles $\{v_{kn}^\pm\mid n\in\BZ_{>0}\}$ of $G_\mathrm s(v,t)$.

\subsection{Stokes jumps}

To compute the Stokes jumps, we follow the strategy from Section \ref{Stokesjumpssect}. The relevant residues  are
\begin{align}
&\Res_{v=v_{kn}^{\pm }}e^{-v/\lambda'}
G_\mathrm s(v,t)=\frac{e^{-v_{kn}^{\pm }/\lambda'}}{(2\pi)^2n^3}(\mp 2\pi\ii\,nk)n=
\pm\frac{1}{2\pi \ii}\frac{k}{n}\,e^{\pm\ii\la n(t'+k)}.\,
\end{align}
It follows that 
the Stokes jumps  across $\pm l_k'$ are explicitly given as
\begin{equation}\label{Fplusjumps}
\begin{aligned}
F'_{\pm \rho_{k+1}}-F'_{\pm \rho_k}&=2\pi\ii\sum_{n=1}^{\infty}
\pm  e^{\pm \la\ii n(t'+k)} \frac{\, k}{2\pi \ii n}=\mp k\log\big(1-e^{\pm \la\ii (t'+k)}\big)\\ &=
\mp k\log\big(1-e^{\pm 2\pi\ii (t+\check\la  k)}\big)=\mp k\log(1-Q^{\pm 1}{q}^{\pm k})\,,
\end{aligned}
\qquad\quad
\begin{aligned}
&Q:=e^{2\pi\ii t},\\[1ex]
&q:=e^{\ii\la}.
\end{aligned}
\end{equation}
Note that there is no jump for $k=0$. \\

For the rest of the section, we will assume that $0<\mathrm{Re}(t')<1$ and $\mathrm{Im}(t')>0$. Taking $\mathrm{Im}(\lambda')>0$ ($\iff \mathrm{Im}(\lambda)<0$), we can sum the jumps in the upper half-plane, and obtain
\begin{align}\label{strongjumps}
\lim_{k \to -\infty}F'_{\rho_k}-F'_{\rho_0}& =\sum_{k=-1}^{-\infty}(F'_{\rho_k}-F'_{\rho_{k+1}})=
-\sum_{k=1}^\infty k\log(1-Q{q}^{-k})\\
&=\sum_{l=1}^\infty\frac{q^{-l}}{l}\frac{Q^l}{(1-q^{-l})^2}
=-\sum_{k=1}^\infty\frac{1}{k}
\frac{e^{2\pi\ii kt}}{\big(2\sin\big(\frac{k\la}{2}\big)\big)^2}\,.
\notag \end{align}
On the other hand, if $\mathrm{Im}(\lambda')<0$ ($\iff \mathrm{Im}(\la)>0$), we can sum the jumps in the lower half-plane of the variable, which leads to
\begin{align}\label{strongjumps2}
\lim_{k\to \infty}F'_{\rho_k}-F'_{\rho_0}& =\sum_{k=0}^\infty(F'_{\rho_{k+1}}-F'_{\rho_k})=
-\sum_{k=1}^\infty k\log(1-Q{q}^{k})\notag \\
&=\sum_{l=1}^\infty\frac{q^{l}}{l}\frac{Q^l}{(1-q^{l})^2}=-\sum_{k=1}^\infty\frac{1}{k}
\frac{e^{2\pi\ii kt}}{\big(2\sin\big(\frac{k\la}{2}\big)\big)^2}\,.
 \end{align}
Note that the domains of definition of $\lim_{k\to \infty}F'_{\rho_k}-F'_{\rho_0}$
and $\lim_{k\to -\infty}F'_{\rho_k}-F'_{\rho_0}$ have empty intersection.\\ 

It will be instructive  to consider the normalised partition functions 
\begin{equation}\label{norm-funct}
\CZ_{\pm \rho_k}(\la',t'):=
\frac{Z'_{\rho_0}(\la',t')}{Z'_{\pm \rho_k}(\la',t')}\bigg(\frac{Z'_{\rho_0}(\la',0)}{Z'_{\pm \rho_k}(\la',0)}\bigg)^{-1},\qquad   
Z'_{\pm \rho_k}(\la',t')=e^{F'_{\pm \rho_{k}}(\la',t')}.
\end{equation}
The jumping behaviour of the normalised partition functions
can be summarised as follows. Equation \rf{Fplusjumps} immediately 
implies that across $l_k'$, we have
\[
{\CZ_{\rho_{k+1}}(\la',t')}=(1-Qq^k)^k{\CZ_{\rho_k}(\la',t')}\,.
\]
It follows that for $k\geq 0$
\[
{\CZ_{\rho_{k+1}}(\la',t')}=\prod_{j=1}^k (1-Qq^j)^j\,,
\]
where we have used that $Z'_{\rho_1}=Z'_{\rho_0}$.\\

Considering the functions $\CZ_{-\rho_k}(\la',t')$, one needs to take into account 
the fact that the jumps of the normalising factor accumulate at the imaginary 
axis. It is then straightforward to compute

\begin{align}
&\lim_{k\to \infty}\CZ_{\rho_k}(\la',t')=\prod_{k=1}^\infty (1-q^k Q)^k,\\
&\lim_{k\to -\infty}\CZ_{-\rho_{k}}(\la',t')=(M(q))^2
\prod_{k=1}^\infty (1-q^k Q)^k,\\
&\CZ_{-\rho_0}(\la',t')=(M(q))^2
\prod_{k=1}^\infty (1-q^k Q)^k(1-q^{k} Q^{-1})^{k},
\end{align}
with $M(q)=\prod_{k=1}^\infty (1-q^k)^{-k}$ being the MacMahon function. We note furthermore that $\CZ_{-\rho_0}(\la',t')$ is the expression obtained in \cite{Szendroi} as a generating function of non-commutative DT invariants.

\subsection{Relation to framed BPS states}

Our findings can be compared with the known results on counting of 
{\it framed} BPS-states, representing bound states of D0- and D2-branes
with a single infinitely heavy D6 in string theory on local CY manifolds. 
A useful characteristic of the spectrum of BPS states are the BPS indices (generalised DT invariants) 
$[\mathrm{DT}]_{n\de+k\be+{\de}^\vee}^{\CC}$ which are locally constant with respect
to the K\"ahler parameters, but may jump along walls of marginal stability in the 
K\"ahler moduli space $\CM_{\mathrm{K\ddot{a}h}}$ and therefore 
depend on the choice of a chamber $\CC\subset\CM_{\mathrm{K\ddot{a}h}}$.
The BPS partition functions are generating 
functions for the BPS indices  for  the case of the conifold defined as
\begin{equation}
\CZ_{\mathrm{BPS}}(u,v;\CC)=\sum_{k= 0}^\infty\sum_{n=1}^\infty\,
[\mathrm{DT}]_{n\de+k\be+\de^\vee}^{\CC}\,
u^n v^k.
\end{equation}

The pattern of chambers can be described as follows \cite{JM}. The 
processes associated to walls of marginal stability
represent decay or recombination of framed BPS--states
with charges $\ga_1=k'\de+m'\be+\de^\vee$ and unframed BPS-state 
 with charges $\ga_2=k\de+m\be$. By regarding the resolved conifold as a 
 limit $\Lambda\ra\infty$ of a family of compact CY manifolds having a 
complexified K\"ahler parameter $\Lambda e^{\ii\vf}$, one may introduce
a regularised central charge function $Z(\ga_1)$, 
to leading order in $\Lambda$ given by $(\Lambda e^{\ii \vf})^3$.  
Unframed BPS-states with charges $\ga_2=k\de+m\be$ have central charge function 
$Z(\ga_2)=mz-k$, where $z$ is the complexified K\"ahler parameter associated to 
the compact two-cycle of the resolved conifold. 
The phases of $Z(\ga_1)$ and $Z(\ga_2)=mz-k$ align if 
\[
3\vf=\mathrm{arg}(mz-k)+2\pi n, \qquad n\in\BZ\,.
\]
Taking into account that there only exist BPS-states with $m=\pm 1$, one arrives 
at the pattern of walls $\CW_k^m$ described in \cite{JM}, decomposing the parameter
space into a collection of chambers $\CC_k^-=[\CW_{k-1}^{-1}\CW_{k}^{-1}]$ and
$\CC_k^+=[\CW_{k}^1\CW_{k-1}^1]$, respectively.\\ 

Of special interest are the core region $\CC_0^+\cup\CC_{1}^+$, the limits
$\CC_\infty^\pm$, and the  chamber $\CC_0^-$ called non-commutative chamber
following \cite{JM}. The partition functions are
\begin{align}
&\CZ_{\mathrm{BPS}}(u,v;\CC_\infty^{+})=\prod_{k=1}^\infty (1-(-u)^k v)^k,\qquad
\CZ_{\mathrm{BPS}}(u,v;\CC_{\rm core})=1,\\
&\CZ_{\mathrm{DT}}(u,v):=\CZ_{\mathrm{BPS}}(u,v;\CC_\infty^{-})=(M(-u))^2
\prod_{k=1}^\infty (1-(-u)^k v)^k,\\
&\label{nc-chamber}\CZ_{\mathrm{BPS}}(u,v;\CC_0^{-})=(M(-u))^2
\prod_{k=1}^\infty (1-(-u)^k v)^k(1-(-u)^k v^{-1})^{k}.
\end{align}
One may identify the exponents in \rf{nc-chamber} 
with the unframed BPS indices defining the 
BPS Riemann--Hilbert problem for the conifold.\\

The GW-DT correspondence \cite{MNOP1,MNOP2,MOOP} 
relates the BPS partition function to the topological 
string partition function through the following relation\footnote{Comparing with \cite{MNOP1,MNOP2}, one should note that the variable $q$ used in these papers corresponds to the quantity $-q$ in our notations.}
\begin{equation}\label{GW-DT}
\CZ_{\mathrm{ DT}}(-q,Q)
=(M(q))^{\chi(X)}e^{F_{\rm GV}(\la,t)}, \qquad q=e^{\ii\la},\quad
Q=e^{2\pi\ii t}.
\end{equation}
Taking into account the relation between the variables $u,v$ 
and $q,Q$ following from \rf{GW-DT},
and identifying $\arg(\la')=3\vf$, $z=t'$,
we find a one-to-one correspondence between the chambers $\CC_k^\pm$ 
and the wedges $[\pm l_k',\pm l_{k-1}']$ representing 
natural domains of definition for the Borel summations $F'_{\pm \rho_k}(\la',t')$ of the strong-coupling 
expansion, together with a precise match between
the BPS partition functions $\CZ_{\mathrm{ BPS}}(u,v;\CC_k^\pm)$
and the normalised partition functions $\CZ_{\pm \rho_k}(\la',t')$ defined in \rf{norm-funct},
chamber by chamber.

\section{S-duality}\label{sec:S-dual}


It seems interesting to observe that the wall-crossing behaviour 
of the generating functions $\CZ_{\mathrm{ BPS}}(u,v;\CC)$
for BPS indices involves jumps related
to the jumps in Bridgeland's RH problem by the replacements
\begin{equation}\label{S-duality}
\la\mapsto \la_{\mathrm D}=-\frac{4\pi^2}\la,\qquad t\mapsto t_{\mathrm D}=\frac{2\pi}{\la}t.
\end{equation}

This suggests that we can use the framed wall-crossing phenomena studied in \cite{JM}
causing the jumps of the BPS partition functions 
$\CZ_{\mathrm{ BPS}}(u,v;\CC)$ to 
define a ``dual'' version of the RH problem studied by Bridgeland in \cite{BridgelandCon}.
The location of walls and the explicit formulae for the jumps 
of the dual RH problem are obtained by replacing 
$\la$ and $t$ by $\la_{\mathrm D}$ and $t_{\mathrm D}$, respectively. \\

The dependence 
on the variable $\la$ suggests that Bridgleland's RH problem describes
wall-crossing phenomena in non-perturbative 
effects due to disk instantons in string theory,  while the dual
RH problem describes the wall-crossing of BPS states in supergravity.
As an outlook we will now 
briefly indicate how weak and strong-coupling expansions can be combined to get
a more global geometric picture of the space
$\CM_{\text{K\"ah}}\times\BC^\times$ with coordinates $(t,\la)$,
outline connections 
to the S-duality conjectures in string theory, and point out 
a relation to the mathematical phenomenon called Langlands modular duality
in the context of quantum cluster algebras \cite{FG}.

\subsection{Global aspects}

In the space $\CM_{\text{K\"ah}}\times\BC^\times$ with coordinates $(t,\la)$,
one may naturally consider two asymptotic regions, referred to as weak and
strong-coupling regions, respectively. The weak coupling region is 
defined by sending $\la\ra 0$ keeping $t$ fixed, while the strong 
coupling region can be described by sending $\la\ra\infty$ with constant $t_{\mathrm D}$.
The asymptotic expansions of the non-perturbative free energy $F_{\mathrm{np}}(\la,t)$
in powers of $\la$ and $\la^{-1}$ are valid in the weak and
strong-coupling regions, respectively.\\

In order to get a more global picture, it seems natural to include the rays and jumps 
of the strong coupling expansion into the definition of a refined version of the 
line bundle discussed in the previous section \ref{hyperholo}. 
More precisely:

\begin{itemize}
    \item On the complex $2$-dimensional parameter space $\CM_{\text{K\"ah}}\times\BC^\times$ parametrized by $(t,\lambda)$ or $(t',\lambda')=(t_D,-\lambda_D/2\pi)$, one can consider the real 3-dimensional walls
\begin{equation}\label{wallsdef}
    \begin{split}
    \mathcal{W}_{\text{weak},k}^{\pm}&:=\{(t,\lambda) \; | \; \lambda \in \pm \mathbb{R}_{<0}2\pi i (t+k)\}, \quad k \in \mathbb{Z},\\
    \mathcal{W}_{\text{strong},k}^{\pm}&:=\big\{(t',\lambda') \; | \; \lambda' \in \pm \mathbb{R}_{<0}2\pi i (t'+k)\big\}, \quad k \in \mathbb{Z}-\{0\}\,\\
    &=\big\{(t,\lambda) \; | \; \lambda={\textstyle\frac{-2\pi t \pm ir}{k}},\quad r\in \mathbb{R}_{>0}\big\}, \quad k \in \mathbb{Z}-\{0\}\,.
    \end{split}
\end{equation}
and the chambers defined by the connected components of the complement of $\mathcal{W}:=\cup_{k\in \mathbb{Z}} \mathcal{W}_{\text{weak},k}^{\pm}\cup_{k\in \mathbb{Z}-\{0\}}\mathcal{W}_{\text{strong},k}^{\pm}$ (notice that since there is no jump associated to $\pm l'_0$ it is safe to exclude the case $k=0$ for the strong coupling walls).
Intersecting $\mathcal{W}_{\text{weak},k}^{\pm}$ with a $t$-slice $\{t\}\times \mathbb{C}^{\times}$, one obtains $\pm l_k$ in the corresponding $\lambda$-plane; while intersecting $\mathcal{W}_{\text{strong},k}^{\pm}$ with a $t'$-slice $\{t'=t/\check{\lambda}=\text{const}\}$, one obtains $\pm l'_k$ in the corresponding $\lambda'$-plane. Furthermore, the intersection of $\mathcal{W}_{\text{strong},k}^{\pm}$ with  $\{t\}\times \mathbb{C}^{\times}$ gives a ray starting at $-2\pi t/k$ and parallel to the imaginary axis in the corresponding $\lambda$-plane. The rays corresponding to the intersection of $\mathcal{W}_{\text{strong},k}^{+}$ and $\mathcal{W}_{\text{strong},k}^{-}$ with  $\{t\}\times \mathbb{C}^{\times}$ combine into a line parallel to the imaginary axis, and missing the point $-2\pi t/k$. In particular, these lines accumulate near the imaginary axis of the $\lambda$-plane $\{t\}\times \mathbb{C}^{\times}$. Assuming $\text{Re}(t)>0$, the ones to the right of the imaginary axis correspond to $k<0$, while the ones on the left correspond to $k>0$.
\item Taking into account the chamber structure on $(\CM_{\text{K\"ah}}\times\BC^\times)-\mathcal{W}$, the corresponding refined line bundle would then have transition functions along the walls determined by the jumps 
obtained at strong and weak coupling. In particular, there is in $(\CM_{\text{K\"ah}}\times\BC^\times)-\mathcal{W}$ a distinguished chamber $\mathcal{D}$ determined by the constraint $0<\mathrm{Re}(t)<1$, $\text{Im}(t)>0$, and the condition that $\mathcal{D}\cap (\{t\}\times \mathbb{C}^{\times})$ gives the region of the Stokes sector $[l_0,l_{-1}]$ to the right of the line $\mathcal{W}_{\text{strong},-1}^{\pm}\cap (\{t\}\times \mathbb{C}^{\times}) $. On this region $F_{\text{np}}(\lambda,t)$ is defined and matches $F_{\mathbb{R}_{>0}}(\lambda,t)$.  We can then use the jumps along the walls to extend $F_{\text{np}}$ to the other regions. In particular, if we fix $t$ and we cross the infinite set of weak coupling walls $\mathcal{W}_{\text{weak},k}^{+}$ for $k>0$ (or $k<0$) while avoiding the strong coupling walls, one is left with $F_{\text{GV}}$; while if we cross the infinite set of strong coupling walls $\mathcal{W}_{\text{strong},k}^{+}$ for $k<0$ while avoiding the weak coupling walls (i.e. while remaining in the sector $[l_0,l_{-1}]$), we are left with $F_{\text{NS}}$.  
\end{itemize} 

The original and dual RH problems have 
jumps arranged according to peacock patterns in the product of two complex planes
with coordinates $(\la,t)$ and $(\la_{\mathrm D},t_{\mathrm D})$, 
respectively. 
Assuming that $0<\text{Re}(t)<1$ and $\text{Im}(t)>0$,
one finds that the 
positive and negative real half-axes are distinguished by the property
of being self-dual in the sense that they are contained both in 
$\CC_{0}^\pm$ and in the wedges between $\pm l_0$ and $\pm l_{-1}$.
The self-duality of the intersection of these chambers
strengthens the sense in which $F_{\mathrm{np}}(\la,t)$ is distinguished 
as a non-perturbative definition of the topological string partition function.

\subsection{Relation to string-theoretic S-duality}

Relation \rf{S-duality}
 resembles the realisation of S-duality discussed in \cite{APSV,APP10} 
 on complex Darboux coordinates for
the QK manifolds representing the hypermultiplet moduli spaces of 
type II string theory (see \cite{Al11,AMPP} for reviews).
This is probably no accident. \\ 

One may in particular notice that a Riemann--Hilbert problem similar to 
the one studied in \cite{BridgelandCon} is expected to be solved by twistor coordinates
for the hypermultiplet moduli space in type II
string theory on the resolved conifold.  
This Riemann--Hilbert problem should reproduce the
problem studied  in \cite{BridgelandCon} in  a limit called the conformal limit. 
Both Riemann--Hilbert problems are defined with the help of the same 
BPS structure, implying that the symplectic transformations used in 
the definitions 
coincide. The main differences 
will  concern the asymptotic conditions imposed in the formulation of the
two problems. 
These considerations suggest that the complex structures on 
$M\times\BC^\times$ defined by the coordinate functions solving 
the RH problem from \cite{BridgelandCon}
are limits  of the complex structures on the conifold 
hypermultiplet moduli space defined by twistor coordinates.\\

The QK metrics defined by mutually local D-instanton corrections have been
studied intensively already \cite{RRSTV,AS09,AB15, CT}. 
Infinite-distance limits of such QK-metrics have been studied in 
\cite{BMW} motivated by the swampland conjectures in type II string theories. 
Two infinite-distance limits play a basic role. The first, called the D1 limit
in \cite{BMW}, is characterised by  large volume and large coupling $g_\mathrm s=1/\tau_2$. 
The second is called the F1 limit. It is simply described by 
small coupling $g_\mathrm s$ at finite values of the K\"ahler moduli.
The two limits are related by S-duality. This implies that
the D1 limit is characterised by a scaling of the form
\[
\tau_2(\si)=e^{-\frac{3}{2}\si}\tau_2(0)\,,\qquad
t(\si)=e^{\frac{3}{2}\si}t(0)\,,
\]
taking into account the leading quantum 
corrections to the QK metric in this limit, as expressed most clearly
in \cite[Equation (3.41)]{BMW}.\\

It is known that a scaling of $g_\mathrm s$ induces the same scaling of the
topological string coupling $\la$
in the conformal limit.  This relates the F1 and D1 limits 
to the weak- and strong-coupling regions in the 
space $M\times\BC^{\times}$, respectively. 
As the F1 and D1 limits are exchanged by S-duality,
it seems natural to conjecture that the relations between 
the Stokes jumps of weak- and strong-coupling expansions 
observed above are related to the S-duality phenomenon
by the conformal limit. \\

It has been argued in \cite{APP}, see also \cite{AMPP} for a review, that the 
string theoretical S-duality conjectures relating D5 and NS5 branes predict relations
between BPS partition functions and NS5-brane partition functions.
As discussed in \cite{APP,AMPP}, the NS5-brane partition function lives precisely
in the line bundle governed by Rogers dilogarithm  discussed in 
Section \ref{hyperholo}.\footnote{This was pointed out to us by S. Alexandrov.}

\subsection{Langlands modular duality}

It seems finally worth pointing out that the coordinate changes
associated to Stokes jumps in weak- and strong coupling
expansions are related by the phenomenon
called Langlands modular duality in the terminology introduced by 
Fock and Goncharov in the 
context of quantum cluster algebras \cite{FG} following \cite{Fad}. 
An essential aspect of this phenomenon, specialized to the case at hand, is the 
possibility to introduce dual shift operators 
\[
({T}f)(t)=f(t+\la/2\pi),\qquad (\tilde{T}f)(t)=f(t+1), 
\]
which act on the variables $\tilde{Q}:=e^{4\pi^2\ii\,t/\la}\equiv w$ and $Q=e^{2\pi\ii t}$
as
\[
{T}Q=qQ,
\qquad \tilde{T}\tilde{Q}=\tilde{q}\tilde{Q}, \qquad\tilde{T}Q=Q, \qquad {T}\tilde{Q}=\tilde{Q}.
\]
This implies in particular that the functions representing
the cluster coordinate transformations associated to the weak coupling jumps
are invariant under the shift $\tilde{T}$, 
while the shift ${T}$ acts trivially on 
the cluster coordinate transformations associated to the strong coupling jumps.
This simple phenomenon has a natural generalisation 
which is the root of some remarkable features of quantized cluster 
algebras \cite{FG}.
We can't help the feeling that
this manifestation of Langlands modular duality in the case of the 
resolved conifold partition functions can be the tip of an iceberg.

\clearpage
\appendix
\begin{section}{Alternative proof for the Borel sum}\label{altproofBorel}

In this section we present an alternative derivation of the Borel sum and transform of $\widetilde{F}(\lambda,t)$. The alternative proof uses the integral representation of the Hadamard product used in Section \ref{boreltransproof}.

\begin{prop}  Take $t \in \mathbb{C}$ with $\mathrm{Im}(t)>0$ and $0<\mathrm{Re}(t)<1$. Then
$F_{\mathrm{np}}(\lambda,t)$ equals $F_{\mathbb{R}_{>0}}$ on their common domain of definition. More specifically:
\begin{equation}\label{eq:Fnp-G}
    F_{\mathrm{np}}(\lambda,t) =  \frac{1}{\lambda^2} \mathrm{Li}_{3}(Q)+ \frac{B_2}{2}\mathrm{Li}_{1}(Q)+\int_{0}^\infty\mathrm d\xi\, e^{-\xi/\check{\lambda}} G(\xi,t)\,. 
\end{equation}
\end{prop}
\begin{proof}

We will first write down an integral representation for $F_{\mathbb{R}_{>0}}$ assuming that $t\in (0,1)$, and $\check{\lambda}>0$ satisfies the conditions of Proposition \ref{intrep2}. We will then deform $t$ to $\text{Im}(t)>0$ and show what we want.\\

We recall the Hadamard product representation (see Proposition \ref{Borelhadprod}):

\begin{equation}
    G(\xi,t)=\frac{1}{2\pi i}\int_{\gamma}\frac{\mathrm ds}{s}f_1(s)\,f_2\Big(\frac{\xi}{s},t\Big)\,,
\end{equation}
where $\gamma$ was an appropriate counterclockwise contour around $0$, and 
\begin{equation}
    \begin{split}
        f_1(s)&= -\frac{1}{4\pi^2}\left(  \frac{1}{\xi^3}- \frac{1}{\xi (e^{\xi/2}-e^{-\xi/2})^2}-\frac{1}{12\xi}\right),\\
     f_2(\xi,t)&=\frac{(2\pi \I)^3}{2}\left(\mathrm{Li}_0\big(e^{2\pi \I(t + \xi)}\big)-\mathrm{Li}_0\big(e^{2\pi \I(t - \xi)}\big)\right).\\
    \end{split}
\end{equation}
Integrating $e^{-\xi/\check \lambda} G(\xi,t)$ along  the positive real line and swapping the integral signs, we get
\begin{equation}
\begin{split}
    &\int_{0}^{\infty} \mathrm d\xi\,e^{-\xi/{\check \lambda}} G(\xi,t) \\
    &= \frac{(2\pi \I)^2}{2} \int_{\gamma}\frac{\mathrm ds}{s} \left(\int_{0}^{\infty}\mathrm d\xi\, f_1(s)e^{-\xi/{\check \lambda}}\mathrm{Li}_0\big(e^{2\pi i(t+\xi/s)}\big) - f_1(s)e^{-\xi/{\check \lambda}}\mathrm{Li}_0\big(e^{2\pi i(t-\xi/s)}\big)\right).
\end{split}
\end{equation}
Next we simultaneously rescale $s\mapsto{\check \lambda}s$ and $\xi \mapsto {\check \lambda}s\xi$ on the first term, while simultaneously rescaling $s\mapsto - {\check \lambda}s$ and $\xi \mapsto {\check \lambda}s\xi$ on the second term to obtain
\begin{align*}
    \int_{0}^{\infty} \mathrm d\xi \,e^{-\xi/{\check \lambda}} G(\xi,t)&= \frac{(2\pi \I)^2}{2} \int_{\gamma} \frac{\mathrm ds}{s} \left(\int_{0}^{s^{-1}\infty}\mathrm d\xi\,{\check \lambda} s \big(f_1({\check \lambda} s)-f_1(-{\check \lambda} s)\big)e^{-s\xi}\mathrm{Li}_0\big(e^{2\pi \I(t+\xi)}\big)\right) \\
    &=(2\pi \I)^2\int_{\gamma}\mathrm ds\,  \check \lambda   f_1({\check \lambda}s)\left(\int_{0}^{s^{-1}\infty}\mathrm d\xi\, e^{-s\xi}\mathrm{Li}_0\big(e^{2\pi \I(t+\xi)}\big)\right). \numberthis
\end{align*}
Let $\mathcal C$ and $\mathcal C'$ denote the contours  following the real line  from $-\infty $ to $\infty  $ avoiding $0$ by a small detour in the upper and lower half-planes respectively. We may in fact take them to the lines with imaginary parts $\epsilon $ and $-\epsilon$ respectively, for some small $\epsilon > 0$. Since $\mathcal C' - \mathcal C = \gamma$ up to homology, we can  write
\begin{align*}
     \int_{0}^{\infty} \mathrm d\xi\,e^{-\xi/{\check \lambda}} G(\xi,t) &=(2\pi \I)^2\left(\int_{\mathcal C'}-\int_{\mathcal C} \right)\mathrm ds\, {\check \lambda}  f_1({\check \lambda}s)\left(\int_{0}^{s^{-1}\infty}\mathrm d\xi\, e^{-s\xi}\mathrm{Li}_0\big(e^{2\pi \I(t+\xi)}\big)\right) \\
    &=(2\pi \I)^2\int_{-\infty}^\infty\mathrm ds\Bigg({\check \lambda}  f_1({\check \lambda}(s-\I\epsilon))\left(\int_0^{(s+\I\epsilon)\infty}\mathrm d\xi\, e^{-(s-\I\epsilon)\xi}\mathrm{Li}_0\big(e^{2\pi \I(t+\xi)}\big)\right)\\
    &\qquad -{\check \lambda}  f_1({\check \lambda}(s+\I\epsilon))\left(\int_0^{(s-\I\epsilon)\infty}\mathrm d\xi\, e^{-(s+\I\epsilon)\xi}\mathrm{Li}_0\big(e^{2\pi \I(t+\xi)}\big)\right) \Bigg). \numberthis
\end{align*}

Now taking the limit $\epsilon \rightarrow 0^+$ then gives us
\begin{equation}
\int_{0}^{\infty} \mathrm d\xi\,e^{-\xi/{\check \lambda}} G(\xi,t) =(2\pi \I)^2\int_{-\infty}^\infty\mathrm ds\,{\check \lambda}  f_1({\check \lambda}s)\left(\int_{\mathcal H_s}\mathrm d\xi\, e^{-s\xi}\mathrm{Li}_0\big(e^{2\pi \I(t+\xi)}\big)\right),
\end{equation}
where $\mathcal H_s$ is a counterclockwise Hankel contour along the negative real axis when $s<0$, a clockwise Hankel contour along the positive real axis when $s>0$, and the imaginary axis from $-\I\infty$ to $\I\infty$ when $s=0$.\\

The poles and residues of the inner integrand are given by
\begin{equation}
    \begin{split}
        \mathrm{Res}_{-(t+k)} \big(e^{-s\xi}\mathrm{Li}_0(e^{2\pi \I(t+\xi)})\big) =-\frac{1}{2\pi \I}e^{s(t+k)}, 
    \end{split}
\end{equation}
for all $k\in \Z$. We can thus deduce using Cauchy's residue theorem that the inner integral is the sum of $2\pi \I$ times the residues from the poles at $-(t+k)$ with $k \ge 0$ when $s<0$ and minus the sum of $2\pi \I$ times the residues from the poles at $-(t+k)$ with $k > 0$ when $s>0$: 
\begin{align*}
   \int_{\mathcal H_s} \mathrm d\xi\,e^{-s\xi}\mathrm{Li}_0\big(e^{2\pi \I(t+\xi)}\big) &=  -\sum_{k=0}^{\infty}e^{s(t+k)}=-\frac{e^{st}}{1-e^s}\,, \qquad \mbox{when $s<0$},\\
  \int_{\mathcal H_s} \mathrm d\xi\,e^{-s\xi}\mathrm{Li}_0\big(e^{2\pi \I(t+\xi)}\big)&=\sum_{k=-1}^{-\infty}e^{s(t+k)}= \frac{e^{st}e^{-s}}{1-e^{-s}}=-\frac{e^{st}}{1-e^s}\,, \qquad \mbox{when $s>0$}. \numberthis
\end{align*}

Putting everything together, we get
\begin{align*}\label{intrep}
    \int_{0}^\infty \mathrm d\xi\,e^{-\xi/\check{\lambda}}  G(\xi,t)
    &= -\int_{-\infty}^\infty \frac{\mathrm ds}{s}\left(\frac{e^{\check \lambda s}}{(e^{\check \lambda s} -1)^2}-\frac{1}{({\check \lambda s})^2}+\frac{1}{12}\right)\frac{e^{st}}{e^s -1} \\
    &=-\int_{\mathcal C} \frac{\mathrm ds}{s}\left(\frac{e^{\check \lambda s}}{(e^{\check \lambda s} -1)^2}-\frac{1}{({\check \lambda s})^2}+\frac{1}{12}\right)\frac{e^{st}}{e^s -1}\,, \numberthis
\end{align*}
where we remark that the integrand of the integral over $\mathbb{R}$ is actually regular at $s=0$.\\

Both expressions in the equality (\ref{intrep}) above are analytic in $t$ and $\check{\lambda}$, so we can deform $t$ to $\text{Im}(t)>0$ with $\text{Im}(t)$ small, and $\lambda$ away from $\mathbb{R}_{>0}$, so that (\ref{intrep}) continues to hold in their common domain of definition. \\

The result to be proved will follow if we can show the following for $m=0$ and $m=1$:
\begin{equation}
\begin{split}
    \int_{\mathcal{C}} \frac{\mathrm ds}{s^{2m+1}}\frac{e^{ts}}{e^s -1}  &= \frac{1}{(2\pi \I)^{2m}} \mathrm{Li}_{2m+1}(e^{2\pi \I t}).
    \end{split}
\end{equation}
The integrand has a  pole of order $2m+2$ at $s=0$ and simple poles at $s=2\pi \I k$ with residues $e^{2\pi \I kt}/(2\pi \I k)^{2m+1}$ for all nonzero integers $k$. The contour $\mathcal C$ contains  only the simple poles with $k >0$. Thus, again by Cauchy's residue theorem, we have:

\begin{equation}
    \int_{\mathcal{C}} \frac{\mathrm ds}{s^{2m+1}}\frac{e^{ts}}{e^s -1}  = 2\pi \I \sum_{k=1}^\infty \frac{e^{2\pi \I kt}}{(2\pi \I k)^{2m+1}} = \frac{1}{(2\pi \I)^{2m}} \mathrm{Li}_{2m+1}(e^{2\pi \I t})\,.
\end{equation}
Where in the last equality we used the fact that $\text{Im}(t)>0$ and hence $|e^{2\pi \I t}|<1$, so that the series representation of $\mathrm{Li}_s(z)$ holds.
This completes the proof.
\end{proof}

\end{section}


\section{Asymptotic series from Borel transforms}\label{App:TaylorBorel}
\begin{lem}\label{lemmaApp}
The expression
\begin{equation}
     G(\xi,t)= -\frac{1}{4\pi^2}\sum_{g=2}^{\infty} \frac{ B_{2g}}{2g (2g-2)! (2g-3)!} \xi^{2g-3}\,\partial_t^{2g} \mathrm{Li}_3(Q)\,,
\end{equation}
of the Borel transform can be obtained back from
\begin{equation}
    \begin{split}
        G(\xi,t) &= -\sum_{m\in\Z \setminus \{0\}}\frac{1}{(2\pi \I)^2}\left(\frac{1}{m^3}\left(\frac{e^{2\pi \I t + \xi/m}}{1-e^{2\pi \I t + \xi/m}}-\frac{e^{2\pi \I t - \xi/m}}{1-e^{2\pi \I t - \xi/m}}\right)\right.\\ &\qquad\qquad\qquad+\left.\frac{\xi}{2m^4}\left(\frac{e^{2\pi \I t + \xi/m}}{(1-e^{2\pi \I t + \xi/m})^2}+\frac{e^{2\pi \I t - \xi/m}}{(1-e^{2\pi \I t - \xi/m})^2}\right)\right).
    \end{split}
\end{equation}

\end{lem}
\begin{proof}
We first write the second expression of $G(\xi,t)$ as
\begin{equation}
    \begin{split}
        G(\xi,t) &= -\frac{1}{\xi} \frac{\partial}{\partial \xi} \left( \frac{\xi^2}{(2\pi \I)^2}\sum_{m=1}^{\infty} \left(\frac{1}{m^3}\left(\frac{e^{2\pi \I t + \xi/m}}{1-e^{2\pi \I t + \xi/m}}-\frac{e^{2\pi \I t - \xi/m}}{1-e^{2\pi \I t - \xi/m}}\right)\right.\right)
    \end{split}
\end{equation}
we next use the Taylor expansion around $\xi=0$:
\begin{equation}
    \frac{e^{2\pi \I t + \xi/m}}{1-e^{2\pi \I t + \xi/m}}= \textrm{Li}_0(e^{2\pi \I t + \xi/m})= \sum_{k=0}^{\infty} \frac{\xi^k}{m^k} \textrm{Li}_{-k}(e^{2\pi \I t})\,,
\end{equation}
which makes use of the property
\begin{equation} 
\theta_Q \textrm{Li}_s(Q) =\textrm{Li}_{s-1} (Q)\,,  \quad \theta_Q:= Q \,\frac{\mathrm d}{\mathrm dQ}\,,
\end{equation}
We thus obtain:
\begin{align*}
    G(\xi,t) &= -\frac{1}{\xi} \frac{\partial}{\partial \xi} \left( \frac{\xi^2}{(2\pi \I)^2}\sum_{m=1}^{\infty} \left(\frac{2}{m^3}\left( \sum_{k=0}^{\infty} \frac{1}{(2k+1)!} \left(\frac{\xi}{m}\right)^{2k+1} \textrm{Li}_{-2k-1}(e^{2\pi \I t})\right)\right)\right)\, \\
        &=-\frac{2}{\xi} \frac{\partial}{\partial \xi} \left( \frac{\xi^2}{(2\pi \I)^2}\left( \sum_{k=0}^{\infty} \zeta(2k+4) \frac{1}{(2k+1)!} \xi^{2k+1} \textrm{Li}_{-2k-1}(e^{2\pi \I t})\right)\right)\\
        &=-\frac{2}{\xi} \frac{\partial}{\partial \xi} \left( \frac{\xi^2}{(2\pi \I)^2}\left( \sum_{k=0}^{\infty} (-1)^{k+3}\frac{B_{2k+4}\, (2\pi)^{2k+4} }{2(2k+4)!} \frac{1}{(2k+1)!} \xi^{2k+1} \textrm{Li}_{-2k-1}(e^{2\pi \I t})\right)\right)\,\\
        &= -\left( \frac{1}{(2\pi \I)^2} \sum_{k=0}^{\infty} (-1)^{k+3} \frac{B_{2k+4}\, (2\pi)^{2k+4}\, (2k+3)}{(2k+4)! (2k+1)!} \xi^{2k+1} \, \textrm{Li}_{-2k-1}(e^{2\pi \I t})\right)\\
        &= -\frac{1}{4\pi^2}\sum_{g=2}^{\infty} \frac{ B_{2g}}{2g (2g-2)! (2g-3)!} \xi^{2g-3}\,\partial_t^{2g} \textrm{Li}_3(Q)\,, \numberthis
\end{align*}
        
where in going from the first to the second line we have used the following expression for the Riemann zeta function:
$$ \zeta(s)= \sum_{n=1}^{\infty}\frac{1}{n^s}\,, \quad \textrm{Re}(s)>0\,,$$
and in going from the second to the third line we have used the following identity:
$$ \zeta(2n)= \frac{(-1)^{n+1} B_{2n} (2\pi)^{2n}}{2 (2n)!}\,,$$
and where we have changed the summation variable in the fifth line to $g=k+2$ and made use of
$$ (-1)^g \, (2\pi)^{2g} \textrm{Li}_{3-2g}(Q) = \partial_t^{2g} \textrm{Li}_{3}(Q)\,.$$
\end{proof}

\begin{lem}\label{constboreltrans}
The expression $-G(\xi,0)-\frac{1}{12\xi}$ gives the Borel transform of $F_0(\la) + \zeta(3)/\lambda^2-F_0^1$.
\end{lem}

\begin{proof}
From the definition of $F_0(\lambda)$ in (\ref{constmapterms}), we have the following (recall that we take $\chi(X)=2$ for the resolved conifold):
\begin{equation}
    F_0(\la) + \zeta(3)/\lambda^2-F_0^1=\sum_{g\geq 2}\lambda^{2g-2}\frac{(-1)^{g-1}\, B_{2g}\, B_{2g-2}}{2g (2g-2)\, (2g-2)!}\,.
\end{equation}
The Borel transform  $G_0(\xi)$ of the previous series is then given by 
\begin{equation}
    G_0(\xi)=\sum_{g\geq 2}\xi^{2g-3}\frac{(-1)^{g-1}\, B_{2g}\, B_{2g-2}(2\pi)^{2g-2}}{2g\, ((2g-2)!)^2}\,.
\end{equation}
On the other hand, we have
\begin{equation}
    \begin{split}
    G(\xi,0)&=\frac{2}{(2\pi)^2}\sum_{m>0}\frac{1}{m^3}\Big(1+\frac{\xi}{2}\frac{\partial}{\partial \xi}\Big)\Big(\frac{1}{1-e^{\xi/m}}-\frac{1}{1-e^{-\xi/m}}\Big)\\
    &=\frac{2}{(2\pi)^2}\sum_{m>0}\frac{1}{m^3}\Big[\frac{1}{2}\Big(\frac{1}{1-e^{\xi/m}}-\frac{1}{1-e^{-\xi/m}}\Big)+\frac{1}{2}\frac{\partial}{\partial \xi}\Big(\frac{\xi}{1-e^{\xi/m}}-\frac{\xi}{1-e^{-\xi/m}}\Big)\Big]\\
    \end{split}
\end{equation}
Using (\ref{polylogbern}) one finds 
    \begin{align*}
        G(\xi,0)&=\frac{2}{(2\pi)^2}\sum_{m>0}\frac{1}{m^3}\Big[-\frac{m}{\xi}\sum_{k=0}^{\infty}\frac{B_{2k}}{(2k)!}\Big(\frac{\xi}{m}\Big)^{2k}-m\frac{\partial}{\partial \xi}\Big(\sum_{k=0}^{\infty}\frac{B_{2k}}{(2k)!}\Big(\frac{\xi}{m}\Big)^{2k}\Big)\Big]\\
        &=-\frac{2}{(2\pi)^2}\sum_{k=0}^{\infty}(2k+1)\frac{B_{2k}}{(2k)!}\xi^{2k-1}\zeta(2k+2)\\
        &=-\frac{2}{(2\pi)^2}\sum_{k=0}^{\infty}(2k+1)\frac{B_{2k}}{(2k)!}\xi^{2k-1}\Big(\frac{(-1)^kB_{2k+2}(2\pi)^{2k+2}}{(2k+2)!2}\Big)\\
        &=-\sum_{g\geq 1}\xi^{2g-3}\frac{(-1)^{g-1}\, B_{2g}\, B_{2g-2}(2\pi)^{2g-2}}{2g\, ((2g-2)!)^2}\\
        &=-G_0(\xi) -\frac{1}{12\xi}\numberthis
    \end{align*}
and the result follows.
\end{proof}

\newcommand{\etalchar}[1]{$^{#1}$}

\end{document}